\documentclass{article}
\usepackage[letterpaper, left=1in, right=1in, top=1in,bottom=1in]{geometry}
\usepackage{graphicx} 
\usepackage{amssymb}
\usepackage{amsthm}
\usepackage{mathrsfs}
\usepackage{amsmath}
\usepackage{float}
\usepackage{color,xcolor}
\usepackage[citecolor=blue]{hyperref}
\usepackage{booktabs}
\usepackage{longtable}
\usepackage{comment}
\usepackage{tikz}
\usetikzlibrary{calc}
\usetikzlibrary{positioning}
\usetikzlibrary{decorations.pathreplacing}
\usepackage{diagbox} 
\usepackage{multirow}
\usepackage{array}  
\usepackage{subcaption}
\usepackage{titletoc}
\usepackage{microtype}
\usepackage{natbib}

\newtheorem{theorem}{Theorem}
\newtheorem{lemma}{Lemma}[section]
\newtheorem{app_theorem}[lemma]{Theorem}
\newtheorem{proposition}[lemma]{Proposition}
\newtheorem{corollary}[lemma]{Corollary}
\newtheorem{remark}[lemma]{Remark}

\newtheorem{definition}[lemma]{Definition}

\newtheorem{condition}[lemma]{Condition}
\newtheorem{example}[lemma]{Example}

\usepackage[T1]{fontenc}
\usepackage{etoolbox}

\makeatletter

\patchcmd{\@maketitle}{\null}{}{}{} 
\renewcommand\section{%
  \@startsection{section}{1}
                {\z@}%
                {-3.5ex \@plus -1ex \@minus -.2ex}%
                {2.3ex \@plus.2ex}%
                {\normalfont\fontsize{14}{13.6}\bfseries}
}
\renewcommand\subsection{%
  \@startsection{subsection}{2}
                {\z@}%
                {-3.25ex\@plus -1ex \@minus -.2ex}%
                {2.3ex \@plus.2ex}
                {\normalfont\fontsize{12}{13.6}\bfseries}
}
\renewcommand\subsubsection{%
  \@startsection{subsubsection}{3}
                {\z@}%
                {-3.25ex\@plus -1ex \@minus -.2ex}%
                {1sp}
                {\normalfont\fontsize{12}{13.6}\bfseries}
}
\makeatother

\title{Human-AI Productivity Paradoxes: \\ Modeling the Interplay of Skill, Effort, and AI Assistance}

\author{Ali Aouad\thanks{Massachusetts Institute of Technology, maouad@mit.edu} \and Thodoris Lykouris\thanks{Massachusetts Institute of Technology, lykouris@mit.edu} \and Huiying Zhong\thanks{Massachusetts Institute of Technology, zhong826@mit.edu}\\}

\date{First version: May 2026\\Current version: October 2026\footnote{An extended abstract appeared at the Twenty-Seventh ACM Conference on Economics and Computation (EC 2026).}}

\begin{document}

\maketitle
\allowdisplaybreaks

\begin{abstract}
Generative Artificial Intelligence (AI) tools are rapidly adopted in the workplace and in education, yet the empirical evidence on AI's impact remains mixed across different operating contexts. We propose a model of human-AI interaction to better understand and analyze several mechanisms by which AI affects productivity. In our setting, human agents with varying skill levels exert utility-maximizing effort to produce certain task outcomes with AI assistance. We find that incorporating endogeneity in skill development or AI unreliability can induce a {\em productivity paradox:} increased levels of AI assistance may degrade productivity, leading to potentially significant shortfalls. We construct interpretable measures that characterize when the productivity paradox occurs. Moreover, we examine the long-term distributional effect of AI on skill, and demonstrate that \emph{skill polarization} can emerge in steady state. We identify an effort-skill non-monotonicity condition under which skill polarization emerges.

\end{abstract}

\section{Introduction}\label{sec:introduction}
Large Language Models (LLMs) and Generative AI (GenAI) tools promise to enhance human capabilities but also pose salient risks. On the one hand, these tools are increasingly used to assist humans in cognitive and creative tasks. The ``transformative AI'' hypothesis underscores the potential of AI technologies to boost productivity in the workplace \citep{brynjolfsson2025generative,cui2024effects,noy2023experimental} and in critical sectors such as education \citep{henkel2024effective, kestin_ai_2025,nie2025gpt}.
On the other hand, recent literature flags several limitations. 

First, GenAI tools exhibit hallucination and overconfidence \citep{ji2023survey}, and humans may overly rely on AI-assisted outputs or exert insufficient oversight, resulting in degraded output quality \citep{niederhoffer2025ai}. Due to such failures, workers assisted by AI tools can exhibit a decline in productivity compared to those without AI assistance \citep{dell2023navigating, snyder2026algorithm}.
Second, passive use of AI assistance may hinder important cognitive processes in human learning and self-improvement~\citep{kosmyna2025your}. A false sense of having mastered a task without exerting adequate effort or a lack of engagement in problem-solving may lead to long-term knowledge shortfall, as seen in several educational contexts \citep{abbas2024harmful,poulidis2025self}. A recent report  \citep{shen2026ai} illustrates the tension between productivity, skill, and effort in the presence of AI assistance in a software development setting.
\begin{example} \label{example1} In a randomized controlled trial~\citep{shen2026ai} where developers learned a new Python library with or without AI assistance, those using AI showed no statistically significant productivity gains yet scored 17\% lower on follow-up evaluations of conceptual understanding, code reading, and debugging. Interaction patterns involving higher cognitive effort---such as asking conceptual questions or manually adapting AI suggestions---preserved skill formation, while passive delegation did not.
\end{example}

{More broadly, current empirical evaluations of GenAI-powered systems show mixed evidence across different operating contexts. Interpreting this evidence, however, requires accounting for three key moderating factors that are often difficult to isolate empirically.
 
First, the quality of \emph{AI assistance} is uneven across task domains and complexity---this is known as the ``jagged frontier'' \citep{dell2023navigating}. In customer service, productivity can increase or decrease relative to a no-AI baseline depending on whether the AI system is sufficiently proficient for the query at hand \citep{snyder2026algorithm}.

Second, empirical evidence offers contrasting findings in terms of \emph{skill} heterogeneity.  \cite{brynjolfsson2025generative,cui2024effects,noy2023experimental} find that low-skill workers benefit more from AI assistance. In contrast to this leveling effect, \cite{eames2026computer,roldan2024genai} suggest that high-skill agents benefit more owing to their greater ability to leverage AI-generated information---a notion we refer to as AI literacy.

Third, productivity outcomes are often heterogeneous depending on how individuals engage with AI tools and the modalities of their interactions. Results vary based on whether individual usage is merely passive or involves critical thinking~\citep{lehmann2024ai,snyder2026algorithm}---resulting in different {\em effort} levels. The cost of effort also depends on the availability of resources and frictions. For instance, GenAI tutoring in education may be highly beneficial to resource-constrained environments where technology and human tutors are less accessible \citep{henkel2024effective}, but it may harm academic growth in other settings \citep{Bastani2025doi:10.1073/pnas.2422633122}.

These production factors---{\em AI assistance}, {\em human skills}, and {\em effort}---have intricate connections.  Access to AI assistance affects human effort on specific tasks, which in turn affects long-term skill development; moreover, endogenous effort decisions are likely moderated by the perceived usefulness of the AI system at such tasks. 
This raises fundamental questions: 
\begin{center}
\emph{When is AI assistance beneficial instead of harmful to long-term human productivity and skill? \\What mechanisms may plausibly explain the contrasting empirical findings to date?}
\end{center}

\subsection{Our contributions}
This paper formally studies the impact of AI assistance on productivity and skill in hybrid human-AI operating contexts. 
In Section~\ref{ssec:basic}, we develop a model for human-AI productivity, where human agents with heterogeneous {\em skill levels} exert utility-maximizing {\em effort} to produce certain task outcomes. The access to AI technology for the task at hand is captured by an {\em AI assistance level}.  We assume that skill, effort, and AI assistance jointly determine the total input to production, with effort entering additively and  AI assistance interacting with human skill in a flexible manner. On the other hand, exerting effort incurs additive costs. Within this basic model, we show that a deterministic level of AI assistance acts as a substitute for effort and skill, but results in higher productivity, aligning with findings from related modeling literature \citep{Agrawal2025NBERw34034}. 

The crux of our study is to extend this model (Section~\ref{sec:model_skill_endo}) by introducing plausible mechanisms---skill development and AI unreliability---that endogenously modify the human-AI interaction modality. Each of these mechanisms gives rise to a \emph{productivity paradox} where a greater level of AI assistance may harm productivity. When both mechanisms operate jointly, and human agents have endogenous AI literacy, {\em skill polarization} can emerge, where the long-term skill distribution becomes multimodal. Our work identifies under which conditions these phenomena arise:

\begin{itemize}
\item \textbf{The impact of skill development (Section~\ref{sec:bdc}).} We introduce a birth-death chain model of skill development, which captures how effort shapes future skills. We analyze the steady-state productivity and find that increasing the level of AI assistance can potentially decrease productivity (Theorem~\ref{thm:bdc_productivity}). The productivity shortfall can be arbitrarily large (Proposition~\ref{prop:bcd_bad_epsilon_general}). We characterize when this productivity paradox occurs---the outcome depends on a notion that we refer to as {\em semi-elasticity gap}, measuring the relative influence of AI assistance on skill and productivity. The counterintuitive decline in productivity is analogous to Simpson's paradox: AI assistance is beneficial to productivity for a fixed skill level but it also modifies the skill distribution across the population.
\item \textbf{The impact of AI unreliability (Section~\ref{sec:AI_unreliability}).} We introduce a model of AI unreliability, in which AI assistance can either produce proficient outputs or fail to do so. We analyze the expected productivity and show counterintuitively that increasing AI proficiency, while holding AI reliability constant, can decrease productivity (Theorem~\ref{thm:unreliable_concave3}). Moreover, the productivity shortfall can be arbitrarily large (Proposition~\ref{prop:exante_linear_bad}). We identify when this productivity paradox occurs---the outcome depends on a so-called {\em absolute risk aversion} notion, which measures the rate of diminishing returns of the production function. The productivity decline happens under accelerating diminishing returns, in which the indirect cannibalization of effort by AI assistance outweighs the direct productivity gain from the AI boost.
\item \textbf{The impact of AI literacy (Section~\ref{sec:AI_literacy}).} Finally, combining the above two mechanisms, we add heterogeneous AI literacy, which considers how individuals react to unreliable AI outputs. We focus on a specific production function and analyze the long-term effect of AI assistance on the steady-state skill distribution. We show how a multimodal skill distribution can emerge (Theorem~\ref{thm:bayes_skill_multimodal}), due to the reversal of effort dependence on skill. 
This highlights a polarizing effect in the population analogous to the Matthew effect, where individuals with early advantages exert more effort and thus amplify their skills. 
\end{itemize}

Our theoretical modeling approach complements existing empirical evidence on AI's impact in three ways. 
First, steady-state analysis allows us to understand long-run effects that short-term empirical studies may not capture. Second, our model clarifies how these effects depend on the operating context, yielding simple measures that characterize the productivity effects of AI assistance. Third, our model-based attribution framework allows us to isolate and quantify the effect of different channels (deskilling, AI unreliability, etc), which may be difficult to disentangle empirically. Most of our findings are robust to convex costs, forward-looking agents, and a welfare objective (see Section~\ref{sec:conclusion}).

\paragraph{Implications.}
We draw three main insights from this analysis. First, in knowledge-critical sectors, AI deployment should be accompanied by knowledge management programs: firms or educational institutions anticipating critical knowledge shortfalls should provide remedial training. Principals who design AI systems can also introduce usage frictions to incentivize human agents to exert effort and oversee AI outputs, helping them maintain or acquire adequate skills. Second, returns to AI proficiency (i.e., average output quality) can be negative for productivity when AI is unreliable (i.e., when output quality is variable). This suggests introducing risk-based measures to track the progress of AI systems.  Third, our work suggests that AI literacy interventions may be important for mitigating the uneven impacts of AI in the workplace.

\subsection{Related literature}
\paragraph{Empirical evaluations of AI impact} 
A growing body of empirical literature evaluates how AI assistance affects individual skill formation and productivity.

On skill formation, AI can enhance individual learning in certain educational contexts \citep{nie2025gpt,henkel2024effective, kestin_ai_2025}, particularly when personalized examples provide significant learning benefits. However, other work highlights potential drawbacks to task engagement and skill development arising from unfettered access to GenAI tools. Neuroscience findings by \cite{kosmyna2025your} suggest that LLM
use reduces cognitive activity and accumulates ``cognitive debt'', where users fail to internalize the knowledge. Randomized controlled trials in classes \cite{Bastani2025doi:10.1073/pnas.2422633122,poulidis2025self} show that individuals' ability to build knowledge may be impaired by a lack of hands-on practice. Consistent with these findings on skill degradation, we develop a model of skill development and characterize its steady state, providing a theoretical understanding of how AI shapes long-term human capital. 
More recently, \cite{shen2026ai} run an RCT on developer learning and find similar patterns. Our model formalizes this learning-by-doing process using a birth-death chain over skill states.

A substream of this empirical research examines AI literacy---understood as the ability to critically evaluate AI-generated output. \cite{roldan2024genai} find that high-ability students benefit more from GenAI tools in a university debate competition, as they are better at using the AI-generated information; \cite{eames2026computer} similarly observes that in math learning, higher-ability students gain more from AI assistance because they invest more complementary effort; \cite{allen2022algorithm} show that workers with greater domain experience are more accurate at judging algorithmic output. We formalize this moderating channel via heterogeneous AI literacy: we uncover that heterogeneity in evaluating AI output shifts effort choices and leads to skill polarization in the long run.

On productivity impact, empirical evidence in the context of AI-assisted coding \citep{cui2024effects}, writing \citep{noy2023experimental}, and customer support tasks \citep{brynjolfsson2025generative} indicates that GenAI tools improve individual productivity and could benefit less-experienced or lower-skilled workers more. These studies primarily focus on tasks within AI’s capability, where humans effectively implement AI's suggestions, as captured by our basic model. When AI assistance is unreliable, however, AI can disrupt the workflow and degrade the outputs. \cite{dell2023navigating} conduct a randomized controlled trial with consultants, which suggests that for hard tasks outside the ``jagged frontier'' of AI capabilities, humans using AI are less likely to produce correct solutions compared to those without AI; \cite{becker2025measuring} suggest that unreliable AI assistance slows down developers, especially in complex coding environments; \cite{snyder2026algorithm} examine algorithm augmentation in customer service queueing systems and find that throughput times decrease only when the algorithm quality is superior. Our productivity paradox under AI unreliability provides a micro-founded theory for these findings: effort loss from unreliable AI can exceed the direct productivity gain from reliable AI depending on the absolute risk aversion of the production function.

\paragraph{Human-AI interaction} Our work relates to recent studies on human-AI interaction in decision-making and cognitive creation tasks. In a meta-analysis of more than 100 empirical studies, \cite{vaccaro2024combinations} find that human-AI interaction often fails to improve performance, with a larger loss in decision tasks (making decisions) than in creation tasks (generating open-response content). \cite{shi2026llms} investigate such a decision failure experimentally in a classification task and find that it arises even when the LLM returns accurate statistics, suggesting the failure can stem from the nuanced interplay of human and AI. The majority of existing modeling work in operations management examines (the failure of) decision-making tasks over discrete labels \citep{boyaci2024human,mclaughlin2024designing} or parameterized decision rules \citep{balakrishnan2026human,grand2026best}. We instead focus on creation tasks and uncover that the productivity shortfall occurs even in these settings that do not involve decision making. For content generation tasks, 
\cite{castro2024human} build a Bayesian framework where users choose the amount of information to share with AI and show how the tradeoff between output fidelity and communication cost can lead to output homogenization and societal bias.
\cite{cui2025unlocking} study an AI-powered service system and show that low-skill individuals gain from AI assistance due to passive reliance on AI, high-skill individuals benefit as they adapt their behavior, while users with medium experience levels do not benefit.
\cite{gans2026ring} consider task automation with task interdependence, where an agent allocates a fixed time endowment on different tasks, and machines can substitute for individual tasks. In contrast, our model focuses on AI-assisted productivity for a single task but considers the interplay of different production factors---skill, effort, and AI assistance. 

Most relevant to our work, \cite{Agrawal2025NBERw34034} develop a model of AI-assisted task improvements in which AI substitutes for human effort while still increasing utility. While this insight is consistent with our findings, our work primarily focuses on the effect of AI assistance on productivity gains and uncovers that under mechanisms of skill dynamics and AI unreliability, the net productivity gain from greater AI assistance can disappear; further, we provide simple measures that characterize when such a productivity paradox occurs.
\cite{Agrawal2025NBERw34034} also establish that the variance of individual benefits is U-shaped in AI level, first narrowing as AI compresses the raw skill gap, then widening as a skill-judgment effect re-amplifies it. This insight is distinct from our skill polarization phenomenon: their result states that increasing AI assistance might increase inequality, while ours shows that, at a {\em fixed} level of AI assistance above a threshold, heterogeneous AI literacy and endogenous skill evolution disperse skills across the population.

\paragraph{Concurrent work} We also note that several concurrent papers study settings similar to ours. 

\cite{bastani2025human,siderius2026use} study a principal-agent model where an agent has access to AI assistance, which they can use in conjunction with spending costly effort; the principal compensates the agent based on its output. \cite{bastani2025human} show a contracting paradox where the principal may prefer a less reliable AI tool since more reliable AI systems are harder and more expensive to monitor. \cite{siderius2026use} show that profit is non-monotonic in AI quality and that the profit-maximizing outcome may come at the cost of long-term skill degradation. Unlike contracting paradoxes that focus on a mismatch of incentives, we expose paradoxes in productivity itself, stemming either from misalignment between AI-driven short-term productivity gains and long-term skill development, or from myopic effort reduction driven by unreliable AI. 

\cite{acemoglu2026ai,bondi2026skill} study models that incorporate social learning in which agents have access to AI and one agent's learning process generates externalities for the learning of others. \cite{bondi2026skill} examines the measurement implications of individual skill degradation and learning spillover bias for AI's productivity effects, arguing that standard estimates overstate AI's net productivity effect. The study of skill endogeneity is similar to our work but the focus of \cite{bondi2026skill} is on how panel and RCT estimands of productivity may differ under AI adoption (due to skill atrophy or spillover bias), which is different from our focus on the mechanisms that can induce the emergence of productivity paradoxes.
\cite{acemoglu2026ai} model the learning externality of human effort, in which successful decision-making requires both societal and individual knowledge, and investigate how agentic AI degrades the collective knowledge stock. While \cite{acemoglu2026ai} mostly focus on externalities of individual-level skill on community-level knowledge, our work isolates individual-level mechanisms that may generate a productivity paradox.

Closest to our work, \cite{caosun2026augmentation}  also show a productivity paradox and skill-polarization effect (the latter is termed skill stratification). The productivity paradox is driven by a similar mechanism to the one we study in Section~\ref{sec:bdc}; that said, \cite{caosun2026augmentation} do not consider the effect of AI unreliability as a distinct mechanism causing a productivity paradox. With respect to skill polarization, their result stems from specifying interaction terms between skill and AI in the production function (absent the notion of effort), whereas we assume a general skill-AI interaction function that includes the perfect substitution case and Cobb-Douglas case, etc.

\paragraph{Innovation and technology}  Our work also relates to a long-standing literature in economics analyzing the impact of major technological revolutions (steam engine, electricity, computers). While such technological innovations are associated with productivity gains, these gains are often lagging and dependent on organizational and institutional changes \citep{david1990dynamo,BRESNAHAN1995,brynjolfsson2021productivity}, reflecting transition periods of partial equilibrium before the economy absorbs the new technology. The productivity failures we identify are distinct from the partial-equilibrium explanation: they arise from the specific characteristics of GenAI technology, rather than from transitional frictions.  GenAI distinguishes itself from previous general-purpose technologies \citep{wachter2024will} in several ways: its versatility in automating a wide spectrum of cognitive and creative tasks and the risks it poses from unreliability and hallucination. As such, we focus on modeling individual-level productivity as the outcome of complex skill-effort dynamics. That said, how these effects play out in general equilibrium, considering demand and supply elasticities, wages, and employment, is a relevant direction for future work.

\section{Model} \label{sec:basic}
This section presents our model incrementally. We begin with the three production factors: skill, effort, and AI assistance, and then introduce two key mechanisms that drive our main results of productivity paradoxes in Sections~\ref{sec:bdc} and ~\ref{sec:AI_unreliability} respectively: endogeneity of skill and unreliability of AI. Section~\ref{sec:AI_literacy} later extends this model to add one more component: heterogeneous AI literacy.

We consider heterogeneous agents, each with a skill level $s\geq 0$. The agent chooses an effort level $e\geq 0$ under an AI assistance level $a \geq 0$. The AI assistance level captures how much the AI tool can support human workers per unit time. A higher AI assistance level means a more proficient AI tool in advancing the task. 
We assume that skill, effort, and AI assistance jointly determine the total input of the production function according to $x=e+h(s,a)$, where the \emph{skill-AI interaction function} $h(s,a)\geq 0$ is an increasing and differentiable function in $s$ and $a$. We model effort as an additive component of the total input to capture its direct contribution to production. Prior worker experience translates into skills that reduce the effort required for a given output productivity level; similarly, AI automation reduces the human effort needed for task completion. We allow human skill and AI assistance to interact in a general manner, captured by the functional form $h(s,a)$, so the benefits of AI assistance may vary with different skills. This includes the perfect substitution interaction function $h(s,a)=s+a$ and Cobb-Douglas interaction functions such as $h(s,a)=\sqrt{sa}$; in the latter, both skill and AI assistance are necessary to generate a positive input contribution. For ease of notation, we use the shorthand $h_a(s,a)$ to represent the partial derivative of the skill-AI interaction function with respect to its second argument, the AI assistance level.

The input yields a productivity $p(x)$, where $p(x)\geq 0$ is a (weakly) increasing, (weakly) concave, continuous, and twice differentiable production function. Exerting effort $e$ incurs a cost of $c(e)$, where $c(e) = \gamma e$ with $\gamma > 0$ is a linear cost function. The cost function depends only on effort and not AI assistance, which isolates the effect of AI in production.

\subsection{Mechanisms driving the productivity paradoxes}

\paragraph{Mechanism 1: Endogeneity of skill}\label{sec:model_skill_endo}
Skill is itself endogenous to effort. As established in cognitive psychology, 
``learning is an active process; to acquire information most effectively, we must actively manipulate 
it through processes such as elaboration and rehearsal'' \cite[Chap.~9]{stangor2014introduction}. We thus let an agent's skill evolve over time as a function of the effort exerted, which provides the first mechanism driving the productivity paradox. We define the class of skill sets as ${\cal S}:=\{(s_1,s_2,\cdots,s_N): s_1 < s_2 < \cdots < s_N, N\in \mathbb Z^+\}$ and consider the corresponding birth-death chain with the set of skill states $S\in {\cal S}$. The skill evolution process is specified by upward transition rates of moving from $s_k$ to $s_{k+1}$ and downward transition rates of moving from $s_k$ to $s_{k-1}$. A higher effort level increases the likelihood of skill improvement. We assume that skills improve based on the endogenous effort level and decay at a time-homogeneous rate, which are given by 
$\mathcal{K}_{s_k s_{k+1}}= \lambda\big(e^\star_k\big)$ and $\mathcal{K}_{s_k s_{k-1}}= \mu,$
where $\mu > 0$ and $\lambda(\cdot)>0$ is an increasing, (weakly) concave, and differentiable transition function. Here, effort $e^\star_k$ denotes the utility-maximizing effort exerted at skill state $s_k$, which we formalize later in \eqref{eq:full_effort_expression}.
This setup extends the two-state (unlearned and learned) skill acquisition model of \cite{corbett1994knowledge} and it is a discrete-time version of the continuous dynamics in \cite{ben1967production}.

\paragraph{Mechanism 2: Unreliability of AI}\label{sec:model_ai_unreliability}
AI systems are rarely fully reliable. For example, LLMs have a so-called ``jagged frontier'' \citep{dell2023navigating}: they can excel at certain complex tasks while being inaccurate or erroneous on more basic ones, making their degree of usability across tasks difficult to discern. To capture this uncertainty in outputs, we model the AI assistance level $a$ as a random variable, which provides the second mechanism driving the productivity paradox. We introduce a reliability level $q\in[0,1]$, such that the AI functions desirably with a proficiency $\bar{a}$ with probability $q$, i.e., $\mathbb{P}(a=\bar{a})=q$, and fails to assist with probability $1-q$, i.e., $\mathbb{P}(a=0)=1-q$. This notion of unreliability abstracts away the practical sources of randomness in LLM outputs---random sampling of a prompt from a distribution of tasks and user instructions, and random generation of output tokens influenced by how the LLM is configured (e.g., temperature).

Suppose that the agent has a prior belief on the AI assistance level and commits an effort level before the AI's output is realized. Another interpretation is that the agent has no means of verifying the output quality, and optimizes the effort ex-ante based on its expectation for the proficiency of the AI system (this will be relaxed when we introduce AI literacy in Section~\ref{sec:AI_literacy}). Given a skill state $s_k$, the expected utility of the agent is 
$$u(e,s_k,\bar{a},q):=q\cdot p(e+h(s_k,\bar{a})) + (1-q)\cdot p(e+h(s_k,0))-\gamma e.$$ 
The utility-maximizing effort is:\footnote{The maximizer may not be unique, for example, when the production function $p$ contains a linear segment with slope $\gamma$. When the maximizer is not unique, the largest one is selected. In the following, we will use $\arg\max u(\cdot)$ instead of $\max (\arg\max u(\cdot))$ to denote the largest maximizer for notational simplicity.}
\begin{equation}\label{eq:full_effort_expression}
    e^\star_k:=e^\star(s_k,\bar{a},q)=\arg\max_{e\geq 0}u(e,s_k,\bar{a},q).
\end{equation}
This effort level contributes to skill development as we discussed earlier, and leads to the following productivity:
$$p^\star(s_k,\bar{a},q):= q\cdot p(e^\star(s_k,\bar{a},q)+h(s_k,\bar{a})) + (1-q)\cdot p(e^\star(s_k,\bar{a},q)+h(s_k,0)).$$

\subsection{Modeling assumptions and evaluation metrics}\label{sec:assumption_metric}
After formulating the model, we introduce the key assumptions underlying our analysis. These assumptions lead to the results in the main body, which focus on several key metrics that we introduce below.

\paragraph{Modeling assumptions.} Our findings hinge on the following modeling assumptions.
\begin{itemize}
\item \emph{Linear cost:} We consider a linear cost function of effort, which is natural when effort represents time spent under a constant per-unit opportunity cost and leads to
cleaner analytical results. However, the paradoxes identified carry over to convex costs in a weaker form (see Section~\ref{sec:conclusion} and Appendix~\ref{app:cvx_cost}).
\item \emph{Myopic users:} We posit that agents optimize their short-term utility and do not account for long-term consequences of AI usage on their skill. This assumption is supported by empirical evidence,\footnote{See, e.g., the 2025 global report on AI attitudes by KPMG and the University of Melbourne: \url{https://assets.kpmg.com/content/dam/kpmgsites/xx/pdf/2025/05/trust-attitudes-and-use-of-ai-global-report.pdf}. Accessed on 05/2026.} which shows myopic patterns of individual usage of AI. However, the paradoxes identified carry over even with forward-looking agents (see Section~\ref{sec:conclusion} and Appendix~\ref{app:non_myopic}).
\end{itemize}

\paragraph{Evaluation metrics.}
Our evaluation focuses on AI's impact on productivity and skill at the steady state. For a fixed set of states $S$, transition rates $\lambda(e)$ and $\mu$, let $\pi_k$ be the steady-state probability of an agent being at skill state $s_k$. Standard birth-death process result \cite[Chap.~6.11]{grimmett2001probability} shows that the steady state distribution $\boldsymbol{\pi}(\bar{a},q)=(\pi_1(\bar{a},q),\pi_2(\bar{a},q),\cdots,\pi_N(\bar{a},q))$ is given below\footnote{We suppress the dependence on $S, \lambda,\mu$ here. In Section~\ref{sec:bdc}, when we have full reliability ($q=1$), we will drop the dependence on $q$, and the AI assistance level is deterministic, so we write it as $\boldsymbol{\pi}(a)$.} (see Appendix~\ref{app:steady_state} for the derivation). 
\begin{equation}\label{eq:steady_state_prob}
    \pi_k(\bar{a},q)=\pi_1( \bar{a},q)\prod_{i=1}^{k-1}\frac{\lambda(e^\star(s_i, \bar{a}, q))}{\mu} \quad\text{and}\quad  \pi_1(\bar{a},q)=\dfrac{1}{1+\sum_{k=2}^{N}\prod_{i=1}^{k-1}\frac{ \lambda(e^\star(s_i, \bar{a}, q))}{\mu}}.
\end{equation}
Here, by convention, we write $\prod_{i=1}^{0} x_i = 1$.
We evaluate the impact of AI on productivity at the steady state. 
\begin{align}\label{eq:productivity_obj}
    \mathcal{P} (\bar{a}) :=\sum_{k=1}^{N} \pi_k(\bar{a},q) \cdot p^\star(s_k,\bar{a},q).
\end{align}

Our first object of interest is whether productivity always increases in the proficiency of the AI system. A \emph{productivity paradox} may emerge when increasing AI proficiency fails to increase productivity. Although we focus on productivity, which disregards the cost of effort, similar results can arise under a welfare notion that takes the cost into account (see Section~\ref{sec:conclusion} and Appendix~\ref{app:welfare}).
\begin{definition}[Productivity paradox]\label{def:productivity_paradox}
A productivity paradox refers to a situation in which an increase in AI proficiency $\bar{a}$ leads to a lower productivity level $\mathcal{P} (\bar{a})$.
\end{definition}

Our second object of interest considers the distribution $\boldsymbol{\pi}(\bar{a},q)$ of human skills in the long run. We call a distribution unimodal when it has at most one peak, and multimodal when it has at least two peaks, where the latter corresponds to \emph{skill polarization}.\footnote{ Notice that the connection to skill polarization is qualitative: the steady-state distribution exhibits distinct peaks separated by less populated intermediate skill levels. Econometric measures of polarization are related to multimodality and spikiness but are not restricted to this notion~\citep{2990f903-a1fe-36dd-8511-1b59b65a89d2}.}
\begin{definition}[Skill polarization]\label{def:skill_polarization}
A distribution $\boldsymbol{\pi}=(\pi_1,\pi_2,\cdots,\pi_N)$ is unimodal if there exists an index $k\in [N]$ such that $\pi_1\leq \pi_2 \leq \cdots \leq \pi_k$ and $\pi_k \geq \pi_{k+1} \geq \cdots \geq \pi_N$, and it is multimodal otherwise. A skill polarization refers to a situation in which the steady-state skill distribution is multimodal.
\end{definition}

\subsection{Basic model}\label{ssec:basic}
To end this section, we analyze a basic model in which skill is exogenous, i.e., $N=1$, so we write $s$ for the skill level, and AI assistance is fully reliable, i.e., $q=1$, so we write $a$ for the deterministic assistance level. This allows us to first understand the model in the absence of the two mechanisms introduced above.
The utility-maximizing effort and corresponding productivity then simplify to 
\begin{equation}\label{eq:basic_effort_def}
    e^\star(s,a):= \arg\max_{e\geq 0} \big(p(e+h(s,a))-\gamma e\big) \quad\text{and}\quad p^\star(s,a):= p(e^\star(s,a)+h(s,a)).
\end{equation}
We define the \emph{critical level} $x^\star= x^\star(p,\gamma)$ as the (largest) maximizer of the net function $p(x)-\gamma x$: 
\begin{equation*}
    x^\star = \arg\max_{x\geq 0} \big(p(x)-\gamma x\big).
\end{equation*}
Beyond the diminishing return in input level, we further assume that $\limsup_{x\to\infty} \frac{p(x)}{x}<\gamma$ to ensure boundedness of the critical level $x^\star$.
It represents the largest point where the slope of the production function is greater than or equal to $\gamma$, i.e., $p'(x^\star)\geq \gamma$, and $p'(x^\star +\varepsilon)< \gamma$ for any $\varepsilon>0$. Equipped with the definition of the critical level, we formalize our finding in the following proposition. Below, we define $x^+ :=\max(x,0)$.
\begin{proposition}\label{prop:effort_productivity_expressions}
The agent's effort level $e^\star(s,a)$ and productivity $p^\star(s,a)$ are given by:
$$ e^\star(s,a)=(x^\star-h(s,a))^+ \quad\text{and}\quad  p^\star(s,a)=\max (p(x^\star), p(h(s,a))).$$
\end{proposition}

\begin{proof}[Proof of Proposition~\ref{prop:effort_productivity_expressions}]
For a given AI assistance level $a$ and a given skill level $s$, the concavity of $p$ implies that the utility is increasing in $e$ for $e\leq x^\star-h(s,a)$ and is decreasing otherwise.
The following case distinction characterizes the agent's effort and corresponding productivity:
\begin{enumerate}
    \item $h(s,a) \leq x^\star$. The agent's utility-maximizing effort is a non-negative value that brings the total input to the critical level $x^\star$, i.e., $e^\star(s,a)=x^\star-h(s,a)$. The corresponding productivity is given by $p^\star(s,a)=p(x^\star)$. 
    \item $h(s,a) > x^\star$. The agent's utility is decreasing in $e$ for all non-negative effort levels. Thus, the agent's utility-maximizing effort is $e^\star(s,a)=0$. The corresponding productivity is given by $p^\star(s,a)=p(h(s,a))$.
\end{enumerate}
Combining both cases, we complete the proof.
\end{proof}

The effort level decreases with both skill and AI assistance, while productivity increases in both. AI assistance thus serves as a ``bicycle for the mind'' \citep{Agrawal2025NBERw34034}: for a fixed skill level, AI assistance can substitute human effort and achieve an overall productivity gain, which recovers the finding of \cite{Agrawal2025NBERw34034}. 
This arises as human effort is substituted by AI assistance: as AI assists more, the marginal benefit of human effort decreases. As we discuss in the following sections, this pattern breaks once we consider endogeneity of skill or unreliability of AI in the model. 

\section{The Impact of Skill Development: A Simpson's Paradox}\label{sec:bdc}
We first study the impact of skill development. Starting from the basic model in Section~\ref{ssec:basic}, we now activate Mechanism 1, allowing skills to evolve endogenously as described in Section~\ref{sec:model_skill_endo}. To isolate the effect of endogeneity of skills, we exclude Mechanism 2 and assume AI assistance remains fully reliable, so we write $a$ for the deterministic assistance level. 
Recall that we have $N\geq 2$ skill states $S=(s_1,s_2,\cdots,s_N)$, where $s_1 < s_2 < \cdots <s_N$. Under full reliability, the transition rates at every skill state are given by
$\mathcal{K}_{s_k s_{k+1}}(a)= \lambda(e^\star(s_k,a))$ and $\mathcal{K}_{s_k s_{k-1}}(a)= \mu,$
where $e^\star(s,a):= \arg\max_{e\geq 0} \big(p(e+h(s,a))-\gamma e\big)$. Notice that our basic model of Section~\ref{ssec:basic} coincides with the birth-death chain when there is only one skill state.

Our object of interest is the steady-state productivity $\mathcal{P}(a)$ as defined in \eqref{eq:productivity_obj}. Once we incorporate Mechanism 1 into the basic model, the monotonic relationship between AI assistance and productivity established in Section~\ref{ssec:basic} breaks down. As one might expect, AI assistance and skill development are substitutes. This drives our first productivity paradox, and specifically, we characterize when this productivity paradox can arise via a condition that explicitly compares the long-term deskilling loss caused by AI with the short-term AI productivity gains. 

\subsection{Main result: emergence of productivity gain and decline regions}\label{sec:bdc_result}
Our main result in this section shows that AI may reduce productivity when the skill development is more sensitive than productivity to AI assistance, and vice versa. To formalize this, we first define a key notion that compares two semi-elasticity terms \citep{wooldridge2009introductory}, which measure the sensitivity of skill loss and productivity gain to AI assistance. This comparison allows us to characterize when the productivity gain of AI assistance is offset by deskilling. Specifically, we define the inverse function $g(\cdot)$ such that $h(s, g(s)) = x^\star$; this denotes the level of AI assistance beyond which the agent at skill level $s$ exerts no effort. Below, we define $[n]:=\{1,2,\cdots,n\}$. 

\begin{definition}\label{def:semi-elasticity_general}
    For a given index $m\in[N-1]$, the semi-elasticity gap between skill development and productivity is \footnote{We assume $\lambda(\cdot)>0$, $s_{m+1}>s_m$, and $h$ is an increasing function, so the terms in the logarithm/denominator are positive.}
\begin{align*}
    \Delta_m
    &= -\underbrace{\frac{d}{da}\log\Big(\prod_{j=1}^m \lambda(x^{\star}-h(s_j,a))\Big)\Big|_{a=g(s_m)}}_{\substack{\text{semi-elasticity of} \\ \text{development to AI assistance}}} - \underbrace{\frac{d}{da}\log\Big(p(h(s_{m+1},a))-p(x^{\star})\Big)\Big|_{a=g(s_m)}}_{\substack{\text{semi-elasticity of} \\ \text{productivity to AI assistance}}}\\
    &= \sum_{j=1}^m\frac{-\frac{d}{da}\lambda(x^\star-h(s_j,a))|_{a=g(s_m)}}{\lambda(x^\star-h(s_j,g(s_m)))}- \frac{\frac{d}{da}p(h(s_{m+1},a))|_{a=g(s_m)}}{p(h(s_{m+1},g(s_m))) -p(x^\star)}
\end{align*}
\end{definition}
The second underbraced term is the semi-elasticity of productivity to AI assistance, measuring the percentage change in the productivity gap between the closest higher skill state $s_{m+1}$ and the productivity at the current skill $s_{m}$ induced by a change in AI assistance. The first underbraced term is the semi-elasticity of skill development to AI assistance, which measures the percentage change in the steady-state distribution $\pi_{m+1}$ induced by a change in AI assistance. 
This gap is sufficient for the characterization, as we focus on a sufficiently large decay rate $\mu$. Intuitively, a positive semi-elasticity gap represents an effort-critical scenario, where a reduction in effort under AI assistance leads to a significant decline in skill development, compared to the productivity gain. This guarantees that the direct productivity gain of AI assistance is dominated by the loss from deskilling.
The positive semi-elasticity gap condition is satisfied by several examples, such as production functions that satisfy unboundedness or vanishing derivative (formalized in Appendix~\ref{app:pf_positive_semi-elasticity_gap}). 

We define the \emph{adjacent interval} ${\cal I}_m=[(g(s_{m+1}))^+, (g(s_m))^+)$ for $m\in [N-1]$, with ${\cal I}_0 = [(g(s_1))^+,+\infty)$ and ${\cal I}_N = [0,(g(s_N))^+)$, forming a partition of $\mathbb R^+ = \cup_{m=0}^{N} {\cal I}_m$. Technically, each sub-interval ${\cal I}_m$ corresponds to a smooth region where the productivity exhibits a unimodal pattern, as we show below. Intuitively, this is the critical region in which an agent with skill $s_m$ exerts effort, while an agent with skill $s_{m+1}$ does not. Our main result is a characterization of the productivity increase and decline regions within each sub-interval; the latter emerges precisely when the semi-elasticity gap is positive. We focus on the limiting behavior of the decay rate $\mu$, but Proposition~\ref{prop:bdc_prod_n=2} quantifies the exact decay-rate magnitude required for productivity decline when $N=2$. In the main result of this section (Theorem~\ref{thm:bdc_productivity}), we assume that $h(s,a)=b(s)+w_1(s)w_2(a)$, where $b(s),w_1(s),w_2(a)\geq 0$ are increasing and differentiable functions. This specification includes many general forms such as the linear substitution  $h(s,a)=\omega_1 s+\omega_2 a$ and the Cobb-Douglas $h(s,a)=s^{\omega_1} a^{\omega_2}$, where $\omega_1,\omega_2>0$.

\begin{theorem}\label{thm:bdc_productivity}
    Consider any concave production function $p(x)$ with a skill-AI interaction function $h(s,a)=b(s)+w_1(s)w_2(a)$ where $b(s),w_1(s),w_2(a)$ are increasing, linear cost function $c(e)=\gamma e$, increasing, concave transition function $\lambda(e)$, and a set of states $S=(s_1,s_2,\cdots,s_N)$. 
    There exists $\bar{\mu}>0$ such that for all $\mu>\bar{\mu}$ and ${\cal I}_m\neq \emptyset$, 
    \begin{enumerate}
    \item If $m\in[N-1]$ and $\Delta_m>0$, productivity $\mathcal{P}(a)$ is increasing-then-decreasing on ${\cal I}_m$. 
    \item If $m\in \{0,N\}$ or $\Delta_m<0$, productivity $\mathcal{P}(a)$ is (weakly) increasing on ${\cal I}_m$. 
    \end{enumerate}
\end{theorem}
Theorem~\ref{thm:bdc_productivity} (proof provided in Section~\ref{sec:bdc_degradation}) reveals a possibly non-monotonic relationship between productivity and AI assistance.
When the semi-elasticity of skill development to AI assistance is larger than that of productivity to AI assistance, a higher AI assistance level can harm productivity; on the other hand, when the semi-elasticity of skill development to AI assistance is smaller than that of productivity to AI assistance, productivity increases constantly. The semi-elasticity gap $\Delta_m$ determines whether AI can negatively influence the productivity within the $m$-th sub-interval: a positive (negative) value represents the cases when AI is potentially harmful (always beneficial). Intuitively, for each sub-interval, AI is only beneficial to productivity until a certain point where AI becomes detrimental. This is to the exclusion of the boundary cases $m=0$ and $N$, where productivity only increases in AI level because either effort or productivity plateaus. Note that this result requires a large decay rate $\mu$. For foundational skills, where decay is minimal, there is almost only skill acquisition and the productivity paradox is unlikely to arise. The identified paradoxes are more likely to happen for newly acquired and more advanced skills, such as learning a new Python library, where decay is more pronounced, which is consistent with empirical findings \citep{shen2026ai}.

We complement Theorem~\ref{thm:bdc_productivity} by establishing that the productivity ratio between higher and lower AI assistance levels can be arbitrarily small; the proof is deferred to Appendix~\ref{app:proof_bcd_bad_epsilon_general}. 
\begin{proposition}\label{prop:bcd_bad_epsilon_general}
    For any $\varepsilon>0$ and any given $N\geq 2$ skill states $S=(s_1,s_2,\cdots,s_N)$ with $0=s_1<s_2<\cdots<s_N$, there exists a concave production function\footnote{For ease of exposition, we show the proposition for a concave piecewise linear function, but the same conclusion holds with a twice differentiable function via a small smooth perturbation of the function around the kink of the piecewise linear function.} $p$ with a perfect substitution
    skill-AI interaction function $h(s,a)=s+a$, a linear cost function $c$, a concave transition function $\lambda$, and a decay rate $\mu$, such that the ratio of productivity between a higher and a lower AI assistance level is
    $$\inf_{a_\ell <a_h}\dfrac{\mathcal{P}(a_h)}{\mathcal{P}(a_\ell )}\leq \varepsilon.$$ 
\end{proposition}

\subsection{Technical crux: skill degradation and semi-elasticity gap}\label{sec:bdc_degradation}
Here, we unpack the skill degradation mechanism behind the potential productivity decline. We start with analyzing how AI affects the distribution of skills across different states. We examine how the skill distribution changes under different AI levels: because AI assistance and effort are perfect substitutes in our model, AI assistance mechanically leads to effort reduction and thus skill degradation.
\begin{proposition}\label{prop:bdc_skill_dist}
    Consider any concave production function $p(x)$ with a skill-AI interaction function $h(s,a)$, linear cost function $c(e)=\gamma e$, increasing, concave transition function $\lambda(e)$, a set of states $S=(s_1,s_2,\cdots,s_N)$, and two levels of AI assistance $a_{\ell}$ and $a_h$ with $0 \leq a_\ell < a_h$. Then the steady-state distribution under a lower AI assistance level $a_\ell$ is first-order stochastically greater than that under a higher AI assistance level $a_h$:
    $$\sum_{i=1}^k \pi_i(a_\ell) \leq \sum_{i=1}^k\pi_i(a_h) \quad \text{for all} \quad  k\in [N].$$
\end{proposition}
Proposition~\ref{prop:bdc_skill_dist} demonstrates that AI leads to a downward shift in the steady-state skill distribution. It provides an intuition for why AI assistance may not directly translate to an overall productivity gain; in Section~\ref{sec:bdc_discussion}, we connect this result to the uneven distributional effect that is central to Simpson's paradox. The proof of Proposition~\ref{prop:bdc_skill_dist} is deferred to Appendix~\ref{app:proof_bdc_skill_dist}.

\paragraph{Semi-elasticity condition and memory decay} Proposition~\ref{prop:bdc_skill_dist} provides insights into the potential productivity decline; we now delve into the sufficient conditions on the limiting behavior of the decay rate $\mu$. Technically, a productivity decline region emerges under two conditions. The first condition is a positive semi-elasticity gap, which quantifies the dominance of the loss from deskilling over AI gains beyond Proposition~\ref{prop:bdc_skill_dist}. The second condition is a sufficiently large memory decay rate. In particular, when $\mu$ approaches zero, skill converges to the highest state with high probability, thereby recovering the basic model in Section~\ref{ssec:basic} and excluding the possibility of a productivity decline. 
Below, we provide the proof idea of Theorem~\ref{thm:bdc_productivity} to illustrate how these conditions matter. We begin with a lemma that characterizes the derivative of the productivity with respect to AI assistance. 
Given that $\lambda(\cdot)>0$, we define 
$$A_{k,m}(a):=\frac{1}{\prod_{i=k}^{m}\lambda(x^\star -h(s_i,a))}, \forall k\in[m] \quad \text{and} \quad B_{k,m}=\lambda(0)^{k-m-1} , \forall k\in[N]\setminus [m].$$
\begin{lemma}\label{lem:thm_productivity_derivative}
Given $m\in[N-1]$, for any $a\in {\cal I}_m$, the derivative of productivity is
\begin{align*}
    \mathcal{P}'(a) 
    = \frac{1}{\big(\sum_{k=1}^m \frac{A_{k,m}(a)}{\mu^{k-1}} +\sum_{k=m+1}^{N} \frac{B_{k,m}}{\mu^{k-1}}\big)^2} \cdot
    &\bigg[\Big( \sum_{k=m+1}^{N} \frac{B_{k,m}}{\mu^{k-1}} \frac{d}{da}p(h(s_k,a))\Big) \Big(\sum_{k=1}^m \frac{A_{k,m}(a)}{\mu^{k-1}} +\sum_{k=m+1}^{N} \frac{B_{k,m}}{\mu^{k-1}}\Big) \\
    &- \Big( \sum_{k=m+1}^{N} \frac{B_{k,m}}{\mu^{k-1}}(p(h(s_k,a))-p(x^\star))\Big) \sum_{k=1}^m \frac{A_{k,m}'(a)}{\mu^{k-1}}\bigg].
\end{align*}
\end{lemma}
Lemma~\ref{lem:thm_productivity_derivative} allows us to characterize the existence of a decreasing region, as stated next. Below, we define $\mathcal{P}'_{-}(a):=\lim_{\delta \to 0^-} \frac{\mathcal{P}(a+\delta)-\mathcal{P}(a)}{\delta}$ as the left-hand derivative.\footnote{This limit is well defined at the right endpoint
of ${\cal I}_m$. By Lemma~\ref{lem:thm_productivity_derivative}, the productivity $\mathcal{P}(a)$ is continuously differentiable on the interior of ${\cal I}_m$, so $\mathcal{P}'(a)$
extends continuously to the closure of ${\cal I}_m$.}
\begin{lemma}\label{lem:bdc_productivity_unimodal}
    For $m\in [N-1]$, the productivity $\mathcal{P}(a)$ is either monotone or increasing-then-decreasing on an adjacent interval ${\cal I}_m$. In particular, a decreasing region of productivity $\mathcal{P}(a)$ exists within ${\cal I}_m$ if and only if $\mathcal{P}'_{-}(a)<0$ at the right endpoint of ${\cal I}_m$. 
\end{lemma}
The proofs of Lemma~\ref{lem:thm_productivity_derivative} and Lemma~\ref{lem:bdc_productivity_unimodal} are deferred to Appendix~\ref{app:thm_productivity_derivative} and Appendix~\ref{app:proof_bdc_productivity}, respectively.
Equipped with these lemmas, we now prove Theorem~\ref{thm:bdc_productivity}. 

\begin{proof}[Proof of Theorem~\ref{thm:bdc_productivity}]
Recall the steady-state productivity is given by Equation~\eqref{eq:steady_state_prob}. We start from the cases $m=N$ and $m=0$, where we want to show that the steady-state productivity (weakly) increases. As the agent optimizes the utility at the current state, Proposition~\ref{prop:effort_productivity_expressions} implies that 
\begin{itemize}
    \item If $a \in {\cal I}_N$, the productivity at each state remains  $p(x^\star)$ and the productivity $\mathcal{P}(a)=p(x^\star)$ is a constant.
    \item If $a \in {\cal I}_0$, the utility-maximizing effort is $0$ at each state and the productivity $ \mathcal{P}(a)$ is increasing in~$a$.
\end{itemize}
We then focus on intermediate cases where $a \in {\cal I}_m$ with $m\in [N-1]$ and ${\cal I}_m$ being non-empty. As the denominator of the derivative is always positive, Lemma~\ref{lem:thm_productivity_derivative} implies that the sign of the derivative at the right endpoint of ${\cal I}_m$ is determined by the numerator $X_m$, where
\begin{align*}    
    X_m &= \Big( \sum_{k=m+1}^{N} \frac{B_{k,m}}{\mu^{k-1}} \frac{d}{da}p(h(s_k,a))|_{a=g(s_m)} \Big) \cdot\Big(\sum_{k=1}^m \frac{A_{k,m}(g(s_m))}{\mu^{k-1}} +\sum_{k=m+1}^{N} \frac{B_{k,m}}{\mu^{k-1}}\Big) \\
    &- \Big( \sum_{k=m+1}^{N} \frac{B_{k,m}}{\mu^{k-1}}(p(h(s_k,g(s_m)))-p(x^\star))\Big) \cdot \Big( \sum_{k=1}^m \frac{A_{k,m}'(g(s_m))}{\mu^{k-1}}\Big).
\end{align*}
Lemma~\ref{lem:bdc_productivity_unimodal} implies that the existence of a productivity decreasing region depends solely on the sign of $\mathcal{P}'_{-}(g(s_m))$. The numerator $X_m$ can be viewed as a polynomial function of $\frac{1}{\mu}$ that we denote by $f(\frac{1}{\mu})$. Substituting the expression of $A_{1,m}$ and $B_{m+1,m}=1$, the coefficient of the lowest-order term $(\frac{1}{\mu})^{m}$ is 
\begin{align*}
    &B_{m+1,m}\frac{d}{da}p(h(s_{m+1},a))|_{a=g(s_m)} A_{1,m}(g(s_m)) - B_{m+1,m}(p(h(s_{m+1},g(s_m)))-p(x^\star))  A_{1,m}'(g(s_m))\\
    = &\frac{\frac{d}{da}p(h(s_{m+1},a))|_{a=g(s_m)}-\Big(p(h(s_{m+1},g(s_m))) -p(x^\star)\Big) \sum_{j=1}^m\frac{-\frac{d}{da}\lambda(x^\star-h(s_j,a))|_{a=g(s_m)}}{\lambda(x^\star-h(s_j,g(s_m)))}}{\prod_{i=1}^{m}\lambda(x^\star -h(s_i,g(s_m)))} 
    = -C_m\Delta_m,
\end{align*}
where $C_m=\frac{p(h(s_{m+1},g(s_m))) -p(x^\star)}{\prod_{i=1}^{m}\lambda(x^\star -h(s_i,g(s_m)))}$ is a positive value. By the continuity of $f$, there exists a value $\bar{\mu}_m$ such that when $\mu>\bar{\mu}_m$, it holds that the sign of $f(\frac{1}{\mu})$ coincides with the sign of $-\Delta_m$. Thus, we have the following case distinction:
1) $\Delta_m>0$: when $\mu>\bar{\mu}_m$, it holds that $\mathcal{P}'_{-}(g(s_m))<0$ and there exists a decreasing region on ${\cal I}_m$;
2) $\Delta_m<0$: when $\mu>\bar{\mu}_m$, it holds that $\mathcal{P}'_{-}(g(s_m))>0$ and productivity $\mathcal{P}(a)$ is always increasing on ${\cal I}_m$. Finally, taking $\bar{\mu}:=\max_{m\in[N-1]} \bar{\mu}_m$ completes the proof.
\end{proof}

\paragraph{Complete characterization for the two-state case}
The effect of memory decay can be made more explicit in the special case of $N=2$ states ($s_1<s_2$). Here, we derive a sharper version of Theorem~\ref{thm:bdc_productivity} specialized for $N=2$. To do so, we define the following threshold on the memory decay rate:
\begin{equation*}
    \bar{\mu}_2:= \dfrac{\lambda^2(0)p'(h(s_2,g(s_1)))\frac{h_a(s_2,g(s_1))}{h_a(s_1,g(s_1))}}{  \lambda'(0) \big(p(h(s_2,g(s_1)))-p(x^\star)\big)-\lambda(0)p'(h(s_2,g(s_1))) \frac{h_a(s_2,g(s_1))}{h_a(s_1,g(s_1))}}.
\end{equation*}
Here, $h_a$ is the partial derivative of function $h$ with respect to $a$.
\begin{proposition}\label{prop:bdc_prod_n=2}
   Consider any concave production function $p(x)$ with a skill-AI interaction function $h(s,a)$ of the form $h(s,a)=b(s)+w_1(s)w_2(a)$ where $b(s),w_1(s),w_2(a)$ are increasing, linear cost function $c(e)=\gamma e$, increasing, concave transition function $\lambda(e)$, and two states $s_1<s_2$. 
    \begin{enumerate}
        \item If $\Delta_1>0$ and $\mu > \bar{\mu}_2$, productivity $\mathcal{P}(a)$ is increasing-then-decreasing on ${\cal I}_1$, and (weakly) increasing on ${\cal I}_0$ and ${\cal I}_2$.
        \item If $\Delta_1\leq 0$ or $\mu \leq \bar{\mu}_2$, productivity $\mathcal{P}(a)$ is (weakly) increasing on $[0,\infty)$.
    \end{enumerate}
\end{proposition}
Proposition~\ref{prop:bdc_prod_n=2} gives a full characterization of productivity in the case of $N=2$, the productivity exhibits an N-shaped pattern---an initial increase followed by a decline and a subsequent increase---given that the semi-elasticity of skill development to AI assistance and the memory decay rate are large enough ($\Delta_1>0$ and $\mu>\bar{\mu}_2$). Conversely, we prove that if these conditions are not met, productivity always increases with AI assistance. The proof is deferred to Appendix~\ref{app:proof_bdc_productivity_n=2}.
Compared to the general case, the expression of productivity in the two-state case exhibits a simpler dependence on $\mu$. Proposition~\ref{prop:bdc_prod_n=2} gives the sufficient and necessary conditions considering semi-elasticity and memory decay jointly. In contrast, in the general setting, we require $\mu$ to be sufficiently large.

\subsection{Discussion}\label{sec:bdc_discussion}
\paragraph{Connection to Simpson's paradox.}
Simpson’s Paradox \citep{simpson1951interpretation} describes a phenomenon where a trend or relationship observed within subgroups reverses when these subgroups are aggregated. Even if a treatment appears to be effective in every subgroup, its effectiveness in the overall population is not guaranteed due to the uneven distributional effects.
In our setting, one might expect AI to (weakly) increase the overall productivity since AI directly (weakly) increases productivity within each skill state (Proposition~\ref{prop:effort_productivity_expressions}). However, skill levels are themselves endogenous. Proposition~\ref{prop:bdc_skill_dist} shows that AI leads to a downward shift in the steady-state skill distribution. As a larger proportion concentrates on the low-skill states, the overall productivity can decrease if the negative effect of skill degradation exceeds the direct productivity gains from AI assistance.
The paradox is driven by the intertemporal nature of the endogenous skill development mechanism. In fact, the connection to Simpson’s paradox can be interpreted in two complementary ways: 1) At the individual level, the steady-state captures the long-run skill distribution, and the downward shift in this distribution under AI assistance drives the productivity paradox. 2) At the population level, for heterogeneous agents evolving in parallel, the steady-state distribution can be regarded as the cross-sectional distribution, and Simpson’s paradox arises because AI shifts the mass to lower-skill states.

\paragraph{The paradox disappears under exogenous transition rates.} 
Theorem~\ref{thm:bdc_productivity} identifies a productivity paradox. This phenomenon broadly emerges in settings where skill evolves endogenously as a function of effort and the semi-elasticity of skill is sufficiently strong.
To isolate the effect of endogenous effort in skill development, we contrast with the setting where transition rates are fully exogenous. If the transition rates depend only on the state (i.e., $\lambda_k$ depends on $s_k$) rather than on effort, then a similar version of Proposition~\ref{prop:effort_productivity_expressions} still holds, i.e., productivity is increasing in $a$; see Appendix~\ref{app:exogenous_transition}.

\section{The Impact of AI Unreliability: Risk Premium}\label{sec:AI_unreliability}
In this section, we study the impact of AI unreliability. Starting from the basic model in Section~\ref{ssec:basic}, we now activate Mechanism 2, allowing AI assistance to be unreliable as described in Section~\ref{sec:model_ai_unreliability}. To isolate the effect of AI unreliability, we exclude Mechanism 1 and assume skill is static, so we write $s$ for the skill level.
Recall that we introduce a reliability level $q\in[0,1]$, such that the AI functions desirably with a proficiency $\bar{a}$ with probability $q$, i.e., $\mathbb{P}(a=\bar{a})=q$, and fails to assist with probability $1-q$, i.e., $\mathbb{P}(a=0)=1-q$. The agent has a prior belief on the AI assistance level and commits an effort level before the AI's output is realized. 
Formally, the utility-maximizing effort and the resulting productivity are:
\begin{align*}
    & e^\star(s,\bar{a},q):=\arg\max_{e\geq 0}\big(q\cdot p(e+h(s,\bar{a})) + (1-q)\cdot p(e+h(s,0))-\gamma e\big)\\
    &  p^\star(s,\bar{a},q):=q\cdot p(e^\star(s,\bar{a},q)+h(s,\bar{a})) + (1-q)\cdot p(e^\star(s,\bar{a},q)+h(s,0)).
\end{align*}
Note that when $q=1$, this model reduces to our basic model in Section~\ref{ssec:basic}, where $e^\star(s,\bar{a},1)=e^\star(s,\bar{a})$ and $p^\star(s,\bar{a},1)=p^\star(s,\bar{a})$, as characterized in Proposition~\ref{prop:effort_productivity_expressions}.

Our object of interest is the steady-state productivity $\mathcal{P}(\bar{a})$ as defined in \eqref{eq:productivity_obj}, which corresponds to $p^\star(s,\bar{a},q)$ in this case. When we account for Mechanism 2, the monotonic relationship between AI assistance and productivity in Section~\ref{ssec:basic} breaks down again. This second productivity paradox arises when the marginal benefit of increased AI proficiency is offset by the reduction of effort in the face of unreliable AI outputs.

\subsection{Main result: emergence of productivity decline region}
Our main result in this section shows that, for production processes with accelerating diminishing returns, a notion that we formalize in Definition~\ref{def:risk_averse}, an increase in AI proficiency can reduce productivity in the presence of AI unreliability. This occurs when (unreliable) AI substantially cannibalizes human effort, outweighing its direct productivity gain. In contrast, for production processes with decelerating diminishing returns, productivity increases with AI proficiency. 
To formalize this, we define the rate of diminishing marginal returns for arbitrary production functions. This notion is defined using the terminology of risk aversion.

\begin{definition}\label{def:risk_averse}
    Given a twice differentiable production function $p$, for an input level $x$, its absolute risk aversion (ARA) is $A(x):=-\frac{p''(x)}{p'(x)}$. We say that $p$ satisfies increasing (resp. decreasing) absolute risk aversion, or IARA (resp. DARA) for short, if $A(x)$ is increasing (resp. decreasing).
\end{definition}

Here we adopt terminology from standard economic literature \cite[Chap.~11]{varian1992microeconomic}. For a concave production function $p$, the absolute risk aversion $A(x)\geq 0$ measures the rate at which marginal return decreases. An increasing value of $A(x)$ means that the marginal return decays more rapidly as the input $x$ increases---IARA is a ``strong form'' of diminishing returns; while a decreasing value of $A(x)$, or DARA, corresponds to weakening marginal return decay.  
Although DARA appears more frequently in utility theory, there are also commonly used production functions in practice that (partially) satisfy IARA, as discussed in Section~\ref{sec:unreliability_discussion}. 

Intuitively, the IARA property creates situations where unreliability makes indirect cannibalization of effort by AI assistance outweigh the direct productivity gain. A productivity paradox ensues for the corresponding production functions, as stated in the next theorem, whose proof is provided in Section~\ref{sec:unreliable_technical}.

\begin{theorem}\label{thm:unreliable_concave3}
    Consider any concave production function $p(x)$ with a skill-AI interaction function $h(s,a)$ that is increasing in both arguments, and a linear cost function $c(e)=\gamma e$. Suppose  $h(s,0)<x^\star$ and the reliability $q$ satisfies $0<q<1$. 
    \begin{enumerate}
        \item If $p$ satisfies IARA, productivity $p^\star(s, \bar{a}, q)$ is initially decreasing then possibly increasing in $\bar{a}$.
        \item If $p$ satisfies DARA, productivity $p^\star(s, \bar{a}, q)$ is always increasing in $\bar{a}$. 
    \end{enumerate}
\end{theorem}
Theorem~\ref{thm:unreliable_concave3} shows that under unreliability, for the large family of production functions satisfying IARA, improving AI proficiency below a certain threshold may hurt productivity; beyond this threshold, productivity increases in AI proficiency. This is reminiscent of but distinct from the J-curve in macroeconomics (see Section~\ref{sec:unreliability_discussion} for a comparison). Framing our production properties within risk premium and expected utility theory gives further intuition on the implications of Theorem~\ref{thm:unreliable_concave3}; in Section~\ref{sec:unreliability_discussion}, we observe that an IARA production function amounts to a negative interaction between effort and a notion of {\em certainty-equivalent} AI assistance, which corresponds to an expansion of the risk premium; this explains the risk of productivity loss. Later in Lemma~\ref{lem:unreliable_productivity_derivative_origin_sign}, we characterize the productivity paradox for general production functions that may be neither IARA nor DARA, showing that it depends on the difference between two absolute risk aversion terms. 

As a complement to Theorem~\ref{thm:unreliable_concave3}, we show that the productivity ratio between higher and lower AI assistance levels can be arbitrarily small, implying an arbitrarily large productivity decline. The proof is provided in Section~\ref{sec:unreliable_technical}.
\begin{proposition}\label{prop:exante_linear_bad}
    For any $\varepsilon>0$, there exists a concave production function $p(x)$ with a perfect substitution skill-AI interaction function $h(s,a)=s+a$, a linear cost function $c(e)$, a skill level $s$, and a reliability $q\in(0,1)$, 
    such that the ratio of productivity between a higher and a lower proficiency is
    $$\inf_{a_\ell <a_h}  \dfrac{p^\star(s,a_h, q)}{p^\star(s, a_\ell, q)} \leq \varepsilon.$$
\end{proposition} 

\subsection{Technical crux: uncertainty-driven effort choice}\label{sec:unreliable_technical}
We now examine the mechanism behind potential productivity decline under unreliability and illustrate how absolute risk aversion matters.
Intuitively, the impact of AI can be decomposed into two opposite effects: implicit effort reduction and direct productivity boost. When AI is unreliable, the reduction of human effort can dominate the actual AI contribution.
To see this, we distinguish between two cases, depending on whether the agent exerts positive or zero effort. Note that the agent's effort is determined by its \emph{marginal benefit}:
$$M(s,\bar{a},q):= qp'(h(s,\bar{a})) + (1-q)p'(h(s,0))=\frac{\partial}{\partial e}\Big(qp(e+h(s,\bar a)) + (1-q)p(e+h(s,0))\Big)\big|_{e=0}.$$
The agent exerts positive effort if $M(s,\bar{a},q)>\gamma$ and exerts zero effort if $M(s,\bar{a},q)<\gamma$. The next lemma characterizes the sign of the derivative of the productivity function under the case distinction of $M$; and in the interior case, the sign depends on the gap between two absolute risk aversion measures. Below, we define $\operatorname{sgn}(\cdot)$ as the sign function.

\begin{lemma}\label{lem:unreliable_productivity_derivative_origin_sign}
    Consider any strictly concave production function $p(x)$ with a skill-AI interaction function $h(s,a)$ that is increasing in both arguments, and a linear cost function $c(e)=\gamma e$. Fix $s$ such that $h(s,0)<x^\star$, $\bar{a}>0$, and $0<q<1$.  The sign of the derivative is 
    \begin{equation*}
        \operatorname{sgn} \Big(\frac{\partial p^\star}{\partial \bar{a}} (s,\bar{a},q)\Big)=
        \begin{cases}
            \operatorname{sgn}\Big(-\big(A(e^\star (s,\bar{a},q)+h(s,\bar{a})) - A(e^\star (s,\bar{a},q)+h(s,0))\big)\Big), & \text{if } M(s,\bar{a},q) > \gamma,\\
            1, & \text{if } M(s,\bar{a},q) < \gamma.      \end{cases}
    \end{equation*}
\end{lemma}
Lemma~\ref{lem:unreliable_productivity_derivative_origin_sign} connects the sign of the derivative and the value of two absolute risk aversion terms, whose proof is deferred to Appendix~\ref{app:proof_sign_lemma}. Equipped with Lemma~\ref{lem:unreliable_productivity_derivative_origin_sign}, we now prove Theorem~\ref{thm:unreliable_concave3}.
\begin{proof}[Proof of Theorem~\ref{thm:unreliable_concave3}] 
Since $p$ is concave and $h(s,0)<x^\star$, there exists a unique threshold $\tau\in (0,\infty]$ such that the \emph{marginal benefit} $M$ is greater than $\gamma$ if and only if $\bar{a} < \tau$, i.e., $$M(s,\bar{a},q)=q p'(h(s,\bar{a}))+ (1-q) p'(h(s,0)) > \gamma \iff \bar{a} < \tau.$$ 
The threshold $\tau$ denotes a value of AI proficiency level $\bar{a}$ above which the agent exerts zero effort.
\begin{enumerate}
\item If $p$ satisfies IARA, Lemma~\ref{lem:unreliable_productivity_derivative_origin_sign} implies $\frac{\partial p^\star(s,\bar{a},q)}{\partial \bar{a}} < 0$ when $M(s,\bar{a},q) > \gamma$ (i.e., $\bar{a} < \tau$), and $\frac{\partial p^\star(s,\bar{a},q)}{\partial \bar{a}} > 0$  when $M(s,\bar{a},q) < \gamma$ (i.e., $\bar{a} > \tau$). Thus, the productivity is decreasing-then-increasing in $\bar{a}$, with the increasing region degenerate when $\tau=\infty$ (see Remark~\ref{rmk:possibly_decrease} for a discussion on the degenerate case).
\item If $p$ satisfies DARA, Lemma~\ref{lem:unreliable_productivity_derivative_origin_sign} implies that the derivative is always positive. Thus, the productivity $p^\star(s, \bar{a}, q)$ is always increasing in $\bar{a}$.
\end{enumerate}
Combining both cases, we complete the proof.
\end{proof}
\begin{remark}\label{rmk:possibly_decrease}
    As shown in the proof of Theorem~\ref{thm:unreliable_concave3}, the productivity is always decreasing when the production function satisfies IARA and $(1-q)p'(h(s,0))\geq \gamma$. In this case, the utility-maximizing effort is always positive and decreases in AI proficiency; the effort cannibalization outweighs the direct gain from unreliable AI assistance, leading to a productivity decline constantly.
\end{remark}
To establish the ratio in Proposition~\ref{prop:exante_linear_bad}, we investigate the effort and productivity for a specific piecewise linear class of production functions with perfect skill-AI substitution, whose proof is deferred to Appendix~\ref{app:proof_unreliable_linear}.
\begin{lemma}\label{lem:unreliable_linear}
    Consider production function $p(x)=\min(1, \beta x)$ with a skill-AI interaction function of $h(s,a)=s+a$, and cost function $c(e)=\gamma e$, where $\beta>\gamma$. The utility-maximizing effort and productivity are:
    \begin{align*}
    &e^\star (s,\bar{a},q)=
        \begin{cases}
            (\frac{1}{\beta}-s)^+, & \text{if } q \leq \frac{\beta-\gamma }{\beta},\\
            (\frac{1}{\beta}-s-\bar{a})^+, & \text{otherwise}.
        \end{cases}\\
        &p^\star (s,\bar{a},q)=
        \begin{cases}
            1, & \text{if } q \leq \frac{\beta-\gamma }{\beta},\\
            q + (1-q)\cdot \max\left( \min(1,\beta s),  1-\beta \bar{a}\right), & \text{otherwise}.
        \end{cases}
    \end{align*}
\end{lemma}
Lemma~\ref{lem:unreliable_linear} characterizes the effort and productivity for the piecewise linear class: when reliability $q$ is low, agents maintain full effort and achieve high productivity; when reliability $q$ is high, agents reduce their effort, leading to a decline in productivity. The phase transition happens at $q=\frac{\beta-\gamma}{\beta}$, where the agent's trust in the AI assistance flips, leading to a discontinuous jump in effort. This specification provides an abstraction of how AI reliability influences effort.
The piecewise linear instance establishes the arbitrarily small ratio in Proposition~\ref{prop:exante_linear_bad} and we restrict to this instance in Section~\ref{sec:AI_literacy}.
\begin{proof}[Proof of Proposition~\ref{prop:exante_linear_bad}]
Consider $s=0$, $p(x)=\min(1, \beta x)$, $h(s,a)=s+a$, and $c(e)=\gamma e$, where $\beta >\gamma $. Let $q> \frac{\beta -\gamma }{\beta}$. Lemma~\ref{lem:unreliable_linear} implies that $p^\star (0,\frac{1}{\beta}, q)=q$ and $p^\star (0, 0, q)=1$. Thus, by selecting two proficiency levels $a_\ell=0$ and $a_h=\frac{1}{\beta}$, the productivity ratio is given by $\frac{p^\star (0,\frac{1}{\beta}, q)}{p^\star (0, 0, q)}=q$, which can be made equal to $\varepsilon$ by choosing $\beta, \gamma, q$ such that $\frac{\beta -\gamma }{\beta}=\frac{\varepsilon}{2}$ and $q=\varepsilon$. 
\end{proof} 
\begin{remark}\label{rmk:relaxed_IARA}
    Although the standard notion of IARA requires strict monotonicity, a similar version of Theorem~\ref{thm:unreliable_concave3} holds for a relaxed notion of IARA, where the production function is capped. In particular, this allows for a class of perturbed piecewise linear production functions. We detail this extension in Appendix~\ref{app:proof_unreliable_relaxation}. Moreover, a similar arbitrarily small ratio holds with a slight perturbation of the piecewise linear instance that also satisfies the relaxed IARA condition; see Appendix~\ref{app:exante_linear_bad}.
\end{remark}

\subsection{Discussion}\label{sec:unreliability_discussion}
\paragraph{Productivity versus reliability.}  In reality, GenAI models have joint upgrades, advancing in both reliability and proficiency. In this section, we focus on investigating how AI proficiency affects productivity, given that AI is unreliable. Interestingly, if we regard the unreliability of AI as a mechanism and hold proficiency fixed, increasing reliability does not necessarily increase productivity. Lemma~\ref{lem:unreliable_linear} makes this clear: when reliability $q$ is smaller than $\frac{\beta-\gamma}{\beta}$, the productivity equals $1$; once the reliability crosses the threshold, the productivity drops immediately. Hence, an improvement in AI reliability can decrease productivity. In general, the derivative of steady-state productivity with respect to $q$ is 
$$\frac{\partial p^\star }{\partial q}(s,\bar{a},q) = p(e^\star +h(s,\bar{a})) - p(e^\star+h(s,0) ) + \gamma \frac{\partial e^\star}{\partial q}.$$ 
Again, it decomposes into an indirect effort reduction effect and a direct productivity boost effect by AI.

\paragraph{IARA versus DARA functions.} 
IARA and DARA describe different patterns of diminishing returns in the production function, distinguishing cases in which AI is potentially harmful or always beneficial.\footnote{Studies in agricultural economics also make the contrast between IARA and DARA \citep{leathers1991interactions}: for risk-increasing agricultural resources (e.g., fertilizer), an increase in the price decreases usage under DARA but increases it under IARA utilities.}
In utility theory, examples of DARA include power law and logarithmic functions, $p(x)=x^{\frac{1}{1+c}}$ and $p(x)=\log(c x+1)$, where $c>0$; examples of IARA include Gaussian function $p(x)=\int_0^x \exp(-t^2)dt$ \cite[Chap.~14]{abbas2018foundations}, expo-power function $p(x)=1-\exp(-b x^c)$ for $b>0, c>1$ \citep{saha1993expo}, and the relaxed IARA includes truncated quadratic function $p(x)=-c_2x^2+c_1 x$ for $x\leq \frac{c_1}{2c_2}$ \citep{vcerny2012computation}. 
In utility theory, DARA appears more common than IARA, yet production economics more naturally admits IARA. Among commonly used production functions, the translog \citep{berndt1973translog} and transcendental multiproduct function \citep{mundlak1964transcendental} are IARA under certain domains and parameter specifications (see Appendix~\ref{app:IARA_examples}). 

\paragraph{Connection to increasing hazard rate.} IARA can also be interpreted via an equivalence to increasing hazard rate (IHR). Consider the production function $p(x)=\mathbb E [\min(X,x)]$, where $X\geq 0$ is a random variable, with a cumulative distribution function $F(x)$ and a probability density function $f(x)$. The associated hazard rate is then $H(x)=\frac{f(x)}{1-F(x)}=-\frac{p''(x)}{p'(x)}$. Thus, under this specification, the IARA condition of $p(x)=\mathbb E [\min(X,x)]$ is mathematically equivalent to IHR of $X$, which holds for various distributions such as uniform, Gamma, and truncated Gaussian \citep{bagnoli2005log}. Under this lens, uniform distributions can be mapped to truncated quadratic functions, and Gaussian distributions to Gaussian functions.

\paragraph{Comparison with the J-curve in macroeconomics.} We identify either a decreasing or a decreasing-then-increasing V-shaped pattern in the relationship between productivity and AI proficiency, given production functions satisfying IARA.
The result in our setting should not be confused with the J-curve in technology adoption in macroeconomics \citep{brynjolfsson2021productivity}, which describes transient adaptation of economic factors to fixed productivity shocks. The J-curve depicts a temporal phenomenon where technology initially depresses productivity growth, while later productivity rises as early intangible investments begin to yield returns.
In our setting, we show that, for production functions satisfying IARA, unreliability can create a similar decrease-then-increase behavior with respect to technology adoption (or AI proficiency in our setting). That said, our result sheds light on a distinct mechanism: productivity can decrease because effort cannibalization dominates the direct benefit of AI assistance when the technology is unreliable. Using a micro-founded model of productivity, this shows how the adoption of immature AI technology could backfire; the productivity decline occurs exactly under accelerating diminishing returns of production.

\paragraph{Connection to risk premium and expected utility theory.} The analysis of this section depends on absolute risk aversion, a standard concept in expected utility theory \citep{varian1992microeconomic}. In this framework, an individual has an expected utility $\mathbb E [U(x)]$, where $x$ represents the uncertain outcome being consumed. The uncertain outcome under risk can be associated with a \emph{certainty-equivalent} value $x_{CE}$, where $U(x_{CE})=\mathbb E [U(x)]$. The absolute \emph{risk premium} is defined as the difference between the expected value and the certainty-equivalent value, which measures the cost of risk, or intuitively, the amount of input one gives up to achieve certainty. Building upon the Arrow-Pratt approximation \cite[Chap.~3.4]{gollier2001economics}, the risk premium $\mathbb E [x] - x_{CE} \approx \frac{1}{2}A(\mathbb E[x]) \operatorname{Var}(x)$ (Figure~\ref{fig:tikz_expected_utility_a}).
The random input $x=e+h(s,a)$ can be associated with a certainty-equivalent value $x_{CE}=e+h(s,a_{CE})$ defined as $p(x_{CE})=\mathbb E [p(x)]$, where $a_{CE}$ is the certainty-equivalent AI assistance. The absolute risk premium is approximated by $qh(s,\bar{a})+(1-q)h(s,0) - h(s,a_{CE}) \approx \frac{1}{2}A\big(e+qh(s,\bar{a})+(1-q)h(s,0)\big) q(1-q)\big(h(s,\bar{a})-h(s,0)\big)^2$ (Figure~\ref{fig:tikz_expected_utility_b} shows the special case with $h(s,a)=s+a$, where the risk premium is approximated by $q\bar{a}-a_{CE}\approx \frac{1}{2}A\big(e+s+q\bar{a}\big) q(1-q)\bar{a}^2$). Thus, under the IARA condition, increasing the effort $e$ raises the absolute risk aversion, enlarges the risk premium, and thus reduces the certainty-equivalent AI assistance $a_{CE}$. Because exerting more effort diminishes the effective contribution of AI, this negative interaction can lead to a decline in total productivity. Note that this is not merely a substitution effect of AI; rather, it reduces the equivalent contribution of AI itself. In contrast, the negative interaction between AI and effort, and the productivity paradox, do not arise in the DARA case.\footnote{The certainty-equivalent representation is purely expository. It is a reinterpretation of realized expected productivity and should not be confused with our technical analysis.}

\begin{figure}[H]
    \centering
    \begin{minipage}[b]{0.48\textwidth}
        \centering
        \begin{tikzpicture}[scale=0.7]
        \draw[->] (0,0) -- (6,0) node[right] {$x$};
        \draw[->] (0,0) -- (0,4) node[left] {$U$};
        \draw[thick, black]
          plot[smooth, domain=0:5]
          (\x, {1.7*sqrt(\x)}); 
        
        \draw[dashed] (2.3,0) -- (2.3,{1.7*sqrt(2.3)});
        \draw[dashed] (0, {1.7*sqrt(2.3)}) -- (2.3,{1.7*sqrt(2.3)});
        \node[below] at (2.3,0) {$x_{CE}$};
        \draw[dashed] (4,0) -- (4,{1.7*sqrt(4)});
        \draw[dashed] (0, {1.7*sqrt(4)}) -- (4,{1.7*sqrt(4)});
        \node[below] at (4,0) {$\mathbb E [x]$};
        
        \node[left] at (0,{1.7*sqrt(2.3)}) {$\mathbb E [U(x)]$};
        \node[left] at (0,{1.7*sqrt(4)}) {$U(\mathbb E [x])$};
        
        \filldraw[black] (2.3, {1.7*sqrt(2.3)}) circle (2pt);
        \filldraw[black] (4, {1.7*sqrt(4)}) circle (2pt);
        
        \filldraw[black] (2.3, 0) circle (2pt);
        \filldraw[black] (4, 0) circle (2pt);
        
        \draw[decorate,decoration={brace, amplitude=4pt, mirror}]
          (2.3,-0.5) -- (4,-0.5)
          node[midway,below=4pt] {Risk Premium $\approx \frac{1}{2}A(\mathbb E[x]) \operatorname{Var}(x)$};
        
        \end{tikzpicture}
        \subcaption[]{Expected utility theory: utility $U(x)$}
        \label{fig:tikz_expected_utility_a}
    \end{minipage}
    \hfill
    \begin{minipage}[b]{0.48\textwidth}
        \centering
        \begin{tikzpicture}[scale=0.7]
        \draw[->] (0,0) -- (6,0) node[right] {$x$};
        \draw[->] (0,0) -- (0,4) node[left] {$p$};
        \draw[thick, black]
          plot[smooth, domain=0:5]
          (\x, {1.7*sqrt(\x)}); 
        
        \draw[dashed] (2.3,0) -- (2.3,{1.7*sqrt(2.3)});
        \draw[dashed] (0, {1.7*sqrt(2.3)}) -- (2.3,{1.7*sqrt(2.3)});
        \node[below] at (2.3,0) {$x_{CE}$};
        \draw[dashed] (4,0) -- (4,{1.7*sqrt(4)});
        \draw[dashed] (0, {1.7*sqrt(4)}) -- (4,{1.7*sqrt(4)});
        \node[below] at (4,0) {$\mathbb E [x]$};
        
        \node[left] at (0,{1.7*sqrt(2.3)}) {$\mathbb E [p(x)]$};
        \node[left] at (0,{1.7*sqrt(4)}) {$p(\mathbb E [x])$};
        
        \filldraw[black] (2.3, {1.7*sqrt(2.3)}) circle (2pt);
        \filldraw[black] (4, {1.7*sqrt(4)}) circle (2pt);
        
        \filldraw[black] (2.3, 0) circle (2pt);
        \filldraw[black] (4, 0) circle (2pt);
        \node[left] at (2.4, 0.2) {\scriptsize $e+s+a_{CE}$};
        \node[right] at (4, 0.2) {\scriptsize $e+s+q\bar{a}$};
        
        \draw[decorate,decoration={brace, amplitude=4pt, mirror}]
          (2.3,-0.5) -- (4,-0.5)
          node[midway,below=4pt] {Risk Premium $\approx \frac{1}{2}A\big(e+s+q\bar{a}\big) q(1-q)\bar{a}^2$};
        \end{tikzpicture}
        \subcaption[]{Our model: production $p(x)$, where $h(s,a)=s+a$}
        \label{fig:tikz_expected_utility_b}
    \end{minipage}
    \vspace{0.1cm}
    \caption{Connection to risk premium and expected utility theory}
    \label{fig:tikz_expected_utility}
\end{figure}

\paragraph{The paradox disappears under full adaptation ability.} Theorem~\ref{thm:unreliable_concave3} gives a productivity paradox. It emerges in settings when the agent optimizes the effort ex-ante, given accelerating diminishing returns. This arises from the lack of adaptation under unreliability, and here we contrast with the setting where the agent perfectly evaluates the AI output after its realization. If the agent chooses the utility-maximizing effort based on the actual outcome, then a similar version of Proposition~\ref{prop:effort_productivity_expressions} still holds, i.e., productivity is increasing in AI proficiency $\bar{a}$; see Appendix~\ref{app:full_adapt}. 
This comparison offers an initial view of AI literacy---an adaptivity based on observation of AI outputs. In the next section, we study AI literacy more formally.

On the other hand, the productivity paradox carries over when the agent has partial adaptation ability, i.e., when the agent is a Bayesian signal follower as in the model of Section~\ref{sec:AI_literacy}. In that case, the productivity is a convex combination of the productivities under two posterior reliabilities, so the initial productivity decrease region in Theorem~\ref{thm:unreliable_concave3} carries over to Bayesian signal followers as well.

\section{The Impact of AI Literacy: Skill Polarization}\label{sec:AI_literacy}
So far, we consider two mechanisms---skill development and AI unreliability---in isolation and study whether and how each of them entails a productivity paradox. In reality, both mechanisms may occur simultaneously and reinforce each other; our previous results therefore imply the emergence of a productivity paradox in a more general model in which both mechanisms operate. A critical assumption of our framework concerns how individuals make use of the AI outputs: we assume that the agent paces their effort ex-ante before the stochastic AI assistance level is realized. However, highly skilled individuals may adapt their effort to the ex-post quality of AI outputs. 
This is captured by the notion of \emph{AI literacy} which corresponds to a set of competencies that enable individuals to collaborate efficiently with and critically evaluate AI tools \citep{long2020ai}. Digital literacy and evaluation ability are known bottlenecks for AI users to achieve productive outcomes \citep{passi2022overreliance}. In this section, we incorporate a notion of AI literacy into our model. Consequently, we analyze the distributional effect of AI on effort across various skill levels, and reciprocally uncover how effort shapes skill development. This drives our main finding: heterogeneous AI literacy leads to multimodal skill distributions in the long run, intuitively increasing the {\em dispersion} of human capital. 

We now present our holistic model. 
The components of Mechanism 1 (endogeneity of skills) and Mechanism 2 (unreliability of AI) remain unchanged. As before, let $S=(s_1,s_2,\cdots,s_N)$ denote the skill states, where $s_1 < s_2 < \cdots < s_N$, and let $a$ denote the (unreliable) AI assistance level, where $\mathbb{P}(a=\bar{a})=q$ and $\mathbb{P}(a=0)=1-q$. In this context, we assume that the agent can review the AI output and adapt its effort accordingly.
Specifically, the agent receives a noisy signal $\hat{a}$ of the true assistance level $a$, representing the individual's verification of the AI output. The accuracy of this verification depends on the agent's current skill $s$, as given by $\mathbb P(\hat{a}=a|s)=v(s)$, where $v(s)\in [1/2,1]$ is an increasing, concave, and differentiable function. Verification ability can be interpreted as the agent's AI literacy, which is positively correlated with the implementation skill $s$. We assume $v(0)=1/2$, so an individual with zero skill performs no better than random guessing.  After observing the signal, the agent determines their effort following Bayes rule. We capture this behavior by introducing the notion of a \emph{Bayesian signal follower}.
\begin{definition}\label{def:bayes}
    The posterior reliabilities of a Bayesian signal follower with skill $s$ are
    \begin{align*}
    & q_{\bar{a}}(s):=\mathbb P(a=\bar{a}|\hat{a}=\bar{a},s)=\frac{q v(s)}{q v(s)+(1-q)(1-v(s))}, \\
    & q_0(s):=\mathbb P(a=\bar{a}|\hat{a}=0,s)=\frac{q (1-v(s))}{q(1-v(s))+(1-q)v(s)}.
    \end{align*}
    A Bayesian signal follower chooses effort based on the posterior belief about the assistance level $a$. Specifically, when $\hat{a}=\bar{a}$, the agent exerts an effort $e^\star(s,\bar{a},q_{\bar{a}}(s))$; when $\hat{a}=0$, the agent exerts effort $e^\star(s,\bar{a},q_0(s))$. 
\end{definition}
The signal distribution satisfies $\mathbb P(\hat{a}=\bar{a}|s)=q \cdot v(s) + (1-q) \cdot (1-v(s))$ and $\mathbb P(\hat{a}=0|s) =q \cdot (1-v(s)) + (1-q) \cdot v(s)$. The expected effort and productivity of a Bayesian signal follower are
\footnote{We use the subscript $v$ to denote the verification function. When $v(s)=\frac{1}{2}$, i.e., the signal is uninformative and we recover the model of Section~\ref{sec:basic} where the effort of a Bayesian signal follower $e^\star_v$ reduces to $e^\star$.} 
\begin{align*}
    e^\star_v(s, \bar{a}, q)&:= \mathbb P(\hat{a}=\bar{a}|s)e^\star(s,\bar{a},q_{\bar{a}}(s))+\mathbb P(\hat{a}=0|s) e^\star(s,\bar{a},q_0(s)),\\
     p^\star_v(s, \bar{a}, q)&:= \mathbb P(\hat{a}=\bar{a}|s)p^\star(s,\bar{a},q_{\bar{a}}(s)) +\mathbb P(\hat{a}=0|s) p^\star(s,\bar{a},q_0(s)).
\end{align*}
The transition rates are given by 
$\mathcal{K}_{s_k s_{k+1}}(\bar{a},q)= \lambda( e^\star_v(s_k, \bar{a}, q))$ and $\mathcal{K}_{s_k s_{k-1}}(\bar{a},q) = \mu$, where $\mu>0$ and $\lambda(\cdot)>0$ is an increasing function.
To ease our analysis, throughout this section, we restrict attention to the piecewise differentiable production function $p(x)=\min(1, \beta x)$ with a perfect substitution skill-AI interaction function $h(s,a)=s+a$. The
cost function remains $c(e)=\gamma e$. To avoid settings where effort is never worth it for the agents, we further assume that $\beta> \gamma$. 

\subsection{Warm-up: Example of  AI literacy extremes}\label{sec:literace_example}
To highlight heterogeneous AI literacy as the driver behind skill polarization, we first analyze the long-run skill distribution when the AI literacy is homogeneous, i.e., agents have identical verification ability, $v(s)=v\geq \frac{1}{2}$ for any $s$. This includes as special cases the model with only Mechanism 1 activated when $v\geq \frac{1}{2}$ and $q=1$, and the model with both Mechanisms 1 and 2 activated when $v=\frac{1}{2}$ and $q\in (0,1)$. Our previous results imply that, under any such homogeneous AI-literacy setting, the resulting skill distribution is unimodal, i.e., concentrated around a single peak (see Appendix~\ref{app:literacy_constant_v}).

Before the general analysis on heterogeneous AI-literacy, we illustrate the mechanism on a simple example with two extremes of AI literacy: $v=\frac{1}{2}$ (no ability to verify the AI output) and $v=1$ (perfect verification).
\begin{example}\label{ex:two_extreme_verify}
    Consider $p(x)=\min(1,x)$ and $c(e)=\gamma e$ with $0<\gamma<1$, AI proficiency level $\bar{a}=1$ and reliability $0<q<1$, and two skill levels $0<s_L<s_H<1$ with $v(s_L)=\frac{1}{2}$ and $v(s_H)=1$. Since $\bar{a}=1$, we have $p'(s+\bar{a})=0$ and $p'(s)=1$ for both skill levels. A Bayesian signal follower with skill $s$ who observes the signal $\hat{a}$ exerts positive effort if and only if
    $$\mathbb P(a=\bar{a}|\hat{a},s)\cdot \underbrace{p'(s+\bar{a})}_{=0} + \mathbb P(a=0|\hat{a},s)\cdot \underbrace{p'(s)}_{=1} > \gamma,\quad \text{i.e.} \quad \mathbb P(a=0|\hat{a},s)>\gamma$$
    \begin{itemize}
    \item For the low skill type, the signal is uninformative, so the posterior is the same as the prior: $\mathbb P(a=0|\hat{a},s_L)=1-q$ for either $\hat{a}=\bar{a}$ or $\hat{a}=0$. The agent exerts positive effort if and only if $1-q>\gamma$.
    \item For the high skill type, verification is perfect, so the posterior is the observed empirical distribution $\mathbb P(a=0|\hat{a},s_H)=\mathbb I(\hat{a}=0)$. The agent exerts positive effort if and only if $\hat{a}=0$, which is an event with probability $1-q$.
    \end{itemize}
    Consequently, when $1-q<\gamma<1$, only the high-skill agent exerts effort. In particular,
\begin{equation*}
        e^\star_v(s_L,\bar{a},q)=0 \quad\text{and}\quad e^\star_v(s_H,\bar{a},q)=(1-q)(1-s_H)>0 .
    \end{equation*}
\end{example}
Example~\ref{ex:two_extreme_verify} is a simplified setting with two skill levels. It nonetheless captures the main mechanism that drives our skill polarization result---an effort-skill reversal---high-skill agents can expend more effort because their efforts are better targeted, owing to their superior AI literacy. This reverses the substitution pattern between effort and skill in the homogeneous-verification model.
This reversal then breaks the unimodality of the steady-state skill distribution because higher effort in high-skill states and lower effort in low-skill states modulate the upward transition rates in a way that resembles a `spread' operation on the stationary distribution. To see why, consider an additional intermediate state $s_M\in (s_L,s_H)$ with $v(s_M)=1$. Hence,
\begin{equation*}
    e^\star_v(s_L,\bar{a},q)=0 \quad\text{and}\quad e^\star_v(s',\bar{a},q)=(1-q)(1-s')>0 \text{ for } s'\in\{s_M, s_H\}.
\end{equation*}
Consider a skill decay rate $\mu$ such that $\lambda(0)<\mu<\lambda\big((1-q)(1-s_M)\big)$, the expression of steady-state distribution in Equation~\eqref{eq:steady_state_prob} then implies that $\frac{\pi_2}{\pi_1}=\frac{\lambda(0)}{\mu}<1$ and $\frac{\pi_3}{\pi_2}=\frac{\lambda((1-q)(1-s_M))}{\mu}>1$. Thus, the steady state is multimodal, with the population concentrating on the low and high skill states compared to the intermediate skill state, which connects to the pattern of skill polarization (Definition~\ref{def:skill_polarization}). We now show that the insights from this example can be generalized, and we characterize when skill polarization arises.

\subsection{Main result: emergence of multimodal skill distribution}\label{sec:literacy_result}
Our main result shows that endogenous AI literacy can lead to a multimodal skill distribution; in contrast, without AI literacy, we know the skill distribution to be unimodal (Appendix~\ref{app:literacy_constant_v}). However, this multimodality hinges on a critical two-pronged condition, which we introduce next.

To facilitate the subsequent analysis we define the {\em marginal return gap of effort and AI} as follows
$$\omega= \underbrace{(\beta-\gamma)}_{\substack{\text{marginal return} \\ \text{to effort}}}- \underbrace{q \beta}_{\substack{\text{marginal return} \\ \text{to AI assistance}}}.$$
Intuitively, a positive value of $\omega$ represents an effort-critical scenario, where a unit of effort yields a larger gain than a unit of AI; while a negative value of $\omega$ represents an AI-critical scenario. Formally, Lemma~\ref{lem:unreliable_linear} already analyzes the piecewise linear function and identifies two regimes depending on $\omega$: $\omega\geq 0$ corresponds to a high utility-maximizing ex-ante effort, whereas $\omega<0$ corresponds to a low utility-maximizing ex-ante effort.

Moreover, when agents are Bayesian signal followers, the sign of $\omega$ also distinguishes two cases for the posterior beliefs. Recall that $v(0)=1/2$ and $0<q<1$, and thus, the posterior reliability $q_{\bar{a}}(s)$ increases as a function of $s$ while $q_0(s)$ decreases as a function of $s$, with initial starting values corresponding to the uninformed prior belief on reliability $q_{\bar{a}}(0)= q_0(0)=q$.  Consider the threshold $\frac{\beta-\gamma}{\beta}$ on reliability, previously defined in Lemma~\ref{lem:unreliable_linear} as the point of discontinuity at which the agent's trust in AI flips, causing the agent to drop their effort as reliability increases and crosses that threshold. Based on preceding observations, exactly one of the two reliability functions $q_{\bar{a}}, q_0$ crosses the threshold $\frac{\beta-\gamma}{\beta}$ and the sign of $\omega$ determines which one of them. That is, (a) If $\omega<0$ (i.e., $q>\frac{\beta-\gamma}{\beta}$), then $q_0(\cdot)$ crosses the threshold $\frac{\beta-\gamma}{\beta}$ from above; and (b) if $\omega\geq 0$ (i.e., $q\leq \frac{\beta-\gamma}{\beta}$), then $q_{\bar{a}}(\cdot)$ crosses $\frac{\beta-\gamma}{\beta}$ from below. We define the {\em critical skill value} $\Tilde{s}$ at the crossing point in each of these cases as follows:\footnote{$v$ is an increasing function. We define $v^{-1}(y)=\infty$ when $y$ is not achievable by $v$.}
\begin{equation*}
    \Tilde{s} := v^{-1}\left(\dfrac{\frac{q}{1-q}\frac{\gamma}{\beta -\gamma}}{\frac{q}{1-q}\frac{\gamma}{\beta -\gamma}+1}\right) \mathbb I(\omega < 0) + v^{-1}\left(\dfrac{1}{\frac{q}{1-q}\frac{\gamma}{\beta -\gamma}+1}\right) \mathbb I(\omega \geq 0) \ .
\end{equation*}
The functions $q_{\bar{a}}$, $q_0$, and the crossing points are visualized in Figure~\ref{fig:paradox_cases}. 
\begin{figure}[ht]
    \centering
    \begin{minipage}{0.45\textwidth}
        \centering
        \includegraphics[width=\linewidth]{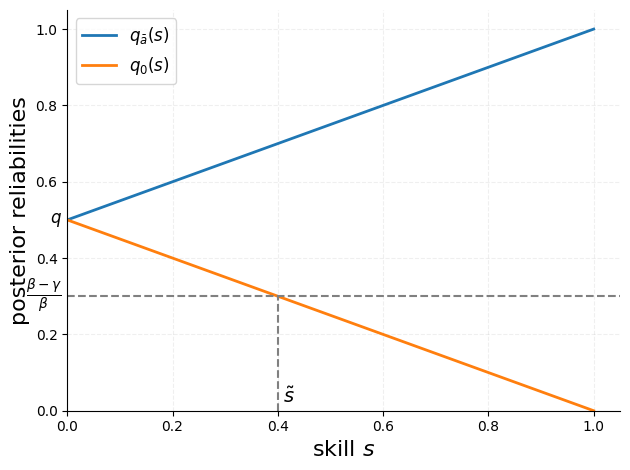}
        \subcaption[]{$\omega< 0$: $q_0(\cdot)$ crosses  $\frac{\beta-\gamma}{\beta}$ at $\Tilde{s}= v^{-1}(\frac{\frac{q}{1-q}\frac{\gamma}{\beta -\gamma}}{\frac{q}{1-q}\frac{\gamma}{\beta -\gamma}+1})$, $\beta =1, \gamma=0.7, q=0.5, v(s)= 0.5(1 + s)$}
        \label{fig:paradox_cases_A}
    \end{minipage}
    \hfill
    \begin{minipage}{0.45\textwidth}
        \centering
        \includegraphics[width=\linewidth]{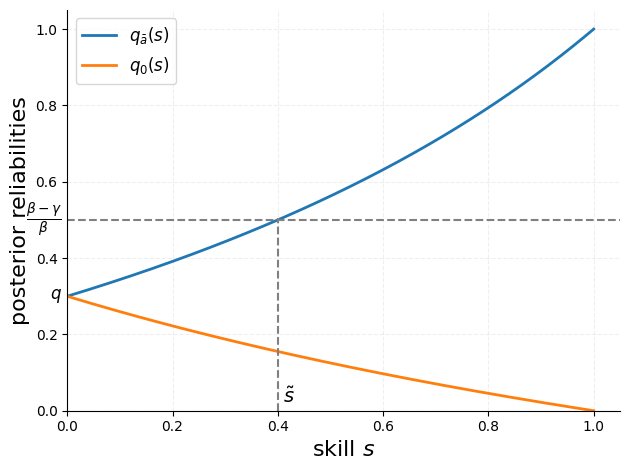}
        \subcaption[]{$\omega\geq 0$: $q_{\bar{a}}(\cdot)$ crosses $\frac{\beta-\gamma}{\beta}$ at $\Tilde{s}= v^{-1}(\frac{1}{\frac{q}{1-q}\frac{\gamma}{\beta -\gamma}+1})$, $\beta =1, \gamma=0.5, q=0.3, v(s)= 0.5(1 + s)$}
        \label{fig:paradox_cases_B}
    \end{minipage}
    \caption{Posterior reliabilities $q_{\bar{a}}$ and $q_0$ versus skill $s$}
    \label{fig:paradox_cases}
\end{figure}
With these definitions at hand, we are ready to state the sufficient condition that governs the emergence of a multimodal distributional effect of AI.
\begin{condition}\label{cond:multimodal_cond}
   Suppose that the critical skill value is upper bounded by $\Tilde{s} <\frac{1}{\beta} - \frac{(1-2q) v(\Tilde{s}) + q}{(1-2q) v'(\Tilde{s})}\cdot \mathbb I(\omega\geq 0)$.
   In addition, we require either (a) $\omega<0$, or (b) $\omega\geq 0$ and $q<1/2$.
\end{condition}

We explain Condition~\ref{cond:multimodal_cond} in each case, depending on the sign of $\omega$. In both cases (a) and (b), the multimodal skill distribution proceeds from an effort-skill reversal pattern like in the warm-up example, and our condition requires the critical skill $\tilde{s}$, i.e., the skill level at which the agent's trust in the AI assistance flips and results in a discontinuity, to be small enough. Case (a), where $\omega<0$, boils down to a simple upper bound requirement $\Tilde{s}<\frac{1}{\beta}$, indicating that the critical skill value $\tilde{s}$ occurs before the skill level $s=\frac{1}{\beta}$, beyond which no effort needs to be exerted as the production function saturates. Here, as skill $s$ increases and crosses $s=\tilde{s}$, the discontinuous positive jump of effort occurs on the branch $q_0(s)$; so an effort-skill reversal already occurs and no further requirement is needed. Case (b), where $\omega\geq 0$, is more complex to interpret, because at $s=\tilde{s}$, the discontinuous negative drop of effort occurs on the branch $q_{\bar{a}}(s)$ where reliability increases with $s$; as such, this discontinuity is not the source of skill-effort reversal and a more nuanced mechanism is at play. First, we further require $q<\frac{1}{2}$ that places us in a low reliability environment. Second, notice that the expected effort combines the effort levels for the two signals $\hat{a} = 0$ and $\hat{a} = \bar{a}$ based on their mixture; hence, a sufficient shift of the posterior belief towards the signal $\hat{a}=0$ could, in principle, cause the expected effort to increase with $s$. This comes with an additional term $\frac{(1-2q) v(\Tilde{s}) + q}{(1-2q) v'(\Tilde{s})}$ that increases with $\Tilde{s}$, implying an even tighter upper bound on $\Tilde{s}$. The introduction of this additional term stems from a sufficient shift of posterior belief towards the signal $\hat{a} = 0$. 

Under Condition~\ref{cond:multimodal_cond}, multimodality in skill distribution can emerge as stated next.
\begin{theorem}\label{thm:bayes_skill_multimodal}
Consider $p(x)=\min(1, \beta x)$, $c(e)=\gamma e$, where $\beta>\gamma$, and $\lambda(e)$ is an increasing transition function. Suppose $v(s)\in [1/2,1]$ is increasing and concave with $v(0)=\frac{1}{2}$, and the reliability satisfies $0<q<1$.
\begin{enumerate}
\item If Condition~\ref{cond:multimodal_cond} holds, then there exists a threshold $\Tilde{a}\geq 0$ such that: if $\bar{a}>\Tilde{a}$, there exists $\mu>0$ and $S \in {\cal S}$, such that the steady-state skill distribution $\boldsymbol{\pi}(\mu,\bar{a},q,S)$ is multimodal; whereas if $\bar{a}\leq \Tilde{a}$, the steady-state skill distribution $\boldsymbol{\pi}(\mu, \bar{a},q,S)$ is unimodal for all $\mu>0$ and $S\in {\cal S}$.\footnote{In Case (a) of Condition~\ref{cond:multimodal_cond}, the threshold $\Tilde{a}$ on AI proficiency is degenerate, i.e., $\Tilde{a}= 0$.}
\item Otherwise, when Condition~\ref{cond:multimodal_cond} does not hold, the steady-state skill distribution $\boldsymbol{\pi}(\mu, \bar{a},q,S)$ is unimodal for all $\bar{a}\geq 0$, $\mu>0$, and $S\in {\cal S}$.
\end{enumerate}
\end{theorem}
Theorem~\ref{thm:bayes_skill_multimodal} implies that the adoption of proficient AI leads to multimodal steady-state skill distribution, meaning intuitively that the population's skill is more dispersed. The result is consistent with heterogeneous impacts of AI assistance in various empirical studies; for example, \citep{eames2026computer} find that high-achieving students benefit more from computer-assisted learning as they spend more time and master more skills. Our model uncovers a similar mechanism: high-skill agents may amplify their advantage by collaborating with AI through appropriate judgment (i.e., output verification), while the low-skill agents tend to over-rely on AI outputs and thus degrade their skills. Consequently, AI assistance has the potential to exacerbate inequality; in Section~\ref{sec:literacy_discussion}, we relate this divergence to the Matthew effect.

\subsection{Technical crux: bidirectional influence of effort and skill}\label{sec:literacy_technical}
We now examine the mechanism behind the skill multimodality. We first analyze the static model and establish the relationship between effort and skill. We then show that, under skill dynamics, the driver of multimodal steady-state skill distribution is exactly this effort-skill reversal.

\paragraph{Effort-skill reversal.} Building upon the intuition behind Condition~\ref{cond:multimodal_cond}, we formalize the reversal pattern of effort dependence on skill in the following lemma, whose proof is deferred to Appendix~\ref{app:proof_effort_caseA+B}.
\begin{lemma}\label{lem:effort_caseA+B}
Consider $p(x)=\min(1, \beta x)$, $c(e)=\gamma e$, where $\beta>\gamma$. Suppose $v(s)\in [1/2,1]$ is increasing and concave with $v(0)=\frac{1}{2}$, and the reliability satisfies $0<q<1$.
\begin{enumerate}
\item If Condition~\ref{cond:multimodal_cond} holds, there exists a threshold $\Tilde{a}\geq 0$ such that: if $\bar{a}>\Tilde{a}$, the expected effort $ e^\star_v(s,\bar{a},q)$ is non-monotonic in $s$; whereas if $\bar{a}\leq \Tilde{a}$, effort $ e^\star_v(s,\bar{a},q)$ is decreasing in $s$.
\item Otherwise, when Condition~\ref{cond:multimodal_cond} does not hold, for any $\bar{a}$, effort $ e^\star_v(s,\bar{a},q)$ is decreasing in $s$.
\end{enumerate}
\end{lemma}

\paragraph{Reversal as a driver of multimodality.}
Next, we translate the effort-skill reversal into the multimodality of the steady-state skill distribution in two lemmas below.

When effort $ e^\star_v(s,\bar{a},q)$ is decreasing in $s$, the skill distribution remains unimodal, as stated next. 
\begin{lemma}\label{lem:bayes_effort_decrease}
    Suppose $\lambda(e)$ is increasing. Fix $\bar{a}$ and $q$, if $e^\star_v(s,\bar{a},q)$ is decreasing in $s$, then for any $\mu>0$ and $S\in {\cal S}$, the steady-state distribution  $\boldsymbol{\pi}(\mu, \bar{a},q,S)$ is unimodal.
\end{lemma}
The proof of Lemma~\ref{lem:bayes_effort_decrease} is deferred to Appendix~\ref{app:pf_skill_unimodal_case}. In contrast, when the effort $ e^\star_v(s,\bar{a},q)$ is non-monotonic in $s$, the multimodal skill distribution can emerge, as stated next.
\begin{lemma}\label{lem:bayes_effort_nonmonotonic}
    Suppose $\lambda(e)$ is increasing. Fix $\bar{a}$ and $q$, if  $ e^\star_v(s,\bar{a},q)$ is non-monotonic in $s$, then there exists $\mu>0$ and $S\in {\cal S}$, such that $\boldsymbol{\pi}(\mu, \bar{a},q,S)$ is multimodal.
\end{lemma}
\begin{figure}[H]
    \centering
    \begin{minipage}[b]{0.45\textwidth}
        \centering
        \begin{tikzpicture}[scale=0.9]
        \draw[->, thick] (-0.5, -1) -- (-0.5, 1.2) 
            node[left] {$r(s;\mu)$};
        \draw[->, thick] (-0.5, -1) -- (2.6, -1) 
            node[right] {$s$};
    
        \draw[thick, black, domain=-0.2:2.3, samples=100] 
            plot (\x, {sin(0.5*pi*\x r)});
    
        \draw[blue, thick, dashed] (-0.5,0.4) -- (2.5,0.4);
        \node[black, left] at (-0.5,0.4) {1};
    
        \draw[blue, dashed] (0.3,-1) -- (0.3,0.4);
        \draw[blue, dashed] (1.7,-1) -- (1.7,0.4);
        
        \node[black, below] at (0.1,-1) {$\mathbf S_1$};
        \node[black, below] at (1.1,-1) {$\mathbf S_2$};
        \node[black, below] at (2.2,-1) {$\mathbf S_3$};
        
        \draw[<->, thick] (0,-1.6) -- (2.3,-1.6);
        \node[black, below] at (1.2,-1.6) {$\mathbf S$};
    \end{tikzpicture}
        \subcaption[]{The figure plots the transition ratio $r(s;\mu)$, which is unimodal. The horizontal dashed line at the unit ratio
        partitions the s-axis into sub-segments ${\mathbf S}_1 , {\mathbf S}_2 , {\mathbf S}_3$}.
                \label{fig:modality_illustration_a}
    \end{minipage}
    \hfill
    \begin{minipage}[b]{0.45\textwidth}
        \centering
        \begin{tikzpicture}[scale=1]
            \draw[->, thick] (-1.7, -1) -- (-1.7, 1.2) node[left] {$\pi(s,\bar{a},q)$};
            \draw[->, thick] (-1.7, -1) -- (1.7, -1) node[right] {$s$};
            \draw[thick, black, domain=-1.5:1.5, samples=100] 
                plot (\x, {0.9*sin(0.9*pi*\x r)});
            \node[black, below] at (-1.4,-1) {$s_1$};
            \node[black, below] at (-0.6,-1) {$s_{i_1}$};
            \draw[blue, dashed] (-0.56,-1) -- (-0.56,-0.9);
            \node[black, below] at (0.5,-1) {$s_{i_2}$};
            \draw[blue, dashed] (0.56,-0.9) -- (0.56,0.9);
            \node[black, below] at (1.4,-1) {$s_{i_3}$};
        \end{tikzpicture}
        \vspace{1em}
        \subcaption[]{The figure plots the  skill distribution $\pi(s,\bar{a},q)$, which is bimodal: we select $s_1< \cdots< s_{i_1}$ from ${\mathbf S}_1$; $s_{i_1+1}< \cdots< s_{i_2}$ from ${\mathbf S}_2$; $s_{i_2+1}< \cdots< s_{i_3}$ from ${\mathbf S}_3$.}
        \label{fig:modality_illustration_b}
    \end{minipage}
    \vspace{0.1cm}
    \caption{A stylized illustration of transition ratio and skill distribution}
    \label{fig:modality_illustration}
\end{figure}

\begin{proof}[Proof of Lemma~\ref{lem:bayes_effort_nonmonotonic}]
Recall from Equation~\eqref{eq:steady_state_prob} that $\frac{\pi_{k+1}}{\pi_{k}}=\frac{\lambda( e^\star_v(s_k,\bar{a},q))}{\mu}$. Since $\lambda$ is an increasing function, the transition ratio $r(s;\mu):=\frac{\lambda( e^\star_v(s,\bar{a},q))}{\mu}$ has the same reversal pattern in skill $s$ as the effort $ e^\star_v(s,\bar{a},q)$. 
Because $e^\star_v(s,\bar a,q)$ is non-monotonic in $s$ and equals zero for all sufficiently large $s$, there exists an interval ${\mathbf S}$ on which $e^\star_v(s,\bar a,q)$ first increases and then decreases. Hence, $r(s;\mu)$ exhibits a local maximum on ${\mathbf S}$. 

Given that the numerator of the transition ratio $r(s;\mu)$ is independent of $\mu$ and the denominator is equal to $\mu$, we can thus choose $\mu>0$ such that $r(s;\mu)>1$ at the local maximum and $r(s;\mu)<1$ at the endpoints of ${\mathbf S}$. As a result,
$r(s;\mu)$ crosses $1$ twice and the two crossing points partition ${\mathbf S}$ into three sub-segments ${\mathbf S}= {\mathbf S}_1 \cup {\mathbf S}_2 \cup {\mathbf S}_3$, where ${\mathbf S}_1 , {\mathbf S}_2 , {\mathbf S}_3$ denote
the leftmost, middle, and rightmost sub-segments (see Figure~\ref{fig:modality_illustration_a}).

Now select skill states $s_1< \cdots< s_{i_1}$ from ${\mathbf S}_1$, $s_{i_1+1}< \cdots< s_{i_2}$ from ${\mathbf S}_2$, and $s_{i_2+1}< \cdots< s_{i_3}$ from ${\mathbf S}_3$, where $i_1\geq 1$, $i_2-i_1\geq 1$, and $i_3-i_2\geq 2$ (see Figure~\ref{fig:modality_illustration_b}). By construction, the transition ratio is less than $1$ on ${\mathbf S}_1$, greater than $1$ on ${\mathbf S}_2$, and less than $1$ on ${\mathbf S}_3$. Hence, we have
\begin{equation}\label{eq:prob_ratio}
    \pi_{i+1}/\pi_{i}<1 \text{ for } i\leq i_1, \quad \pi_{i+1}/\pi_{i}>1  \text{ for }  i_1<i\leq i_2, \quad \pi_{i+1}/\pi_{i}<1 \text{ for } i>i_2.
\end{equation}
This implies that the steady-state probabilities first decrease, then increase, and then decrease again. As a result, the selected skill set $S^\star=(s_1,\cdots,s_{i_3})$ leads to a bimodal distribution $\boldsymbol{\pi}(\mu, \bar{a},q,S^\star)$.
Finally, adding additional skill states outside $S^\star$ preserves the bimodality, which implies for any richer skill set $S\supseteq S^\star$, the steady-state distribution is multimodal.
\end{proof} 
Another implication of \eqref{eq:prob_ratio} is that: the multimodality established for a skill set is preserved for any superset of that skill set, since adding skill states does not alter the pattern of $\pi_{i+1}/\pi_i$ over ${\mathbf S}_1 , {\mathbf S}_2 , {\mathbf S}_3$.
Equipped with these lemmas, we now prove Theorem~\ref{thm:bayes_skill_multimodal}.
\begin{proof}[Proof of Theorem~\ref{thm:bayes_skill_multimodal}]
    If Condition~\ref{cond:multimodal_cond} holds, combining the first part of Lemma~\ref{lem:effort_caseA+B} with Lemma \ref{lem:bayes_effort_nonmonotonic}, it follows that there exists a threshold $\tilde{a}\geq 0$ such that: if $\bar{a}>\Tilde{a}$, the skill distribution is multimodal; otherwise, the skill distribution is unimodal; this shows the first part of the theorem. Otherwise, if Condition~\ref{cond:multimodal_cond} does not hold, combining the second part of Lemma~\ref{lem:effort_caseA+B} with Lemma~\ref{lem:bayes_effort_decrease}, it follows that the skill distribution is unimodal, which shows the second part of the theorem.
\end{proof}

\subsection{Discussion}\label{sec:literacy_discussion}
\paragraph{Connection to the Matthew effect.}  The Matthew effect describes a rich-get-richer and poor-get-poorer phenomenon, which is commonly observed in education~\citep{stanovich2009matthew}. It depicts an organism-environment correlation: differentially advantaged individuals are exposed to endogenous differences in environmental qualities. Individuals who have early advantages are better at engaging in tasks more efficiently; in turn, this engagement reinforces their advantages. 
In our setting, agents with heterogeneous literacy are exposed to different distributions of signal qualities (signal $\hat{a}=\bar{a}$ or $\hat{a}=0$). High-skill agents are more able to detect AI failures and therefore compensate by exerting additional effort. Effort, in turn, shapes the evolution of skill. 
Thus, we identify a multimodal steady-state skill distribution analogous to the Matthew effect. 

\section{Conclusion}\label{sec:conclusion}
GenAI tools are increasingly integrated into work and education, yet empirical evidence on their impact is mixed, with productivity gains in some settings and losses in others. This work proposes a model of human-AI interaction in which AI assistance, effort, skill, and productivity are interdependent. It characterizes when productivity paradoxes and skill polarization arise, providing analytical measures and sufficient conditions. Together, our results clarify when GenAI tools are likely to improve outcomes and when they may backfire. The principal failure in the use of GenAI is a reduction of effort and its downstream consequences (due to skill decay, AI unreliability, or lack of AI literacy), which is consistent with the empirical evidence to date \citep{poulidis2025self,dell2023navigating,eames2026computer}. 

Most of our results remain robust when key assumptions are relaxed. 
First, while we assume the marginal opportunity cost of effort remains constant, the productivity paradoxes carry over to the convex cases where the marginal cost of effort is increasing (Appendix~\ref{app:cvx_cost}). The effort reduction effect of AI assistance is less pronounced, but productivity loss still occurs.
Second, while we consider agents who maximize short-term utility based on current skill, the productivity paradoxes occur when they are forward-looking and, in fact, even fully rational. This requires a non-trivial extension where agents solve a continuous-time dynamic program in a two-skill state space (Appendix~\ref{app:non_myopic}). While low-skill agents exert more effort compared to the myopic case, increasing AI proficiency can cause productivity decline again. Finally, though we focus on productivity, steady-state welfare can also decrease with the level of AI assistance (Appendix~\ref{app:welfare}). These robustness checks indicate that the paradoxes are not artifacts of specific assumptions, but a more general phenomenon that can arise from various settings, which provides a caveat to the usage of AI assistance.

Our analysis of these failure modes hints at several guidelines for practice. First, system designers could introduce incentives or guardrails to elicit more effort from users. From a mechanism design standpoint, platforms could introduce ex-ante frictions that require users to exert effort before accessing outputs, such as token rationing or follow-up clarification questions. Recent LLM tools (e.g., Claude Code's Interactive Interview Mode and ChatGPT's Study mode) have introduced prompt clarification strategies, turning the user experience from that of an ``answering machine'' to a dialogue with the AI system, inviting iterative exploration. Ex-post output verification is harder to enforce, but may be feasible if LLMs advance in uncertainty quantification, providing calibrated caveats on outputs that signal when human scrutiny is warranted, thus prompting users to exert more effort. Second, for users and organizations, AI adoption in knowledge-critical work should be accompanied by knowledge management practices. Investing in AI literacy is important to sustain long-term development and avoid skill erosion.  Recent regulations underscore this point: for example, the EU Artificial Intelligence Act requires AI deployment to consider individual knowledge, experience, education, training, and the context in which AI systems are to be used.\footnote{See Article 4 of the EU Artificial Intelligence Act (Regulation (EU) 2024/1689): \url{https://eur-lex.europa.eu/eli/reg/2024/1689/oj}. Accessed on 05/2026.}

\bibliographystyle{alpha}
\bibliography{ref_paper}

\newpage
\appendix
\begin{center}
    \textbf{\Large Appendix}
\end{center}
\startcontents[appendices]
\printcontents[appendices]{}{1}{\setcounter{tocdepth}{2}}
\newpage
\section{Derivation of steady-state distribution (Section~\ref{sec:basic})}\label{app:steady_state}
\begin{lemma}[\cite{grimmett2001probability}, Chap.~6.11]\label{lem:bdc_std_result}
Suppose $X$ is a continuous-time Markov chain taking values in $[N]$. Suppose that for any integer values  $n\in[N]$, the infinitesimal transition probabilities are given by
\begin{align*}
\mathbb P( X(t+h)=n+m|X(t)=n ) & =
\begin{cases}
    \lambda_n h + o(h), & \text{if } m=1,\\
    \mu_n h + o(h), &  \text{if }  m=-1,\\
    o(h), &  \text{if } |m|>1,
\end{cases} 
\end{align*}
where $\lambda_i, \mu_i >0$ are birth and death rates, for all admissible transitions.
Then, the steady-state distribution $\boldsymbol{\pi}=(\pi_1,\pi_2, \cdots,\pi_N)$ satisfies
\begin{equation*}
    \pi_k=\pi_1\prod_{i=1}^{k-1}\frac{\lambda_i}{\mu_{i+1}} \quad\text{and}\quad  \pi_1=\dfrac{1}{1+\sum_{k=2}^{N} \prod_{i=1}^{k-1}\frac{\lambda_i}{\mu_{i+1}}}.
\end{equation*}
\end{lemma}
\paragraph{Derivation of Eq.~\eqref{eq:steady_state_prob}.}
Our model corresponds to a Markov chain taking values in the skill states $S$. In the language of Lemma~\ref{lem:bdc_std_result}, skill state $s_k$ corresponds to state index $k$, upward transition rate $\mathcal{K}_{s_k s_{k+1}}= \lambda(e^\star(s_k,\bar{a},q))$ to birth rate $\lambda_k$, and downward transition rate $\mathcal{K}_{s_k s_{k-1}}= \mu$ to death rate $\mu_k$. Thus, Lemma~\ref{lem:bdc_std_result} gives the desired steady-state distribution.

\section{Supplementary materials for Section~\ref{sec:bdc}}
\subsection{Positive semi-elasticity gap (Section~\ref{sec:bdc_result})}\label{app:pf_positive_semi-elasticity_gap}
In this section, we explain how the positive semi-elasticity gap condition can be satisfied by production functions exhibiting either unboundedness or a vanishing derivative. Specifically, we consider the skill–AI interaction function of the form $h(s,a)=\phi\big(b(s)+w(a)\big)$, where  $\phi(\cdot)$, $b(\cdot)$ and $w(\cdot)$ are increasing functions. When $\phi$ is the identity function, it reduces to the additive model; when $\phi(\cdot)=\exp(\cdot)$, $b(s)=\omega_1\log{s}$, and $w(a)=\omega_2\log{a}$, it recovers the Cobb-Douglas form where $h(s,a)=s^{\omega_1} a^{\omega_2}$ with $\omega_1,\omega_2>0$. 

\begin{proposition}\label{prop:positive_semi-elasticity_gap}
    Given a concave transition function $\lambda$, a linear cost function $c(e)=\gamma e$, 
    and the first $m$ states $s_1,\cdots ,s_m$, let the skill-AI interaction function be $h(s,a)=\phi\big(b(s)+w(a)\big)$, where $b(\cdot)$ and $w(\cdot)$ are increasing, and $\phi(\cdot)$ is a differentiable and increasing function. If the composite function $\Tilde{p}=p\circ\phi$ is concave and satisfies either 
    1) Unboundedness: $\lim_{u\to \infty}\Tilde{p}(u)=\infty$, or 
    2) Vanishing derivative: $\Tilde{p}'(u)>0$ and $\lim_{u\to \infty}\Tilde{p}'(u)=0$, 
    then there exists a threshold $\bar{s}>s_m$ such that $\Delta_m > 0$ if and only if $s_{m+1}>\bar{s}$.
\end{proposition}
\begin{proof}[Proof of Proposition~\ref{prop:positive_semi-elasticity_gap}]
Recall that $h_a$ is the partial derivative of function $h$ with respect to $a$ and that the condition $\Delta_m>0$ is equivalent to: 
\begin{equation*}
    \sum_{j=1}^m\frac{\lambda'(x^\star-h(s_j,g(s_m))) h_a(s_j,g(s_m))}{\lambda(x^\star-h(s_j,g(s_m)))}> \frac{p'(h(s_{m+1},g(s_m)))h_a(s_{m+1},g(s_m))}{p(h(s_{m+1},g(s_m))) -p(x^\star)} .
\end{equation*}
Let $\Tilde{h}(s,a):=b(s)+w(a)$ and $u_m^\star=b(s_m)+w(g(s_m))$. Then $p(h(s,a))=\Tilde{p}(\Tilde{h}(s,a))$, $p(x^\star)=\Tilde{p}(u_m^\star)$, and $h_a(s,a)=\phi'(\Tilde{h}(s,a))w'(a)$. Substituting the expression of $h_a(s,a)$, the condition $\Delta_m>0$ is equivalent to 
\begin{equation*}
\sum_{j=1}^m\frac{\lambda'(x^\star-h(s_j,g(s_m)))\phi'(\Tilde{h}(s_j,g(s_m)))w'(g(s_m))}{\lambda(x^\star-h(s_j,g(s_m)))}
>
\frac{p'(h(s_{m+1},g(s_m)))\phi'(\Tilde{h}(s_{m+1},g(s_m)))w'(g(s_m))}{p(h(s_{m+1},g(s_m))) -p(x^\star)}.
\end{equation*}
Using $p(h(s,a))=\Tilde{p}(\Tilde{h}(s,a))$, $p(x^\star)=\Tilde{p}(u_m^\star)$, $\Tilde{p}'=(p'\circ\phi)\phi'$, and dividing both sides by $w'(g(s_m))>0$, 
\begin{equation}\label{eq:bcd_key_condition}
     \sum_{j=1}^m\frac{\lambda'(x^\star-h(s_j,g(s_m)))\phi'(\Tilde{h}(s_j,g(s_m)))}{\lambda(x^\star-h(s_j,g(s_m)))} > \frac{\Tilde{p}'(\Tilde{h}(s_{m+1},g(s_m)))}{\Tilde{p}(\Tilde{h}(s_{m+1},g(s_m))) -\Tilde{p}(u_m^\star)}.
\end{equation}
We show below that, under both conditions, the right-hand side of \eqref{eq:bcd_key_condition} goes to $0$ as $s_{m+1}\rightarrow\infty$, i.e.,
\begin{equation}\label{eq:bcd_limit_claim}
    \lim_{s_{m+1}\to\infty} \frac{\Tilde{p}'(\Tilde{h}(s_{m+1},g(s_m)))}{\Tilde{p}(\Tilde{h}(s_{m+1},g(s_m))) -\Tilde{p}(u_m^\star)}=0.
\end{equation}
With \eqref{eq:bcd_limit_claim}, the monotonicity and concavity of $\Tilde{p}$ imply that the right-hand side of \eqref{eq:bcd_key_condition} is a continuous and decreasing function in $s_{m+1}$; since it tends to zero as $s_{m+1}$ goes to infinity, there exists a threshold $\bar{s}$ such that $\Delta_m > 0$ holds if and only if $s_{m+1}>\bar{s}$, which completes the proof. We now verify \eqref{eq:bcd_limit_claim} under each condition:
\begin{enumerate}
\item Under the unboundedness condition, \eqref{eq:bcd_limit_claim} is guaranteed because
\begin{align*}
    & \lim_{s_{m+1}\to\infty} \big( \Tilde{p}(\Tilde{h}(s_{m+1},g(s_m))) -\Tilde{p}(u_m^\star) \big) = \infty, \quad \text{and} \\
    & \Tilde{p}'(\Tilde{h}(s_{m+1},g(s_m)))<\Tilde{p}'(\Tilde{h}(s_m,g(s_m)))=\Tilde{p}'(u_m^\star).
\end{align*}
\item Under the vanishing derivative condition, \eqref{eq:bcd_limit_claim} is guaranteed because
\begin{align*}
    & \lim_{s_{m+1}\to\infty} \Tilde{p}'(\Tilde{h}(s_{m+1},g(s_m)))=0, \quad \text{and}\\
    & \Tilde{p}(\Tilde{h}(s_{m+1},g(s_m))) -\Tilde{p}(u_m^\star)>\Tilde{p}(\Tilde{h}(s_m+1,g(s_m)))-\Tilde{p}(u_m^\star) \quad \text{for any } s_{m+1}>s_m+1.
\end{align*}
\end{enumerate}
\end{proof}

\subsection{Arbitrarily small steady-state productivity ratio (Proposition~\ref{prop:bcd_bad_epsilon_general})}\label{app:proof_bcd_bad_epsilon_general}
We show a stronger result by showing that the productivity ratio can be made arbitrarily small even when restricting attention to an intermediate region of the AI assistance level $a$, $x^\star -s_2< a\leq x^\star-s_1=x^\star$ (recall that $s_1=0$). We first establish a productivity ratio describing the limiting behavior of $\lambda(0)$. 

\begin{lemma}\label{lem:ratio_lambdatozero}
Given a skill set $S=(s_1,s_2,\cdots,s_N)$ with $0=s_1<s_2<\cdots<s_N$ and a value $z$ such that $0 \leq z \leq \min(x^\star, s_2)$, if $\lambda(0)<\mu$ we have
\begin{equation*}
    \dfrac{\mathcal{P}(x^\star)}{\mathcal{P}(x^\star-z)} \leq \dfrac{p(x^\star)+(N-1)\frac{\lambda(0)}{\mu}p(x^\star+s_N)}{ \dfrac{1}{1+\frac{\lambda(z)}{\mu}} \left(p(x^\star) + \frac{\lambda(z)}{\mu} p(x^\star+s_2-z)\right)}.
\end{equation*}
\end{lemma}
Equipped with Lemma~\ref{lem:ratio_lambdatozero}, we now prove Proposition~\ref{prop:bcd_bad_epsilon_general}.
\begin{proof}[Proof of Proposition~\ref{prop:bcd_bad_epsilon_general}]
Fix the production function $p$ with perfect substitution skill-AI interaction function $h(s,a)=s+a$ and the cost function $c(e)=\gamma e$. Recall that the critical level $x^\star = x^\star(p,\gamma)$ is dependent on both the production and the cost functions. Lemma~\ref{lem:ratio_lambdatozero} implies that, when $\lambda(0)<\mu$, we have
\begin{align}\label{eq:arbitrary_ratio_expression}
    \inf_{p,\gamma}\inf_{a_\ell <a_h}\dfrac{\mathcal{P}(a_h)}{\mathcal{P}(a_\ell)}&\leq \inf_{p,\gamma}\inf_{0\leq z\leq \min(x^\star(p,\gamma),s_2)}\dfrac{\mathcal{P}(x^\star(p,\gamma))}{\mathcal{P}(x^\star(p,\gamma)-z)}\nonumber\\
    &\leq \inf_{p,\gamma}\inf_{0\leq z\leq \min(x^\star(p,\gamma),s_2)} \dfrac{p(x^\star(p,\gamma))+(N-1)\frac{\lambda(0)}{\mu}p(x^\star(p,\gamma)+s_N)}{\dfrac{1}{1+\frac{\lambda(z)}{\mu}} \left(p(x^\star(p,\gamma)) + \frac{\lambda(z)}{\mu} p(x^\star(p,\gamma)+s_2-z)\right)}.
\end{align}
We consider a piecewise linear production function $p$ whose slope changes only slightly at $x^\star$, with two different yet arbitrarily close slopes $\gamma+\varepsilon^+$ and $\gamma-\varepsilon^-$, where $\varepsilon^+, \varepsilon^- >0$ are sufficiently small. Specifically,
\begin{align*}
    p(x)& =
        \begin{cases}
            (\gamma+\varepsilon^+)x , & \text{if } x \leq x^\star ,\\
            (\gamma-\varepsilon^-)(x-x^\star) + (\gamma+\varepsilon^+)x^\star, & \text{if } x> x^\star.
        \end{cases}
\end{align*}
The production function remains almost linear, with a kink at $x^\star$ (which is now a free parameter). Let $\varepsilon^+,\varepsilon^-\to 0$. 
We consider a concave transition function $\lambda$ such that $\lambda(e)=\lambda(0)+\frac{16\mu}{\varepsilon^2 s_2}e$ and $\lambda(0)=\frac{\varepsilon s_2}{8(N-1)s_N}\mu$, where we assume $\varepsilon\leq 1$.
Setting $x^{\star}=\frac{\varepsilon}{4}s_2$, we upper bound the ratio in \eqref{eq:arbitrary_ratio_expression} by its evaluation at $z=x^{\star}$ as follows: 
\begin{align*}
    \inf_{p,\gamma}\inf_{a_\ell <a_h}\dfrac{\mathcal{P}(a_h)}{\mathcal{P}(a_\ell)} & \leq \dfrac{p(x^\star)+(N-1)\frac{\lambda(0)}{\mu}p(x^\star+s_N)}{\dfrac{1}{1+\frac{\lambda(z)}{\mu}} \left(p(x^\star) + \frac{\lambda(z)}{\mu} p(x^\star+s_2-z)\right)} 
    = \dfrac{\frac{\varepsilon}{4}+(N-1)\frac{\varepsilon s_2}{8(N-1)s_N}(\frac{\varepsilon}{4}+\frac{s_N}{s_2})}{\dfrac{1}{1+\frac{\lambda(\frac{\varepsilon}{4}s_2)}{\mu}} \left(\frac{\varepsilon}{4} + \frac{\lambda(\frac{\varepsilon}{4}s_2)}{\mu} \right)} \\
    &\leq \dfrac{\frac{\varepsilon}{2}}{\dfrac{1}{1+\frac{\lambda(\frac{\varepsilon}{4}s_2)}{\mu}}  \frac{\lambda(\frac{\varepsilon}{4}s_2)}{\mu}} \leq \dfrac{\frac{\varepsilon}{2}}{1-\frac{\varepsilon}{4}} \leq\varepsilon.
\end{align*}
The first equality substitutes the expression of production function $p$ and divides both the denominator and numerator by $\gamma s_2$. The second inequality uses $\frac{\varepsilon}{4}+\frac{s_N}{s_2}\leq 2\frac{s_N}{s_2}$ (as $s_N>s_2$). The third inequality uses $\frac{\lambda(\frac{\varepsilon}{4} s_2)}{\mu+\lambda(\frac{\varepsilon}{4} s_2)}\geq 1-\frac{\varepsilon}{4}$. The last inequality uses
$\varepsilon\leq 1$. Since the infimum is strictly smaller than $\varepsilon$, there exist
$\varepsilon^+,\varepsilon^->0$ for which $\inf_{a_\ell <a_h}\frac{\mathcal{P}(a_h)}{\mathcal{P}(a_\ell)}\leq\varepsilon$, which completes the proof.
\end{proof}

\begin{proof}[Proof of Lemma~\ref{lem:ratio_lambdatozero}]
Note that the utility-maximizing effort given $a_h=x^\star$ is zero for each skill, i.e., $e^\star(s_i, x^\star)=0$ holds for all~$i$; the utility-maximizing effort given $a_\ell =x^\star-z$ ($z\leq s_2$) is zero for each skill except the lowest one, i.e., $e^\star(s_i,x^\star-z)=0$ holds for all $i\geq 2$. Recall the steady-state distribution $\boldsymbol{\pi}$ is given by Equation~\eqref{eq:steady_state_prob}, we have
\begin{align*}
    \dfrac{\mathcal{P}(x^\star)}{\mathcal{P}(x^\star-z)}
    =&\dfrac{\dfrac{1}{\sum_{k=1}^N(\frac{\lambda(0)}{\mu})^{k-1}}\cdot\left(\sum_{k=1}^{N} (\frac{\lambda(0)}{\mu})^{k-1} p(x^\star+s_k)\right)}{\dfrac{1}{1+\frac{\lambda(z)}{\mu}\sum_{k=2}^{N}(\frac{\lambda(0)}{\mu})^{k-2}}\cdot\left(p(x^\star) + \sum_{k=2}^{N} \frac{\lambda(z)}{\mu}(\frac{\lambda(0)}{\mu})^{k-2} p(x^\star+s_k-z)\right)}. 
\end{align*}
For the numerator, since $\sum_{k=1}^N(\frac{\lambda(0)}{\mu})^{k-1}\geq 1$ and $(\frac{\lambda(0)}{\mu})^{k-1}\leq\frac{\lambda(0)}{\mu}$ for $k\geq 2$ (as $\lambda(0)\leq\mu$), it is upper bounded by $p(x^\star)+(N-1)\frac{\lambda(0)}{\mu}p(x^\star+s_N)$.
For the denominator, it is lower bounded by 
$$\dfrac{p(x^\star) + \sum_{k=2}^{N} \frac{\lambda(z)}{\mu}(\frac{\lambda(0)}{\mu})^{k-2} p(x^\star+s_2-z)}{1+\frac{\lambda(z)}{\mu}\sum_{k=2}^{N}(\frac{\lambda(0)}{\mu})^{k-2}}\geq \dfrac{p(x^\star) + \frac{\lambda(z)}{\mu} p(x^\star+s_2-z)}{1+\frac{\lambda(z)}{\mu}}.$$
Combining both bounds completes the proof.
\end{proof}

\subsection{Steady-state skill degradation (Proposition~\ref{prop:bdc_skill_dist})}\label{app:proof_bdc_skill_dist}
Proposition~\ref{prop:bdc_skill_dist} immediately follows from an auxiliary exchange-of-summation inequality, stated next. 
\begin{lemma}\label{lem:ratio_ineq}
    Suppose $x_1,x_2,\cdots,x_N$ and $y_1,y_2,\cdots,y_N$ are $2N$ non-negative values such that $x_i\geq y_i$ for all $i$. Then the following inequality holds for any $k\in [N]$:
    $$\dfrac{\sum_{m=1}^k\prod_{i=1}^m x_i}{\sum_{m=1}^{N}\prod_{i=1}^m x_i}\leq \dfrac{\sum_{m=1}^k\prod_{i=1}^m y_i}{\sum_{m=1}^{N}\prod_{i=1}^m y_i},$$
    given that both denominators are strictly positive.
\end{lemma}
Equipped with Lemma~\ref{lem:ratio_ineq}, we now prove Proposition~\ref{prop:bdc_skill_dist}.
\begin{proof}[Proof of Proposition~\ref{prop:bdc_skill_dist}]
Recall that the steady-state distribution $\boldsymbol{\pi}(a)$  is given by Equation~\eqref{eq:steady_state_prob}, which involves sums of products of transition ratios $\frac{\lambda(e^\star(s_i,a))}{\mu}$. Let $x_{i+1}=\frac{\lambda(e^\star(s_i,a_\ell))}{\mu}$ and $y_{i+1}=\frac{\lambda(e^\star(s_i,a_h))}{\mu}$ for $i\in [N-1]$, with $x_1=y_1=1$. Then $\prod_{i=1}^k x_i = \prod_{i=1}^{k-1}\frac{\lambda(e^\star(s_i,a_\ell))}{\mu} $ and $\prod_{i=1}^k y_i = \prod_{i=1}^{k-1}\frac{\lambda(e^\star(s_i,a_h))}{\mu}$.
Since higher AI assistance level $a$ induces less effort (Proposition~\ref{prop:effort_productivity_expressions}), we have $e^\star(s_i,a_\ell)\geq e^\star(s_i,a_h)$ for all $i$. Note that $\lambda(e)$ is increasing in $e$, it follows that $x_i \geq y_i$ for all $i$.  
    The desired inequality $\sum_{i=1}^k \pi_i(a_\ell) \leq \sum_{i=1}^k\pi_i(a_h)$ then follows directly from Lemma~\ref{lem:ratio_ineq}.
\end{proof}

\begin{proof}[Proof of Lemma~\ref{lem:ratio_ineq}]
For any $k\in [N]$, we have
\begin{align*}
    \sum_{m=1}^{N}\prod_{i=1}^m x_i \sum_{m=1}^k\prod_{i=1}^m y_i&=\sum_{n=1}^{N}\prod_{j=1}^n x_j \sum_{m=1}^k\prod_{i=1}^m y_i= \sum_{\substack{ 1\leq m\leq k\\   1\leq n \leq N}}\prod_{j=1}^n x_j \prod_{i=1}^m y_i &\text{(exchanging the order of summation)}\\
    &= \sum_{\substack{ 1\leq m\leq k\\   1\leq n \leq k}}\prod_{j=1}^n x_j \prod_{i=1}^m y_i+\sum_{\substack{ 1\leq m\leq k\\  k+1\leq n \leq N}}\prod_{j=1}^n x_j \prod_{i=1}^m y_i & \text{(splitting into $n\leq k$ and $n\geq k+1$)} \\
    &= \sum_{\substack{ 1\leq m\leq k\\   1\leq n \leq k}}\prod_{j=1}^n x_j \prod_{i=1}^m y_i+\sum_{\substack{ 1\leq m\leq k\\  k+1\leq n \leq N}}\prod_{j=m+1}^n x_j \prod_{i=1}^m (x_iy_i). & \text{(regrouping terms)}\\
\end{align*}
Similarly,
\begin{align*}
     \sum_{m=1}^k\prod_{i=1}^m x_i\sum_{m=1}^{N}\prod_{i=1}^m y_i = \sum_{\substack{ 1\leq m\leq k\\   1\leq n \leq k}}\prod_{j=1}^n x_j \prod_{i=1}^m y_i+\sum_{\substack{ 1\leq m\leq k\\  k+1\leq n \leq N}}\prod_{j=m+1}^n y_j \prod_{i=1}^m (x_iy_i).
\end{align*}
Since $x_j\geq y_j$, we have
\begin{align*}
    &\sum_{\substack{ 1\leq m\leq k\\   1\leq n \leq k}}\prod_{j=1}^n x_j \prod_{i=1}^m y_i+\sum_{\substack{ 1\leq m\leq k\\  k+1\leq n \leq N}}\prod_{j=m+1}^n x_j \prod_{i=1}^m (x_iy_i)
    \geq \sum_{\substack{ 1\leq m\leq k\\   1\leq n \leq k}}\prod_{j=1}^n x_j \prod_{i=1}^m y_i+\sum_{\substack{ 1\leq m\leq k\\  k+1\leq n \leq N}}\prod_{j=m+1}^n y_j \prod_{i=1}^m (x_iy_i),
\end{align*}
which implies $\sum_{m=1}^{N}\prod_{i=1}^m x_i \sum_{m=1}^k\prod_{i=1}^m y_i \geq \sum_{m=1}^k\prod_{i=1}^m x_i\sum_{m=1}^{N}\prod_{i=1}^m y_i$. After rearranging the terms, we conclude the proof. 
\end{proof}

\subsection{Derivative characterization (Lemma~\ref{lem:thm_productivity_derivative})}\label{app:thm_productivity_derivative}
\begin{proof}[Proof of Lemma~\ref{lem:thm_productivity_derivative}]
Recall that the steady-state productivity is given by Equation~\eqref{eq:steady_state_prob} and the productivity at each single skill state is given by Proposition~\ref{prop:effort_productivity_expressions}. As a result, the steady-state productivity is
\begin{align*}
    \mathcal{P}(a) &= \sum_{k=1}^{N} \pi_k(a) \cdot p(e^\star(s_k,a)+h(s_k,a)) = \sum_{k=1}^{N} \bigg(\pi_1(a) \prod_{i=1}^{k-1}\dfrac{\lambda(e^\star(s_i,a))}{\mu}\bigg)\cdot \max\left(p(x^\star),p(h(s_k,a))\right). 
\end{align*}
Recall from Section~\ref{sec:bdc_result} that ${\cal I}_m=[(g(s_{m+1}))^+,(g(s_m))^+)$, where $g(s)$ is
defined by $h(s,g(s))=x^\star$. Given $m\in[N-1]$, for any $a\in {\cal I}_m$, the productivity at state $s_k$ satisfies $p^\star(s_k,a)=p(x^\star)$ if $k\leq m$ and $p^\star(s_k,a)=p(h(s_k,a))$ if $k\geq m+1$. Thus, splitting the sum into the cases $k\leq m$ and $k\geq m+1$, we have
\begin{align*}
     \mathcal{P}(a)&= \pi_1(a) \bigg(\sum_{k=1}^m \frac{\prod_{i=1}^{k-1}\lambda(e^\star(s_i,a))}{\mu^{k-1}} p(x^\star) + \sum_{k=m+1}^{N} \frac{\prod_{i=1}^m \lambda(e^\star(s_i,a)) \lambda(0)^{k-m-1}}{\mu^{k-1}}  p(h(s_k,a)) \bigg).
\end{align*}
Recall that Equation~\eqref{eq:steady_state_prob} implies
$$\pi_1(a)=\dfrac{1}{1+\sum_{k=2}^{N}\prod_{i=1}^{k-1}\frac{\lambda(e^\star(s_i,a))}{\mu}}=\dfrac{1}{\sum_{k=1}^{N}\prod_{i=1}^{k-1}\frac{\lambda(e^\star(s_i,a))}{\mu}}.$$
Thus, the steady-state productivity is 
\begin{align*}
    \mathcal{P}(a)&= \dfrac{ \sum_{k=1}^m \frac{\prod_{i=1}^{k-1}\lambda(x^\star -h(s_i,a))}{\mu^{k-1}} p(x^\star) 
    + \sum_{k=m+1}^{N} \frac{\prod_{i=1}^m \lambda(x^\star -h(s_i,a))\lambda(0)^{k-m-1}}{\mu^{k-1}} p(h(s_k,a)) }{ \sum_{k=1}^m \frac{\prod_{i=1}^{k-1}\lambda(x^\star -h(s_i,a))}{\mu^{k-1}}
    + \sum_{k=m+1}^{N} \frac{\prod_{i=1}^m \lambda(x^\star -h(s_i,a))\lambda(0)^{k-m-1}}{\mu^{k-1}} }  \\
    &= p(x^\star)+ \dfrac{ \sum_{k=m+1}^{N} \frac{\prod_{i=1}^m \lambda(x^\star -h(s_i,a))\lambda(0)^{k-m-1}}{\mu^{k-1}} (p(h(s_k,a)) -p(x^\star))}{ \sum_{k=1}^m\frac{\prod_{i=1}^{k-1}\lambda(x^\star -h(s_i,a))}{\mu^{k-1}}
    + \sum_{k=m+1}^{N} \frac{\prod_{i=1}^m \lambda(x^\star -h(s_i,a))\lambda(0)^{k-m-1}}{\mu^{k-1}} } & \text{(subtracting $p(x^\star)$)} \\
    &= p(x^\star)+ \dfrac{ \sum_{k=m+1}^{N} \frac{\lambda(0)^{k-m-1}}{\mu^{k-1}} (p(h(s_k,a)) -p(x^\star))}{ \sum_{k=1}^m \frac{1}{\mu^{k-1} \prod_{i=k}^{m}\lambda(x^\star -h(s_i,a))}
    + \sum_{k=m+1}^{N} \frac{\lambda(0)^{k-m-1}}{\mu^{k-1}} },
\end{align*}
where the last equality follows by dividing both the numerator and the denominator by $\prod_{i=1}^m \lambda(x^\star -h(s_i,a))$. 
Recall that $$A_{k,m}(a):=\frac{1}{\prod_{i=k}^{m}\lambda(x^\star -h(s_i,a))}, \forall k\in[m] \quad \text{and} \quad B_{k,m}=\lambda(0)^{k-m-1} , \forall k\in[N]\setminus [m].$$
Substituting these expressions yields
$$\mathcal{P}(a)= p(x^\star)+ \dfrac{ \sum_{k=m+1}^{N} \frac{B_{k,m}}{\mu^{k-1}} (p(h(s_k,a)) -p(x^\star))}{ \sum_{k=1}^m \frac{A_{k,m}(a)}{\mu^{k-1}} + \sum_{k=m+1}^{N} \frac{B_{k,m}}{\mu^{k-1}}}. $$
Differentiating the productivity with respect to $a$ gives 
\begin{align*}
    \mathcal{P}'(a) 
    = \frac{1}{\big(\sum_{k=1}^m \frac{A_{k,m}(a)}{\mu^{k-1}} +\sum_{k=m+1}^{N} \frac{B_{k,m}}{\mu^{k-1}}\big)^2} 
    &\bigg[\Big( \sum_{k=m+1}^{N} \frac{B_{k,m}}{\mu^{k-1}} \frac{d}{da}p(h(s_k,a))\Big) \Big(\sum_{k=1}^m \frac{A_{k,m}(a)}{\mu^{k-1}} +\sum_{k=m+1}^{N} \frac{B_{k,m}}{\mu^{k-1}}\Big) \\
    &- \Big( \sum_{k=m+1}^{N} \frac{B_{k,m}}{\mu^{k-1}}(p(h(s_k,a))-p(x^\star))\Big) \sum_{k=1}^m \frac{A_{k,m}'(a)}{\mu^{k-1}}\bigg],
\end{align*}
which completes the proof.
\end{proof}

\subsection{Steady-state productivity characterization (Lemma~\ref{lem:bdc_productivity_unimodal})}\label{app:proof_bdc_productivity}
Recall that the derivative of the steady-state productivity is established in Lemma~\ref{lem:thm_productivity_derivative}. Recall that $g(\cdot)$ satisfies $h(s, g(s)) = x^\star$, ${\cal I}_m=[(g(s_{m+1}))^+, (g(s_m))^+)$, and
$$A_{k,m}(a):=\frac{1}{\prod_{i=k}^{m}\lambda(x^\star -h(s_i,a))}, \forall k\in[m] \quad \text{and} \quad B_{k,m}=\lambda(0)^{k-m-1} , \forall k\in[N]\setminus [m],$$ 
given $m\in[N-1]$ and $a\in {\cal I}_m$.
We start with a lemma that facilitates the analysis of derivatives. Recall that $h(s,a)=b(s)+w_1(s)w_2(a)$, where $b,w_1,w_2\geq 0$ are increasing functions. Thus, $h_a(s,a)=w_1(s)w_2'(a)$.

An auxiliary function which facilitates our analysis of Lemma~\ref{lem:bdc_productivity_unimodal} is
$$H_m(a):=\dfrac{\sum_{k=1}^m \frac{A_{k,m}'(a)}{\mu^{k-1}}/w_2'(a)}{\sum_{k=1}^m \frac{A_{k,m}(a)}{\mu^{k-1}}+\sum_{k=m+1}^{N} \frac{B_{k,m}}{\mu^{k-1}}}.$$
Here, as $h$ is increasing, $w_2'(a)>0$, and the expression above is well-defined. The following lemma establishes that $H_m'(a)\geq 0$.
\begin{lemma}\label{lem:thm_productivity_derivative_monotonicity}
    For every $m\in[N-1]$, the function $H_m(a)$ defined above is increasing in $a$ on ${\cal I}_m$.
\end{lemma}
Equipped with Lemma~\ref{lem:thm_productivity_derivative} and Lemma~\ref{lem:thm_productivity_derivative_monotonicity}, we now prove Lemma~\ref{lem:bdc_productivity_unimodal}.
\begin{proof}[Proof of Lemma~\ref{lem:bdc_productivity_unimodal}]
Recall that the derivative of the steady-state productivity is established in Lemma~\ref{lem:thm_productivity_derivative}. Given $m\in[N-1]$, for any $a\in {\cal I}_m$, rearranging the expression of the derivative yields that 
\begin{align}\label{eq:derivative_mono_eq}
    &  \mathcal{P}'(a)\cdot \Big(\sum_{k=1}^m \frac{A_{k,m}(a)}{\mu^{k-1}} +\sum_{k=m+1}^{N} \frac{B_{k,m}}{\mu^{k-1}}\Big) \nonumber\\
    = &  \bigg( \sum_{k=m+1}^{N} \frac{B_{k,m}}{\mu^{k-1}} p'(h(s_k,a))h_a(s_k,a) \bigg) - \bigg( \sum_{k=m+1}^{N} \frac{B_{k,m}}{\mu^{k-1}} \big(p(h(s_k,a))-p(x^\star)\big)\bigg) \frac{\sum_{k=1}^m \frac{A_{k,m}'(a)}{\mu^{k-1}}}{\sum_{k=1}^m \frac{A_{k,m}(a)}{\mu^{k-1}} +\sum_{k=m+1}^{N} \frac{B_{k,m}}{\mu^{k-1}}}\nonumber\\
    =&  \bigg( \sum_{k=m+1}^{N} \frac{B_{k,m}}{\mu^{k-1}} p'(h(s_k,a))w_1(s_k)w_2'(a) \bigg) - \bigg( \sum_{k=m+1}^{N} \frac{B_{k,m}}{\mu^{k-1}} \big(p(h(s_k,a))-p(x^\star)\big)\bigg) H_m(a)w_2'(a)\nonumber\\
    =&\bigg( \underbrace{\sum_{k=m+1}^{N} \frac{B_{k,m}}{\mu^{k-1}} p'(h(s_k,a))w_1(s_k)}_{\text{(I)}} - \underbrace{\sum_{k=m+1}^{N} \frac{B_{k,m}}{\mu^{k-1}} \big(p(h(s_k,a))-p(x^\star)\big)}_{\text{(II)}}H_m(a)\bigg) w_2'(a).
\end{align}
Note that term (I) is decreasing in $a$ by the concavity of $p$, and term (II) is increasing in $a$ by the monotonicity of $p$. Lemma~\ref{lem:thm_productivity_derivative_monotonicity} implies that $H_m(a)$ is increasing in $a$.
Hence, \eqref{eq:derivative_mono_eq} is decreasing in $a$, which implies it can cross $0$ at most once. This implies that the derivative $\mathcal{P}'(a)$ can cross $0$ at most once on ${\cal I}_m$. As a result, the productivity is either monotone or increasing-then-decreasing on ${\cal I}_m$, and a decreasing region of $\mathcal{P}(a)$ exists within ${\cal I}_m$ if and only if $\mathcal{P}'_{-}(a)<0$ at the right endpoint of ${\cal I}_m$. 
\end{proof}

\begin{proof}[Proof of Lemma~\ref{lem:thm_productivity_derivative_monotonicity}]
Given the decay rate $\mu>0$, we define $\bar{A}_{k,m}(a)=\frac{A_{k,m}(a)}{\mu^{k-1}}$ and $\bar{B}_m:=\sum_{k=m+1}^{N} \frac{B_{k,m}}{\mu^{k-1}}=\sum_{k=m+1}^{N} \frac{\lambda(0)^{k-m-1}}{\mu^{k-1}}$. Thus, the function $H_m(a)$ can be expressed as
$$H_m(a)=\dfrac{\sum_{k=1}^m \bar{A}_{k,m}'(a)/w_2'(a)}{\sum_{k=1}^m \bar{A}_{k,m}(a)+ \bar{B}_m}=\dfrac{\sum_{k=1}^m \bar{A}_{k,m}(a) \sum_{j=k}^m\frac{\lambda'(x^\star -h(s_j,a))w_1(s_j)}{\lambda(x^\star -h(s_j,a))}}{\sum_{k=1}^m \bar{A}_{k,m}(a)+ \bar{B}_m}.$$
where the second equality holds because 
\begin{align*}
    \bar{A}_{k,m}'(a) = \frac{1}{\mu^{k-1} \prod_{i=k}^{m}\lambda(x^\star -h(s_i,a))} \sum_{j=k}^m\frac{\lambda'(x^\star -h(s_j,a))h_a(s_j,a)}{\lambda(x^\star -h(s_j,a))} = \bar{A}_{k,m}(a) \sum_{j=k}^m\frac{\lambda'(x^\star -h(s_j,a))w_1(s_j)w_2'(a)}{\lambda(x^\star -h(s_j,a))}.
\end{align*}
To show that $H_m(a)$ is increasing in $a$, we differentiate $H_m(a)$: 
\begin{multline*}
H_m'(a)=\frac{1}{\big(\sum_{k=1}^m \bar{A}_{k,m}(a)+ \bar{B}_m\big)^2}\bigg(\underbrace{\Big(\sum_{k=1}^m \bar{A}_{k,m}(a) \sum_{j=k}^m\frac{\lambda'(x^\star -h(s_j,a))w_1(s_j)}{\lambda(x^\star -h(s_j,a))}\Big)' \Big(\sum_{k=1}^m \bar{A}_{k,m}(a)+ \bar{B}_m\Big)}_{\text{(I)}} \\
-\underbrace{\Big(\sum_{k=1}^m \bar{A}_{k,m}(a) \sum_{j=k}^m\frac{\lambda'(x^\star -h(s_j,a))w_1(s_j)}{\lambda(x^\star -h(s_j,a))}\Big)^2 w_2'(a)}_{\text{(II)}}\bigg)
\end{multline*}
Here the denominator is positive, so the sign is determined by the difference of terms (I) and (II). To analyze this, we first establish an inequality for term (I). Note that 
\begin{align*}
    &\big(\bar{A}_{k,m}(a) \sum_{j=k}^m\frac{\lambda'(x^\star -h(s_j,a))w_1(s_j)}{\lambda(x^\star -h(s_j,a))}\big)'\\
    = & \bar{A}_{k,m}'(a) \sum_{j=k}^m\frac{\lambda'(x^\star -h(s_j,a))w_1(s_j)}{\lambda(x^\star -h(s_j,a))} \\
    + &\bar{A}_{k,m}(a) \sum_{j=k}^m \frac{-\lambda''(x^\star -h(s_j,a))w_1^2(s_j)w_2'(a)\lambda(x^\star -h(s_j,a))+(\lambda'(x^\star -h(s_j,a)))^2 w_1^2(s_j)w_2'(a)}{\lambda^2(x^\star -h(s_j,a))} \\
    = & \bar{A}_{k,m}(a) \Big(\sum_{j=k}^m\frac{\lambda'(x^\star -h(s_j,a))w_1(s_j)}{\lambda(x^\star -h(s_j,a))} \Big)^2 w_2'(a)\\
    + & \bar{A}_{k,m}(a) \sum_{j=k}^m \frac{-\lambda''(x^\star -h(s_j,a))w_1^2(s_j)w_2'(a)\lambda(x^\star -h(s_j,a))+(\lambda'(x^\star -h(s_j,a)))^2w_1^2(s_j)w_2'(a)}{\lambda^2(x^\star -h(s_j,a))}.
\end{align*}
The concavity of $\lambda$ implies that its second derivative is negative, and thus we can lower bound the second term in the above expression by $0$. Therefore, we obtain that
\begin{equation}\label{eq:A_k_second_order}
    \big(\bar{A}_{k,m}(a) \sum_{j=k}^m\frac{\lambda'(x^\star -h(s_j,a))w_1(s_j)}{\lambda(x^\star -h(s_j,a))}\big)'\geq \bar{A}_{k,m}(a) \Big(\sum_{j=k}^m\frac{\lambda'(x^\star -h(s_j,a))w_1(s_j)}{\lambda(x^\star -h(s_j,a))} \Big)^2 w_2'(a).
\end{equation}
Applying the above, we can thus show that (I) is greater than (II) concluding the proof:
\begin{align*}
    \text{(I)}=&\big(\sum_{k=1}^m \bar{A}_{k,m}(a) \sum_{j=k}^m\frac{\lambda'(x^\star -h(s_j,a))w_1(s_j)}{\lambda(x^\star -h(s_j,a))}\big)' (\sum_{k=1}^m \bar{A}_{k,m}(a)+ \bar{B}_m) \\
    >&  \big(\sum_{k=1}^m \bar{A}_{k,m}(a) \sum_{j=k}^m\frac{\lambda'(x^\star -h(s_j,a))w_1(s_j)}{\lambda(x^\star -h(s_j,a))}\big)' \sum_{k=1}^m \bar{A}_{k,m}(a) &\text{(using $\bar{B}_m>0$)}\\
    \geq &\bigg(\sum_{k=1}^m \bar{A}_{k,m}(a) \Big(\sum_{j=k}^m\frac{\lambda'(x^\star -h(s_j,a))w_1(s_j)}{\lambda(x^\star -h(s_j,a))} \Big)^2 w_2'(a) \bigg) \bigg( \sum_{k=1}^m \bar{A}_{k,m}(a) \bigg) &\text{(using \eqref{eq:A_k_second_order})}\\
    \geq & \bigg(\sum_{k=1}^m \bar{A}_{k,m}(a) \sum_{j=k}^m\frac{\lambda'(x^\star -h(s_j,a))w_1(s_j)}{\lambda(x^\star -h(s_j,a))}\bigg)^2  w_2'(a) = \text{(II)} &\text{(using Cauchy-Schwarz inequality)}.
\end{align*}
Thus, the derivative $H_m'(a)$ is positive and we conclude that $H_m(a)$ is increasing. 
\end{proof}

\subsection{Steady-state productivity characterization for two states (Proposition~\ref{prop:bdc_prod_n=2})}\label{app:proof_bdc_productivity_n=2}
\begin{proof}[Proof of Proposition~\ref{prop:bdc_prod_n=2}]
From Equation~\eqref{eq:steady_state_prob}, the steady-state productivity can be written as
$$\mathcal{P}(a)= \dfrac{\mu}{\mu+\lambda(e^\star(s_1,a))}\cdot p(e^\star(s_1,a)+h(s_1,a)) + \dfrac{\lambda(e^\star(s_1,a))}{\mu+\lambda(e^\star(s_1,a))}\cdot p(e^\star(s_2,a)+h(s_2,a)).$$
We start from the cases $m=2$ and $m=0$, where we want to show that the productivity (weakly) increases. Proposition~\ref{prop:effort_productivity_expressions} implies that
\begin{itemize}
    \item If $a \in {\cal I}_2$, the productivity at each state remains $p(x^\star)$ and thus $\mathcal{P}(a)=p(x^\star)$ is a constant.
    \item If $a \in {\cal I}_0$, the utility-maximizing effort is zero at each state and thus the productivity $\mathcal{P}(a)$ is increasing in $a$.
\end{itemize}
We then focus on intermediate cases where $a \in {\cal I}_1$, given ${\cal I}_1$ is non-empty.
The steady-state productivity is 
    \begin{align*}
        \mathcal{P}(a)= \dfrac{\mu}{\mu+\lambda(x^\star -h(s_1,a))}\cdot p(x^\star) + \dfrac{\lambda(x^\star -h(s_1,a))}{\mu+\lambda(x^\star -h(s_1,a))}\cdot p(h(s_2,a)).
    \end{align*}
    Using the product rule $(fg)'=f'g+fg'$, the derivative of the productivity function is given by
    \begin{align*}
         \mathcal{P}'(a)&=\dfrac{\mu \lambda'(x^\star -h(s_1,a))h_a(s_1,a)}{(\mu+\lambda(x^\star -h(s_1,a)))^2} \big(p(x^\star)-p(h(s_2,a))\big)+\dfrac{\lambda(x^\star -h(s_1,a))}{\mu+\lambda(x^\star -h(s_1,a))}p'(h(s_2,a))h_a(s_2,a)\\
        &=\dfrac{\big(\mu+\lambda(x^\star -h(s_1,a))\big) \lambda(x^\star -h(s_1,a))p'(h(s_2,a))h_a(s_2,a) -  \mu\lambda'(x^\star -h(s_1,a)) h_a(s_1,a)\big(p(h(s_2,a))-p(x^\star)\big) }{(\mu+\lambda(x^\star -h(s_1,a)))^2}\\
        &=\dfrac{\big(\mu+\lambda(x^\star -h(s_1,a))\big) \lambda(x^\star -h(s_1,a))p'(h(s_2,a))w_1(s_2) -  \mu\lambda'(x^\star -h(s_1,a)) \big(p(h(s_2,a))-p(x^\star)\big)w_1(s_1) }{(\mu+\lambda(x^\star -h(s_1,a)))^2}w_2'(a).
    \end{align*}
    Notice that $\mu+\lambda(x^\star -h(s_1,a))$, $\lambda(x^\star -h(s_1,a))$, and $p'(h(s_2,a))$ are decreasing in AI assistance level $a$, while $\lambda'(x^\star -h(s_1,a))$ and $p(h(s_2,a))-p(x^\star)$ are increasing in AI assistance level $a$, since $\lambda$ is concave and $p$ is increasing and concave. Thus, the numerator term is a decreasing function in $a$. Recall that $\mathcal{P}'_{-}(a):=\lim_{\delta \to 0^-} \frac{\mathcal{P}(a+\delta)-\mathcal{P}(a)}{\delta}$ denotes the left-hand derivative. Similarly, we define  $\mathcal{P}'_{+}(a):=\lim_{\delta \to 0^+} \frac{\mathcal{P}(a+\delta)-\mathcal{P}(a)}{\delta}$ as the right-hand derivative. Evaluating the derivative at the endpoints, we have
    \begin{align*}
         &  \mathcal{P}'_{+}(g(s_2))=\dfrac{\big(\mu+\lambda(x^\star -h(s_1,g(s_2)))\big) \lambda(x^\star -h(s_1,g(s_2)))p'(x^\star)h_a(s_2,g(s_2))}{(\mu+\lambda(x^\star -h(s_1,g(s_2))))^2} > 0 , \\
         &  \mathcal{P}'_{-}(g(s_1))=\dfrac{\big(\mu+\lambda(0)\big) \lambda(0)p'(h(s_2,g(s_1)))h_a(s_2,g(s_1)) -  \mu\lambda'(0)h_a(s_1,g(s_1)) \big(p(h(s_2,g(s_1)))-p(x^\star)\big) }{\big(\mu+\lambda(0)\big)^2} .
    \end{align*}
    The derivative at the left endpoint is strictly positive. The derivative at the right endpoint is negative if and only if
    \begin{align*}
        \big(\mu+\lambda(0)\big) \lambda(0)p'(h(s_2,g(s_1)))h_a(s_2,g(s_1)) -  \mu\lambda'(0)h_a(s_1,g(s_1)) \big(p(h(s_2,g(s_1)))-p(x^\star)\big)  <0,
    \end{align*}
    or equivalently,
    \begin{align*}
        \lambda^2(0)p'(h(s_2,g(s_1)))\frac{h_a(s_2,g(s_1))}{h_a(s_1,g(s_1))} < \mu \Big( \lambda'(0) \big(p(h(s_2,g(s_1)))-p(x^\star)\big)-\lambda(0)p'(h(s_2,g(s_1))) \frac{h_a(s_2,g(s_1))}{h_a(s_1,g(s_1))}\Big).
    \end{align*}
    This is equivalent to
    \begin{align*}
        &\mu >\dfrac{\lambda^2(0)p'(h(s_2,g(s_1)))\frac{h_a(s_2,g(s_1))}{h_a(s_1,g(s_1))}}{  \lambda'(0) \big(p(h(s_2,g(s_1)))-p(x^\star)\big)-\lambda(0)p'(h(s_2,g(s_1))) \frac{h_a(s_2,g(s_1))}{h_a(s_1,g(s_1))}}, \quad \text{and}\\
        & \lambda'(0) \big(p(h(s_2,g(s_1)))-p(x^\star)\big)-\lambda(0)p'(h(s_2,g(s_1))) \frac{h_a(s_2,g(s_1))}{h_a(s_1,g(s_1))}>0,
    \end{align*}
    which corresponds to $\mu>\bar{\mu}_2$ and $\Delta_1>0$, respectively.
    When these conditions hold, there exists a unique $\tau$ such that $\mathcal{P}'(\tau)=0$, and the steady-state productivity $\mathcal{P}(a)$ is increasing on $(g(s_2),\tau]$ and decreasing on $(\tau, g(s_1)]$. Otherwise, the steady-state productivity $\mathcal{P}(a)$ is always increasing.
\end{proof}

\subsection{Model of exogenous skill development (Section~\ref{sec:bdc_discussion})}\label{app:exogenous_transition}
To highlight the contrast in the model of skill development, we analyze an exogenous setting in which the transition rate depends only on the current skill level but not on the endogenous effort. 
Recall that the birth-death chain consists of $N\geq 2$ states $S=(s_1,s_2,\cdots,s_N)$, where $s_1 \leq s_2 \leq \cdots \leq s_N$. 
At each state $s_k$, the agent's effort level is given by 
$e^\star(s_k,a)=\arg\max_{e\geq 0} \left(p(e+h(s_k,a)) - \gamma e\right).$
Skill transitions occur according to the following \emph{exogenous} transition rates between states:
$\mathcal{K}_{s_k s_{k+1}}(a)= \lambda_k \quad\text{and}\quad \mathcal{K}_{s_k s_{k-1}}(a)= \mu,$
where $\lambda_k, \mu > 0$ are upward and downward transition rates, respectively. Note that the transition rates $\lambda_k$ can depend on the current skill level $s_k$. 
Let $\pi_k$ be the steady-state probabilities of an agent being at skill state $s_k$. Under the skill dynamics with exogenous transition rates, Lemma~\ref{lem:bdc_std_result} implies that the steady-state distribution $\boldsymbol{\pi} = (\pi_1,\pi_2, \cdots,\pi_N)$ is exogenous, given by
\begin{equation*}
    \pi_k=\pi_1\prod_{i=1}^{k-1}\frac{\lambda_i}{\mu} \quad\text{and}\quad  \pi_1=\dfrac{1}{1+\sum_{k=2}^{N} \prod_{i=1}^{k-1}\frac{\lambda_i}{\mu}}.
\end{equation*}
In the special case where $\lambda_k=\lambda$ for all $k$, where $\lambda\neq \mu$, this reduces to $\pi_k=\frac{\frac{\lambda}{\mu}-1}{(\frac{\lambda}{\mu})^{N}-1} (\frac{\lambda}{\mu})^{k-1}$.

The steady-state effort and productivity are then given by
\begin{align*}
    & \mathcal{E}(a) := \sum_{k=1}^{N} \pi_k \cdot e^\star(s_k,a) \quad\text{and}\quad  \mathcal{P}(a) := \sum_{k=1}^{N} \pi_k \cdot p(e^\star(s_k,a)+ h(s_k,a)).
\end{align*}
\begin{proposition}\label{prop:bdc_exogenous}
Consider any concave production function $p(x)$ and linear cost function $c(e)=\gamma e$. The agent's steady-state effort $\mathcal{E}(a)$ and productivity $\mathcal{P}(a)$ are given by:
\begin{align*}
    \mathcal{E}(a)=\sum_{i=1}^{N} \pi_i\cdot(x^\star-h(s_i,a))^+ \quad\text{and}\quad \mathcal{P}(a)=\sum_{i=1}^{N} \pi_i \cdot \max(p(x^\star),p(h(s_i,a))).
\end{align*}
\end{proposition}
\begin{proof}[Proof of Proposition~\ref{prop:bdc_exogenous}]   Proposition~\ref{prop:effort_productivity_expressions} implies that
$$e^\star(s_i,a)=(x^\star-h(s_i,a))^+  \quad \text{and} p^\star(s_i,a)=\max (p(x^\star), p(h(s_i,a)) ).$$ 
Then the steady-state effort and productivity follow by definition.
\end{proof}
As a direct corollary, we can obtain the following monotonicity result.
\begin{corollary}
Consider a model of perfect substitution where $h(s,a)=s+a$. The steady-state effort $\mathcal{E}(a)$ (weakly) decreases in $a$ and eventually declines to $0$.
The productivity $\mathcal{P}(a)$ is initially constant and then becomes (weakly) increasing in $a$.
\end{corollary}
With exogenous transition rates, the steady-state can be interpreted as a convex combination of multiple states with exogenous weights. The monotonicity result in the basic model in Section~\ref{ssec:basic} directly applies here: the steady-state effort is decreasing, yet the productivity gain remains.

\section{Supplementary materials for Section~\ref{sec:AI_unreliability}}
\subsection{The derivative of productivity (Lemma~\ref{lem:unreliable_productivity_derivative_origin_sign})}\label{app:proof_sign_lemma}
Recall that Lemma~\ref{lem:unreliable_productivity_derivative_origin_sign} characterizes the sign of the derivative under IARA and DARA conditions. In this section, we provide a general result that holds for a relaxed version of the IARA and DARA conditions.

\paragraph{Relaxation of risk aversion conditions.} We relax the specification of IARA to allow $A(x)$ to be weakly increasing. The standard definition of IARA is restricted to strictly concave functions, while we now allow for a relaxation in which $p'(x)=p''(x)=0$ whenever $x\geq \kappa_p$ for some threshold $\kappa_p>0$; here, the condition $p'(x)=p''(x)=0$ violates strict concavity but it can be viewed as the limiting case of an infinite absolute risk aversion. Similarly, we relax the specification of DARA to allow $A(x)$ to be weakly decreasing, which includes the case of constant absolute risk aversion (CARA). We give the formal definition as follows.
\begin{definition}\label{def:relaxed_ara}
    Given a twice differentiable production function $p$, for an input level $x$, its absolute risk aversion (ARA) is $A(x):=-\frac{p''(x)}{p'(x)}$. 
    We say that $p$ satisfies 
    \begin{enumerate}
    \item Relaxed increasing absolute risk aversion, or relaxed IARA for short: if $A(x)$ is weakly increasing, and points where $p'(x)=p''(x)=0$ are interpreted as having infinite absolute risk aversion.
    \item Relaxed decreasing absolute risk aversion, or relaxed DARA for short: if $A(x)$ is weakly decreasing and $A(x)>0$ holds for all $x$.
    \end{enumerate}
\end{definition}
Recall that the marginal benefit of effort is $M(s,\bar{a},q)=q p'(h(s,\bar{a}))+ (1-q) p'(h(s,0))$.
We are ready to state the following lemma for the relaxed IARA and DARA. The proof requires dealing with boundary conditions more carefully compared to that of the original IARA and DARA versions.
\begin{lemma}[Generalization of Lemma~\ref{lem:unreliable_productivity_derivative_origin_sign} for relaxed IARA/DARA]\label{lem:unreliable_productivity_derivative}
    Consider any linear cost function $c(e)=\gamma e$, and any concave production function $p(x)$ that a) has a skill-AI interaction function $h(s,a)$ that is increasing in both arguments and b) satisfies 1) strict concavity, 2) relaxed IARA, or 3) relaxed DARA. Fix $s$ such that $h(s,0)<x^\star$, $\bar{a}>0$, and $0<q<1$, and write $e^\star:=e^\star (s,\bar{a},q)$. Define $D(s,\bar{a},q):=q \cdot p''(e^\star +h(s,\bar{a})) +(1-q) \cdot p''(e^\star +h(s,0))$. Thus, the derivative of productivity is characterized by the following cases:
    \begin{enumerate}
    \item If $M(s,\bar{a},q) > \gamma$, then the derivative is given by 
    \begin{equation*}
        \begin{cases}
            \frac{\partial p^\star}{\partial \bar{a}} (s,\bar{a},q)=\frac{q(1-q) p'(e^\star +h(s,0)) p'(e^\star +h(s,\bar{a}))h_a(s,\bar{a})}{D(s,\bar{a},q)} \big(A(e^\star +h(s,\bar{a})) - A(e^\star +h(s,0))\big), & \text{if } D(s,\bar{a},q)\neq 0,\\
            \frac{\partial p^\star}{\partial \bar{a}} (s,\bar{a},q)\leq 0, & \text{if } D(s,\bar{a},q)= 0.
        \end{cases}
    \end{equation*}
    \item If $M(s,\bar{a},q) < \gamma$, then the derivative is given by 
    $$\frac{\partial p^\star}{\partial \bar{a}}(s,\bar{a},q) = q \cdot p'(h(s,\bar{a})) h_a(s,\bar{a}).$$
    \end{enumerate}
\end{lemma}
\begin{proof}[Proof of Lemma~\ref{lem:unreliable_productivity_derivative}]
Recall that the expected utility is given by $u(e,s,\bar{a},q)=q p(e+h(s,\bar{a}))+ (1-q) p(e+h(s,0)) - \gamma e$. For given AI proficiency $\bar{a}$, reliability $q$, and skill level $s$, the concavity of $p$ implies that the utility is increasing in $e$ if the marginal benefit $M(s,\bar{a},q)= q p'(e+h(s,\bar{a}))+ (1-q) p'(e+h(s,0)) > \gamma $ and decreasing in $e$ otherwise. The following case distinction characterizes the agent's effort level:

    \textbf{Case 1}: $M(s,\bar{a},q) > \gamma $. 
    The agent's utility-maximizing effort $e^\star (s,\bar{a},q)$ is a non-negative value characterized by the first-order condition
    \begin{equation}\label{eq:exante_foc}
        q \cdot p'(e^\star (s,\bar{a},q)+h(s,\bar{a})) + (1-q) \cdot p'(e^\star (s,\bar{a},q)+h(s,0)) = \gamma .
    \end{equation}
    Note that the derivative of $p$ is continuous, ensuring that $e^\star(s,\bar{a},q)$ is well-defined.  Differentiating both sides of \eqref{eq:exante_foc} with respect to $\bar{a}$, we obtain that \begin{align*}
        & q \cdot p''(e^\star +h(s,\bar{a})) \cdot \left( \frac{\partial e^\star}{\partial \bar{a}} + h_a(s,\bar{a}) \right)+ (1-q) \cdot p''(e^\star +h(s,0)) \cdot \frac{\partial e^\star}{\partial \bar{a}} = 0.
    \end{align*}
    \textbf{Case 1.1:} $D(s,\bar{a},q) \neq 0$. In this case, at least one of $p''(e^\star +h(s,\bar{a}))$ and $p''(e^\star +h(s,0))$ is nonzero, then the derivative of $e^\star$ versus $\bar{a}$ is given by 
    $$\dfrac{\partial e^\star}{\partial \bar{a}} = -\dfrac{q \cdot p''(e^\star +h(s,\bar{a}))}{ q \cdot p''(e^\star +h(s,\bar{a})) + (1-q) \cdot p''(e^\star +h(s,0))} h_a(s,\bar{a}).$$
    We now differentiate $p^\star$ with respect to $\bar{a}$:
    \begin{align*}
    \frac{\partial p^\star}{\partial \bar{a}}(s,\bar{a},q) = & q \cdot p'(e^\star +h(s,\bar{a})) \cdot \left( \frac{\partial e^\star}{\partial \bar{a}} + h_a(s,\bar{a}) \right)+ (1-q) \cdot p'(e^\star +h(s,0)) \cdot \frac{\partial e^\star}{\partial \bar{a}}\\
    = & \Big(q \cdot p'(e^\star +h(s,\bar{a})) + (1-q) \cdot p'(e^\star +h(s,0))\Big) \frac{\partial e^\star}{\partial \bar{a}} + q \cdot p'(e^\star +h(s,\bar{a})) h_a(s,\bar{a})\\
    = & - \big(q \cdot p'(e^\star +h(s,\bar{a})) + (1-q) \cdot p'(e^\star +h(s,0)) \big)\frac{q \cdot p''(e^\star +h(s,\bar{a}))}{q \cdot p''(e^\star +h(s,\bar{a})) +(1-q) \cdot p''(e^\star +h(s,0))} h_a(s,\bar{a}) \\
    +& q \cdot p'(e^\star +h(s,\bar{a})) h_a(s,\bar{a})\\
    = & \frac{q(1-q) \cdot \Big(p'(e^\star +h(s,\bar{a}))p''(e^\star +h(s,0)) - p''(e^\star +h(s,\bar{a})) p'(e^\star +h(s,0))\Big)h_a(s,\bar{a})}{q \cdot p''(e^\star +h(s,\bar{a})) +(1-q) \cdot p''(e^\star +h(s,0))}\\
    =& \frac{q(1-q) p'(e^\star +h(s,0)) p'(e^\star +h(s,\bar{a}))h_a(s,\bar{a})}{q \cdot p''(e^\star +h(s,\bar{a})) +(1-q) \cdot p''(e^\star +h(s,0))} \Big(-\frac{p''(e^\star +h(s,\bar{a}))}{p'(e^\star +h(s,\bar{a}))} + \frac{p''(e^\star +h(s,0))}{p'(e^\star +h(s,0))}\Big)\\
    = &\frac{q(1-q) p'(e^\star +h(s,0)) p'(e^\star +h(s,\bar{a}))h_a(s,\bar{a})}{q \cdot p''(e^\star +h(s,\bar{a})) +(1-q) \cdot p''(e^\star +h(s,0))} \big(A(e^\star +h(s,\bar{a})) - A(e^\star +h(s,0))\big).
    \end{align*}    
    \textbf{Case 1.2}: $D(s,\bar{a},q) = 0$. In this case, both $p''(e^\star +h(s,\bar{a}))$ and $p''(e^\star +h(s,0))$ are zero. This can only occur under either relaxed IARA or DARA: under strict concavity $p''<0$ everywhere. Then the production function should be locally linear around $e^\star +h(s,\bar{a})$ and $e^\star +h(s,0)$. 
    By the definition of relaxed IARA (DARA), the production function can have at most one positively sloped linear segment and at most one constant segment. Thus, it holds that either $0=p'(e^\star +h(s,\bar{a})) < p'(e^\star +h(s,0))=\gamma/(1-q)$, or $p'(e^\star +h(s,\bar{a})) = p'(e^\star +h(s,0))=\gamma$. 
    \begin{itemize}
        \item $0=p'(e^\star +h(s,\bar{a})) < p'(e^\star +h(s,0))=\gamma/(1-q)$: the local linearity implies that $\frac{\partial e^\star}{\partial \bar{a}}(s,\bar{a},q) = 0$ and $\frac{\partial p^\star}{\partial \bar{a}}(s,\bar{a},q) = 0$.
        \item $p'(e^\star +h(s,\bar{a})) = p'(e^\star +h(s,0))=\gamma$: the local linearity implies that $\frac{\partial e^\star}{\partial \bar{a}}(s,\bar{a},q) = -h_a(s,\bar{a})$ and $\frac{\partial p^\star}{\partial \bar{a}}(s,\bar{a},q) < 0$.
    \end{itemize}
    Taken together, the derivative is non-positive, i.e., $\frac{\partial p^\star}{\partial \bar{a}}(s,\bar{a},q) \leq 0$.
    
    \noindent\textbf{Case 2}: $M(s,\bar{a},q) < \gamma $. The agent's utility is decreasing in $e$ for all non-negative effort levels. Thus, the agent's utility-maximizing effort is $e^\star(s,\bar{a},q)=0$. The corresponding productivity
    $p^\star(s,\bar{a},q)=q \cdot p(h(s,\bar{a}))+(1-q)\cdot p(h(s,0))$ is (weakly) increasing in $\bar{a}$.
    The derivative is given by 
    $$\frac{\partial p^\star}{\partial \bar{a}}(s,\bar{a},q) = q \cdot p'(h(s,\bar{a})) h_a(s,\bar{a}).$$
    The above case distinction concludes the proof.
\end{proof}
Equipped with Lemma~\ref{lem:unreliable_productivity_derivative}, Lemma~\ref{lem:unreliable_productivity_derivative_origin_sign} follows directly as a special case.
\begin{proof}[Proof of Lemma~\ref{lem:unreliable_productivity_derivative_origin_sign}]
Our goal is to study the sign of the productivity derivative for strictly concave production functions (which include IARA and DARA). Recall that $D(s,\bar{a},q):=q \cdot p''(e^\star +h(s,\bar{a})) +(1-q) \cdot p''(e^\star +h(s,0))$. 
Unlike their relaxed counterparts in Definition~\ref{def:relaxed_ara}, for strictly concave production functions (which include IARA and DARA conditions of Definition~\ref{def:risk_averse}), the second derivative of $p$ is negative,  which excludes the cases where $D(s,\bar{a},q)=0$. Thus, Lemma~\ref{lem:unreliable_productivity_derivative} gives the following characterization:
\begin{enumerate}
    \item If $M(s,\bar{a},q) > \gamma$, then the derivative is given by $$\frac{\partial p^\star}{\partial \bar{a}}(s,\bar{a},q) =\frac{q(1-q) p'(e^\star +h(s,0)) p'(e^\star +h(s,\bar{a}))}{q \cdot p''(e^\star +h(s,\bar{a})) +(1-q) \cdot p''(e^\star +h(s,0))} \big(A(e^\star +h(s,\bar{a})) - A(e^\star +h(s,0))\big).$$
    As $p$ is a strictly concave function, the sign of the derivative is given by
    $$\operatorname{sgn} \Big(\frac{\partial p^\star}{\partial \bar{a}} (s,\bar{a},q)\Big)=\operatorname{sgn}\Big(-\big(A(e^\star +h(s,\bar{a})) - A(e^\star +h(s,0))\big)\Big).$$
    \item If $M(s,\bar{a},q) < \gamma$, the sign of the derivative is given by $\frac{\partial p^\star}{\partial \bar{a}}(s,\bar{a},q) = q \cdot p'(h(s,\bar{a}))h_a(s,\bar{a})$, which is positive as $p$ is strictly concave and thus strictly increasing, and $h$ is strictly increasing.. Thus, the sign is given by 
    $$\operatorname{sgn}\Big(\frac{\partial p^\star}{\partial \bar{a}} (s,\bar{a},q)\Big)=1.$$
    \end{enumerate} 
    The above case distinction concludes the proof.
\end{proof}

\subsection{Characterization for linear production functions (Lemma~\ref{lem:unreliable_linear})}\label{app:proof_unreliable_linear}
\begin{proof}[Proof of Lemma~\ref{lem:unreliable_linear}]
Under the piecewise linear production function with the perfect substitution skill-AI interaction function $h(s,a)=s+a$, the utility function is given by 
\begin{align*}
    u(e,s,\bar{a},q) &=q\cdot p(s+e+\bar{a}) + (1-q)\cdot p(s+e) -\gamma e\\
    & = q\cdot \min(1, \beta (s+e+\bar{a})) + (1-q)\cdot \min(1,\beta (s+e)) -\gamma e.
\end{align*}
We consider the following cases based on the range of $s+e$ and $s+e+\bar{a}$:
\begin{equation*}  u(e,s,\bar{a},q)=\begin{cases}
        q\cdot \beta (s+e+\bar{a}) + (1-q)\beta (s+e) -\gamma e = (\beta -\gamma) e+ \beta s +q \beta \bar{a}, & \text{if } s+e+\bar{a}\leq \frac{1}{\beta},\\
        q\cdot  \beta + (1-q)\beta (s+e) -\gamma e = ((1-q)\beta -\gamma) e+ (1-q)\beta s +q \beta , & \text{if } s+e \leq \frac{1}{\beta} < s+e+\bar{a},\\
        1 -\gamma e, & \text{if } s+e > \frac{1}{\beta}.
    \end{cases}
\end{equation*}
The utility is linear in $e$ on each of the three segments, with slopes $\beta-\gamma>0$, $(1-q)\beta-\gamma$ and $-\gamma<0$ respectively. Thus, the maximum is attained at one of the two kinks, $(\frac{1}{\beta}-s-\bar{a})^+$ or $(\frac{1}{\beta}-s)^+$, depending on the sign of the slope $(1-q)\beta-\gamma$ of the second segment:
\begin{enumerate}
    \item If $q\leq \frac{\beta -\gamma}{\beta}$, the agent's utility-maximizing effort is $e^\star (s,\bar{a},q)= (\frac{1}{\beta}-s)^+$. The corresponding productivity is given by $p^\star (s,\bar{a},q)=1$.
    \item If $q> \frac{\beta -\gamma}{\beta}$, the agent's utility-maximizing effort is $e^\star (s,\bar{a},q)= (\frac{1}{\beta}-s-\bar{a})^+$. The corresponding productivity is given by
\begin{align*}
    p^\star (s,\bar{a},q) &= q +(1-q)\cdot \min(1,\beta (s+(\frac{1}{\beta}-s-\bar{a})^+))=q + (1-q)\cdot \max\left( \min(1,\beta s),  1-\beta \bar{a}\right).
\end{align*}
\end{enumerate}
\end{proof}

\subsection{Equivalent productivity result for relaxed conditions (Remark~\ref{rmk:relaxed_IARA})}\label{app:proof_unreliable_relaxation}
We establish a version of Theorem~\ref{thm:unreliable_concave3} under the relaxed condition (Definition~\ref{def:relaxed_ara}). Specifically, we provide an extension of the productivity characterization for the relaxed IARA and DARA conditions.
\begin{app_theorem}\label{thm:unreliable_relaxation}
    Consider any concave production function $p(x)$ with a skill-AI interaction function $h(s,a)$ that is increasing in both arguments, and a linear cost function $c(e)=\gamma e$. Suppose $h(s,0)<x^\star$ and $0<q<1$. 
    \begin{enumerate}
        \item If $p$ satisfies relaxed IARA, productivity $p^\star(s, \bar{a}, q)$ is (weakly) decreasing then possibly increasing in $\bar{a}$.
        \item If $p$ satisfies relaxed DARA, productivity $p^\star(s, \bar{a}, q)$ is (weakly) increasing in $\bar{a}$. 
    \end{enumerate}
\end{app_theorem}

\begin{proof}[Proof of Theorem~\ref{thm:unreliable_relaxation}]
Since $p$ is concave and $h(s,0)<x^\star$, there exists a threshold $\tau\in (0, \infty]$ such that the condition $M(s,\bar{a},q) > \gamma$ holds if and only if $\bar{a} < \tau$. Equivalently, 
\begin{equation*}
    \tau:=\sup\{ z\geq 0: q p'(h(s,z))+ (1-q) p'(h(s,0)) \geq  \gamma \}.
\end{equation*}
\begin{enumerate}
\item If $p$ satisfies relaxed IARA, by the relaxed IARA condition, we have $A(e^\star+h(s,\bar{a})) \geq  A(e^\star+h(s,0))$. Lemma~\ref{lem:unreliable_productivity_derivative} implies that the derivative $\frac{\partial p^\star }{\partial \bar{a}} \leq 0$ when $M(s,\bar{a},q) > \gamma$ (i.e., $\bar{a} < \tau$), and $\frac{\partial p^\star}{\partial \bar{a}} \geq 0$ when $q p'(h(s,\bar{a}))+ (1-q) p'(h(s,0)) < \gamma$ (i.e., $\bar{a} > \tau$). Thus, the productivity is decreasing-then-increasing in $\bar{a}$, with the increasing region degenerate when $\tau=\infty$. 
\item If $p$ satisfies relaxed DARA, Lemma~\ref{lem:unreliable_productivity_derivative} implies that the derivative is always positive. Thus, the productivity $p^\star(s, \bar{a}, q)$ is increasing in $\bar{a}$.
\end{enumerate}
Combining both cases, we complete the proof.
\end{proof}

\subsection{Equivalent ratio result for relaxed IARA (Remark~\ref{rmk:relaxed_IARA})}\label{app:exante_linear_bad}
We establish a version of Proposition~\ref{prop:exante_linear_bad} under the relaxed IARA condition (Definition~\ref{def:relaxed_ara}). Specifically, we provide a similar arbitrarily small ratio that holds with a slight perturbation of the piecewise linear instance that also satisfies the relaxed IARA condition.
\begin{proposition}\label{prop:exante_relax_bad}
    For any $\varepsilon>0$, there exists a concave production function $p(x)$ with a perfect substitution skill-AI interaction function $h(s,a)=s+a$, a linear cost function $c(e)$, a skill level $s$, and a reliability $q\in(0,1)$, where $p$ satisfies the relaxed IARA condition, such that the ratio of productivity at a higher AI assistance level to that at a lower assistance level satisfies $$\inf_{a_\ell <a_h}  \dfrac{p^\star(s,a_h, q)}{p^\star(s, a_\ell , q)} \leq \varepsilon.$$
\end{proposition}
\begin{proof}[Proof of Proposition~\ref{prop:exante_relax_bad}]
    We start from the piecewise linear production function $p(x)=\min(1,\beta x)$ with a perfect substitution skill-AI interaction function $h(s,a)=s+a$, and $c(x)=\gamma x$ where $\beta >\gamma$, and construct a small perturbation to satisfy relaxed IARA. Note that $p(x)$ is not differentiable at $x=1/\beta$. We construct a perturbation $\Tilde{p}(x)$ of $p(x)$ such that $\Tilde{p}(x)$ is concave, differentiable, relaxed IARA, while maintaining the desired property. Choose a sufficiently small $\delta<\frac{1}{2 \beta}$ and define $\Tilde{p}(x)=p(x)-\eta(x)$, where $\eta(x)$ satisfies:
    \begin{itemize}
    \item $\eta(x) = 0$ for $|x-1/\beta | \geq \delta$;
    \item $ 0 \leq \eta(x) \leq \delta \text{ for } x \in [1/\beta -\delta, 1/\beta +\delta]$;
    \item $\eta(x)$ is increasing and convex on $[\frac{1}{\beta}-\delta,\frac{1}{\beta}]$, and decreasing and convex on $[\frac{1}{\beta},\frac{1}{\beta}+\delta]$.
    \end{itemize}
    Specifically, by choosing $\eta(x)$ to be a quadratic function on $[1/\beta -\delta, 1/\beta +\delta]$, the resulting function $\Tilde{p}(x)$ is concave, differentiable, and satisfies the relaxed IARA condition.
    
    We then consider the utility-maximizing effort under perturbation for low and high proficiency, which we select shortly.
    The utility function under perturbation, $\Tilde{u}(e,s,\bar{a},q)$, is given by
     \begin{align*}
        &\Tilde{u}(e,s,\bar{a},q)
        = q \cdot \Tilde{p}(s+e+\bar{a})+ (1-q) \cdot \Tilde{p}(s+e) -\gamma e\\
        =& q \cdot \min\left(1,\beta (s+e+\bar{a})\right) + (1-q) \min\left(1,\beta (s+e)\right) -q \eta(s+e+\bar{a})-(1-q) \eta(s+e) -\gamma e.
    \end{align*}

    First, set skill level $s=0$ and low proficiency $a_\ell =0$, the corresponding utility becomes
    \begin{align*}
        \Tilde{u}(e, s, a_\ell , q)= \Tilde{u}(e,0,0,q)  & =\min \left(1,\beta e\right) - \eta(e)-\gamma e.
    \end{align*}
    By construction, the utility-maximizing effort is within $[\frac{1}{\beta}-\delta,\frac{1}{\beta}+\delta]$.

    Second, set skill level $s=0$ and high proficiency  $a_h=\frac{1}{\beta}$, the corresponding utility becomes
    \begin{align*}
        \Tilde{u}(e,s,a_h,q)&=\Tilde{u}(e,0,\frac{1}{\beta},q)= q  + (1-q) \min \left(1,\beta e\right) -q \eta(e+\frac{1}{\beta})-(1-q) \eta(e) -\gamma e\\
        &=\begin{cases}
            - \big(\gamma -(1-q) \beta \big) e+q-q \eta(e+\frac{1}{\beta}), & \text{if } e\leq \delta,\\
            - \big(\gamma -(1-q) \beta \big) e+q , & \text{if } \delta <e\leq \frac{1}{\beta}-\delta,\\
            - \big(\gamma -(1-q) \beta \big) e+q-(1-q) \eta(e) , & \text{if }  \frac{1}{\beta}-\delta<e\leq \frac{1}{\beta},\\
            -\gamma e + 1 -(1-q) \eta(e) , & \text{otherwise}.
        \end{cases}
    \end{align*}
    Based on the case distinction in the last equation, we have the following analysis.
    \begin{itemize}
    \item $e\leq \delta$. In this region, the utility is at least the utility at $e=0$, which is given by $q-q \eta(\frac{1}{\beta})\geq q-q\delta$.
    \item $\delta <e\leq \frac{1}{\beta}-\delta$. The utility is decreasing in $e$. Then the utility-maximizing effort in this region is $e=\delta$.
    \item  $\frac{1}{\beta}-\delta<e\leq \frac{1}{\beta}$. In this region, the utility is upper bounded by
    $q - \big(\gamma -(1-q) \beta \big) ( \frac{1}{\beta}-\delta)$.
    \item $e>\frac{1}{\beta}$. In this region, the utility is upper bounded by
    $\frac{\beta -\gamma}{\beta}$.
    \end{itemize}
    Consider a reliability
    $q>\frac{\beta -\gamma}{\beta}$ and a sufficiently small perturbation size $\delta$, we have
    $$q-q\delta>\max\left(q - \big(\gamma -(1-q) \beta \big) ( \frac{1}{\beta}-\delta), \frac{\beta -\gamma}{\beta}\right).$$
    This implies that the utility-maximizing effort of $\Tilde{u}(e,s,a_h,q)$ lies on $[0,\delta]$, given $s=0$ and $a_h=\frac{1}{\beta}$. 
     
    Finally, we are ready to establish the productivity ratio under perturbation. Recall that productivity is the utility plus cost, where  
    $$\Tilde{u}(e,0,\bar{a},q) + \gamma e= q \cdot \min\left(1,\beta (e+\bar{a})\right) + (1-q) \min\left(1,\beta e\right) -q \eta(e+\bar{a})-(1-q) \eta(e).$$
    For the constructed function $\Tilde{p}$, selected skill level $s=0$, proficiency levels $a_\ell=0$ and $a_h=\frac{1}{\beta}$, reliability $q>\frac{\beta -\gamma}{\beta}$, and a sufficiently small perturbation size $\delta$, we have 
    $$p^\star (0, 0, q)\geq 1-\delta \quad \text{and} \quad p^\star (0,\frac{1}{\beta}, q) \leq q + (1-q)\beta\delta.$$
    Thus, the productivity ratio is bounded by
    $$\frac{p^\star (0,\frac{1}{\beta}, q)}{p^\star (0, 0, q)}\leq \frac{q + (1-q)\beta\delta}{1-\delta}.$$
    This ratio can be made smaller than $\varepsilon$ by choosing  $\frac{\beta -\gamma}{\beta}< q<\frac{\varepsilon}{2}$ and $\delta<\frac{\varepsilon}{2(\beta+\varepsilon)}$.
\end{proof}

\subsection{Production functions that partially satisfy IARA (Section~\ref{sec:unreliability_discussion})}\label{app:IARA_examples}
\begin{example}[Translog, \cite{berndt1973translog}]
    Consider $p(x)=\log x + b (\log x)^2$, where $x>0$, $b> 0$. For the following regime of $x$: $$\max(2b,1) < 2b \log x +1 < (\sqrt{5}+1) b,$$ the production function $p(x)$ is nonnegative, increasing, concave, and satisfies IARA.
\end{example}
\begin{proof}
    After some algebraic calculations, we have
    \begin{align*}
        & p'(x)=\dfrac{1+2b \log x}{x}\quad \text{and}\quad p''(x)=\dfrac{2b-1-2b \log x}{x^2}.
    \end{align*}
    Then the absolute risk aversion and its derivative is given by
    \begin{align*}
        & A(x) = -\frac{p''(x)}{p'(x)} = \dfrac{1+2b \log x-2b}{x(1+2b \log x)} \quad \text{and}\quad A'(x) = \dfrac{-(1+2b \log x)^2 + 2b(1+2b \log x)+4b^2}{x^2(1+2b \log x)^2}.
    \end{align*}
    When $x>0$, $b > 0$, the condition $2b \log x +1 > 2b \ (> 0)$ implies $p'(x)>0$ and $p''(x)<0$; the condition $2b < 2b \log x +1 < (\sqrt{5}+1) b$ implies $A'(x)>0$; besides, the condition $2b \log x +1 > 1$ implies $x>1$ and thus $p(x)>0$. Thus, the production function $p(x)$ is nonnegative, increasing, concave, and satisfies IARA.
\end{proof}
\begin{example}[Transcendental multiproduct, \cite{mundlak1964transcendental}]
    Consider $p(x)=e^{ax}x^b$, where $x>0$, $a<0$ and $b > 0$. For the following regime of x:   
    $$0< ax+b <\sqrt{b}\quad \text{and} \quad (ax+b)^2>2ax+b,$$ the production function $p(x)$ is nonnegative, increasing, concave, and satisfies IARA.
\end{example}
\begin{proof}
    After some algebraic calculations, the first and second derivatives are given by
    \begin{align*}
         p'(x)=e^{ax}x^b\dfrac{ax+b}{x} \quad \text{and}
        \quad p''(x)=e^{ax}x^b\dfrac{(ax+b)^2-b}{x^2}.
    \end{align*}
    Then the absolute risk aversion and its derivative are given by
    \begin{align*}
        A(x) = -\frac{p''(x)}{p'(x)} = \dfrac{b-(ax+b)^2}{x(ax+b)}\quad \text{and} \quad
        A'(x) = \dfrac{b\big((ax+b)^2-(2ax+b)\big)}{x^2(ax+b)^2}.
    \end{align*}
    When $x>0$, $a<0$, $b > 0$, the condition $0< ax+b <\sqrt{b}$ implies $p'(x)>0$ and $p''(x)<0$; the condition $(ax+b)^2>2ax+b $ implies $A'(x)>0$. Thus, the production function $p(x)$ is nonnegative, increasing, concave, and satisfies IARA.
\end{proof}

\subsection{Model of full adaptation under unreliability (Section~\ref{sec:unreliability_discussion})}\label{app:full_adapt}
To highlight the contrast in the model of AI unreliability, we analyze a full adaptation setting in which the agent has perfect ability to evaluate and respond to the AI output.
Recall that we model $a$ as the AI assistance level and $q$ as the reliability level, such that the AI functions desirably with a level $\bar{a}$ with probability $q>0$, i.e., $\mathbb{P}(a=\bar{a})=q$, and fails to provide assistance with probability $1-q$, i.e., $\mathbb{P}( a=0)=1-q$.
Here, we consider an agent with \emph{full adaptation} who perfectly observes the AI’s output and chooses effort accordingly. That means that the agent
observes $a=\bar{a}$ with probability $q>0$ and exerts $e^\star(s,\bar{a})$; else the agent observes $a=0$ and exerts $e^\star(s,0)$.
The expected effort and productivity with full adaptation are thus given by
\begin{align*}
    e_{\text{full}}(s, \bar{a}, q)&= q e^\star(s,\bar{a}) + (1-q)  e^\star(s,0)=q(e^\star(s,\bar{a}) -  e^\star(s,0)) + e^\star(s,0), \\
    p_{\text{full}}(s, \bar{a}, q)&= q \cdot  p(e^\star(s,\bar{a})+h(s,\bar{a})) +   (1-q) \cdot p(e^\star(s,0)+h(s,0))\\
    &=q \big(p(e^\star(s,\bar{a})+h(s,\bar{a}))-p(e^\star(s,0)+h(s,0))\big)+  p(e^\star(s,0)+h(s,0)).
\end{align*}
\begin{proposition}\label{prop:full_adapt}
Consider any concave production function $p(x)$ and linear cost function $c(e)=\gamma e$. The agent's expected effort $e_{\text{full}}(s, \bar{a}, q)$ is (weakly) decreasing and productivity $p_{\text{full}}(s, \bar{a}, q)$ is (weakly) increasing in $s$, $\bar{a}$ and $q$. 
\end{proposition}
\begin{proof}[Proof of Proposition~\ref{prop:full_adapt}]
Recall that the effort and productivity are given by
\begin{align*}
    e_{\text{full}}(s, \bar{a}, q) &= q (x^\star-h(s,\bar{a}))^+ + (1-q) \max (0, x^\star-h(s,0))\\
    &=(x^\star-h(s,0))^+ -q\big((x^\star-h(s,0))^+- (x^\star-h(s,\bar{a}))^+\big),\\
    p_{\text{full}}(s, \bar{a}, q)&= q \cdot \max\left( p(x^\star), p(h(s,\bar{a})) \right) +   (1-q) \cdot \max\left( p(x^\star), p(h(s,0)) \right) \\
    &=q \big( \max\left( p(x^\star), p(h(s,\bar{a})) \right) - \max\left( p(x^\star), p(h(s,0)) \right)\big) + \max \left( p(x^\star), p(h(s,0)) \right).
\end{align*}
Thus, effort $e_{\text{full}}(s, \bar{a}, q)$ is (weakly) decreasing and productivity $p_{\text{full}}(s, \bar{a}, q)$ is (weakly) increasing in $s$ and $\bar{a}$. Note that $(x^\star-h(s,0))^+ \geq (x^\star-h(s,\bar{a}))^+$ and $\max\left( p(x^\star), p(h(s,\bar{a})) \right) \geq \max\left( p(x^\star), p(h(s,0)) \right)$, then effort $e_{\text{full}}(s, \bar{a}, q)$ is (weakly) decreasing and productivity $p_{\text{full}}(s, \bar{a}, q)$ is (weakly) increasing in $q$.
\end{proof}
Proposition~\ref{prop:full_adapt} is consistent with Proposition~\ref{prop:effort_productivity_expressions}. With full adaptation, the model can be interpreted as a convex combination of multiple cases based on AI output, which preserves the standard result in Proposition~\ref{prop:effort_productivity_expressions}.

\section{Supplementary materials for Section~\ref{sec:AI_literacy}}

\subsection{No skill polarization with constant verification accuracy (Section~\ref{sec:literace_example})}\label{app:literacy_constant_v}
To highlight the contrast in the model of AI literacy, we analyze a setting in which the agents have the same verification ability to evaluate the AI output. Under an additional constant verification assumption, we establish the monotonicity of effort in skill and the unimodality of skill distribution. Recall that in this section we assume a perfect substitution skill-AI interaction function $h(s,a)=s+a$. Thus, we omit this from the lemmas and propositions.

We begin with a lemma that examines the utility-maximizing effort in the ex-ante model of Section~\ref{sec:AI_unreliability}.
\begin{lemma}\label{lem:exante_effort}
    Consider any concave production function $p(x)$ and linear cost function $c(e)=\gamma e$. The agent's utility-maximizing effort $e^\star(s,\bar{a},q)$ is (weakly) decreasing in $s, \bar{a},q$.
\end{lemma}
\begin{proof}[Proof of Lemma~\ref{lem:exante_effort}]
     Recall that 
     $e^\star(s,\bar{a},q):=\arg\max_{e\geq 0}\left(q\cdot p(s+e+\bar{a}) + (1-q)\cdot p(s+e)-\gamma e\right)$.
     The agent's utility-maximizing effort $e^\star (s,\bar{a},q)$ is either zero or a non-negative value characterized by the following first-order condition:
     $q \cdot p'(s+e^\star (s,\bar{a},q)+\bar{a}) + (1-q) \cdot p'(s+e^\star (s,\bar{a},q)) = \gamma$.
    In both cases, we have 
    \begin{equation*}
        e^\star (s,\bar{a},q)= \sup \left\{e\geq 0: q \cdot p'(s+e+\bar{a}) + (1-q) \cdot p'(s+e)\geq \gamma\right\}.
    \end{equation*} 
    The concavity of $p$ implies that the derivative function $q \cdot p'(s+e+\bar{a}) + (1-q) \cdot p'(s+e)$ is decreasing in $s,\bar{a}, q$. Thus, an increase in any of $s,\bar{a}, q$ shrinks the set over which the supremum is defined and, as a result, $e^\star (s,\bar{a},q)$ is (weakly) decreasing in $s,\bar{a}, q$.
\end{proof}

\begin{proposition}\label{prop:unreliable_constant_effort}
    Consider any concave production function $p(x)$ and linear cost function $c(e)=\gamma e$. If the verification accuracy is a constant, i.e., $v(s)=v \geq 1/2$, the expected effort $e^\star_v(s, \bar{a}, q)$ is (weakly) decreasing in $s$ and $q$.
\end{proposition}
\begin{proof}[Proof of Proposition~\ref{prop:unreliable_constant_effort}]
Recall that the posterior probabilities are given by
$$q_{\bar{a}}(q,s)=\frac{q v(s)}{q v(s)+(1-q)(1-v(s))} \quad \text{ and }\quad q_0(q,s)= \frac{q (1-v(s))}{q(1-v(s))+(1-q)v(s)}.$$
Note that $q_{\bar{a}}$ and $q_0$ depend on $q$. We make this dependence explicit here, though it was suppressed previously.
Given $v(s)=v\geq \frac{1}{2}$, we know that $q_{\bar{a}}$ and $q_0$ are independent of $s$, so we suppress its dependence on $s$ for now. 
Recall the expected effort is $e^\star_v(s, \bar{a}, q)=\mathbb P(\hat{a}=\bar{a}|s)e^\star(s,\bar{a},q_{\bar{a}})+\mathbb P(\hat{a}=0|s) e^\star(s,\bar{a},q_0)$, we have
\begin{align*}
    e^\star_v(s, \bar{a}, q)&=\big((2q-1)v(s) + 1-q\big) e^\star(s,\bar{a},q_{\bar{a}}(q)) + \big(1- ((2q-1)v(s) + 1-q)\big) e^\star(s,\bar{a},q_0(q))   \\
    &=((2q-1)v + 1-q) e^\star(s,\bar{a},q_{\bar{a}}(q)) + \big(1- ((2q-1)v + 1-q)\big) e^\star(s,\bar{a},q_0(q))   .
\end{align*}
The signal probabilities are independent of $s$.
Besides, effort $e^\star(s,\bar{a},q_{\bar{a}}(q))$ and $e^\star(s,\bar{a},q_0(q))$ are both decreasing in $s$ (Lemma~\ref{lem:exante_effort}), it follows that the expected effort $e^\star_v(s, \bar{a}, q)$ is decreasing in $s$.

As the utility-maximizing effort $e^\star(s,\bar{a},q)$ is decreasing in $q$ (Lemma~\ref{lem:exante_effort}), and it holds that $q_{\bar{a}}(q)\geq q_0(q)$, we have $e^\star(s,\bar{a},q_{\bar{a}}(q))\leq e^\star(s,\bar{a},q_0(q))$. Besides, note that $q_{\bar{a}}(q)$ and $q_0(q)$ are both increasing in $q$ by definition. Hence, for any $0<q<q'<1$,
\begin{align*}
    e^\star_v(s, \bar{a}, q) &= \Big(q(2v-1) + 1-v\Big) e^\star(s,\bar{a},q_{\bar{a}}(q)) + \Big(1- \big(q(2v-1) + 1-v\big)\Big)  e^\star(s,\bar{a},q_0(q))\\
    &\geq \Big(q(2v-1) + 1-v\Big) e^\star(s,\bar{a},q_{\bar{a}}(q')) + \Big(1- \big(q(2v-1) + 1-v\big)\Big)  e^\star(s,\bar{a},q_0(q'))\\
    &\geq \Big(q'(2v-1) + 1-v\Big) e^\star(s,\bar{a},q_{\bar{a}}(q')) + \Big(1- \big(q'(2v-1) + 1-v\big)\Big)  e^\star(s,\bar{a},q_0(q')) = e^\star_v(s, \bar{a}, q').
\end{align*}
As a result, the expected effort $ e^\star_v(s, \bar{a}, q)$ is decreasing in $q$.
\end{proof}

\begin{proposition}\label{prop:constant_verification_unimodality}
    Consider any concave production function $p(x)$, linear cost function $c(e)=\gamma e$, and increasing transition function $\lambda(e)$. If the verification accuracy is a constant, i.e., $v(s)=v \geq 1/2$, then for any $\bar{a}\geq 0$, $\mu>0$, and $S\in {\cal S}$, the steady-state skill distribution $\boldsymbol{\pi}(\mu, \bar{a},q,S)$ is unimodal.
\end{proposition}
\begin{proof}[Proof of Proposition~\ref{prop:constant_verification_unimodality}]
    Proposition~\ref{prop:unreliable_constant_effort} implies that the expected effort $ e^\star_v(s, \bar{a}, q)$ is decreasing in skill $s$. Combining with Lemma~\ref{lem:bayes_effort_decrease}, the skill distribution $\boldsymbol{\pi}(\mu, \bar{a},q,S)$ is unimodal.
\end{proof}

\subsection{Effort-skill reversal (Lemma~\ref{lem:effort_caseA+B})}\label{app:proof_effort_caseA+B}
We first derive a closed-form expression of the expected effort $e^\star_v(s,\bar{a},q)$ and then prove Lemma~\ref{lem:effort_caseA+B}. Specifically, the critical case, where $\omega\geq 0$ and $q<\frac{1}{2}$, is handled by two auxiliary lemmas (Lemmas~\ref{lem:effort_case_pos_small} and~\ref{lem:effort_case_pos_large}).
\paragraph{Expression of the expected effort.}
Recall that the expected effort is
$$e^\star_v(s, \bar{a}, q)= \mathbb P(\hat{a}=\bar{a}|s)e^\star(s,\bar{a},q_{\bar{a}}(s))+\mathbb P(\hat{a}=0|s) e^\star(s,\bar{a},q_0(s)),$$
where the signal probabilities are $\mathbb P(\hat{a}=\bar{a}|s)=q v(s) + (1-q)(1-v(s))$ and $\mathbb P(\hat{a}=0|s)=q(1-v(s)) + (1-q)v(s)$.
\begin{itemize}
\item When $\omega<0$, the posterior reliability of receiving signal $\hat{a}=0$, $q_0(\cdot)$, crosses the threshold $\frac{\beta-\gamma}{\beta}$ at $\Tilde{s}= v^{-1}(\frac{\frac{q}{1-q}\frac{\gamma}{\beta -\gamma}}{\frac{q}{1-q}\frac{\gamma}{\beta -\gamma}+1})$, see Figure~\ref{fig:paradox_cases_A}. 
In this case, both $q_0$ and $q_{\bar{a}}$ are greater than $\frac{\beta-\gamma}{\beta}$ when $s<\Tilde{s}$; while $q_0\leq \frac{\beta-\gamma}{\beta}$ when $s\geq \Tilde{s}$.
\item When $\omega\geq 0$, the posterior reliability of receiving signal $\hat{a}=\bar{a}$, $q_{\bar{a}}(\cdot)$ crosses the threshold $\frac{\beta-\gamma}{\beta}$ at $\Tilde{s}= v^{-1}(\frac{1}{\frac{q}{1-q}\frac{\gamma}{\beta -\gamma}+1})$, see Figure~\ref{fig:paradox_cases_B}.
In this case, both $q_0$ and $q_{\bar{a}}$ are at most $\frac{\beta-\gamma}{\beta}$ when $s<\Tilde{s}$; while $q_{\bar{a}}\geq \frac{\beta-\gamma}{\beta}$ when $s\geq \Tilde{s}$.
\end{itemize}
We compute the two signal-specific efforts by applying Lemma~\ref{lem:unreliable_linear} with the posterior reliability in place of the prior $q$ and aggregate, obtaining
\begin{equation}\label{eq:effort_v}
    e^\star_v(s, \bar{a}, q)=
    \begin{cases}
        (\frac{1}{\beta} -s-\bar{a})^+, & \text{if } s<\Tilde{s} \text{ and } \omega<0,\\
        (\frac{1}{\beta} -s)^+, & \text{if } s<\Tilde{s} \text{ and } \omega\geq 0,\\
        \big((2q-1)v(s)+1-q\big) (\frac{1}{\beta} -s-\bar{a})^+ + \big((1-2q)v(s)+q\big)(\frac{1}{\beta} -s)^+, & \text{if }  s\geq \Tilde{s}.
    \end{cases}
\end{equation}

\paragraph{A reduction for the case $\omega\geq 0$.} 
We restrict attention to the case of $\Tilde{s}<\frac{1}{\beta}$; otherwise, the expected effort $ e^\star_v(s, \bar{a}, q)=(\frac{1}{\beta} -s)^+$ is (weakly) decreasing. Moreover, since $ e^\star_v(s, \bar{a}, q)$ is (weakly) decreasing on $(0,\Tilde{s})$ and has a downward drop at $\Tilde{s}$, we have the following equivalence:
\begin{equation}\label{eq:reduction}
    e^\star_v(\cdot,\bar{a},q) \text{ is non-monotonic in } s \iff e^\star_v(\cdot,\bar{a},q) \text{ is not (weakly) decreasing for } s>\Tilde{s}.
\end{equation}
Then for $\omega\geq 0$ and $\Tilde{s}<\frac{1}{\beta}$, it suffices to focus on the region $s>\Tilde{s}$.

\begin{lemma}\label{lem:effort_case_pos_small}
Consider $p(x)=\min(1, \beta x)$, $c(e)=\gamma e$, where $\beta>\gamma$. Suppose $v(s)\in [1/2,1]$ is (weakly) increasing and concave with $v(0)=\frac{1}{2}$. Suppose $\omega\geq 0$, $q<\frac{1}{2}$, and $\Tilde{s}<\frac{1}{\beta}$.
\begin{enumerate}
\item If $\Tilde{s} + \frac{(1-2q) v(\Tilde{s}) + q}{(1-2q) v'(\Tilde{s})} < \frac{1}{\beta}$, then there exists a threshold $\Tilde{a}\in(0,\frac{1}{\beta}-\Tilde{s})$ such that the expected effort $e^\star_v(s,\bar{a},q)$ is non-monotonic in $s$ for any $\Tilde{a}<\bar{a}<\frac{1}{\beta}-\Tilde{s}$ and decreasing in $s$ for any $\bar{a}\leq \Tilde{a}$.
\item Otherwise, for any  $\bar{a}<\frac{1}{\beta}-\Tilde{s}$, effort $ e^\star_v(s,\bar{a},q)$ is decreasing in $s$.
\end{enumerate}
\end{lemma}

\begin{lemma}\label{lem:effort_case_pos_large}
Consider $p(x)=\min(1, \beta x)$, $c(e)=\gamma e$, where $\beta>\gamma$. Suppose $v(s)\in [1/2,1]$ is (weakly) increasing and concave with $v(0)=\frac{1}{2}$. Suppose $\omega\geq 0$, $q<\frac{1}{2}$, and $\Tilde{s}<\frac{1}{\beta}$.
\begin{enumerate}
\item If $\Tilde{s} + \frac{(1-2q) v(\Tilde{s}) + q}{(1-2q) v'(\Tilde{s})} < \frac{1}{\beta}$, then for any $\bar{a}\geq\frac{1}{\beta}-\Tilde{s}$, the expected effort $e^\star_v(s,\bar{a},q)$ is non-monotonic in $s$.
\item Otherwise, for any $\bar{a}\geq\frac{1}{\beta}-\Tilde{s}$, effort $ e^\star_v(s,\bar{a},q)$ is decreasing in $s$.
\end{enumerate}
\end{lemma}

\begin{proof}[Proof of Lemma~\ref{lem:effort_caseA+B}]
We distinguish the sign of the marginal return gap $\omega=(\beta-\gamma)-q\beta$.

\textbf{Case $\omega<0$.} In this case, Condition~\ref{cond:multimodal_cond} reduces to $\Tilde{s}<\frac{1}{\beta}$, and our goal is to show that the threshold is $\Tilde{a}=0$.
1) If $\Tilde{s}\geq \frac{1}{\beta}$, then $(\frac{1}{\beta} -s)^+=0$ for all $s\geq\Tilde{s}$, so \eqref{eq:effort_v} gives $e^\star_v(s, \bar{a}, q) =(\frac{1}{\beta} -s-\bar{a})^+$, which is decreasing in $s$.
2) If $\Tilde{s}<\frac{1}{\beta}$ and $\bar{a}=0$, \eqref{eq:effort_v} gives $e^\star_v(s, \bar{a}, q) =(\frac{1}{\beta} -s)^+$, which is decreasing in $s$.
3) If $\Tilde{s}<\frac{1}{\beta}$ and $\bar{a}>0$, the expected effort $e^\star_v(s, \bar{a}, q) =(\frac{1}{\beta} -s-\bar{a})^+$ on $[0,\Tilde{s})$ and jumps upward at $s=\Tilde{s}$: the signal $\hat{a}=0$ is received with positive probability, and $(\frac{1}{\beta}-\Tilde{s})^+ > (\frac{1}{\beta} -\Tilde{s}-\bar{a})^+$. Since the expected effort is positive at $s=\Tilde{s}$ and vanishes for $s\geq\frac{1}{\beta}$, it is not increasing either. It follows that the expected effort is non-monotonic in $s$.

\textbf{Case $\omega\geq 0$.} In this case, if $\Tilde{s}\geq\frac{1}{\beta}$, then Condition~\ref{cond:multimodal_cond} fails, and \eqref{eq:effort_v} gives $e^\star_v(s, \bar{a}, q) =(\frac{1}{\beta} -s)^+$, which is decreasing in $s$. Hence, we assume $\Tilde{s}<\frac{1}{\beta}$ and use the reduction \eqref{eq:reduction}.

\textbf{Sub-case $q\geq\frac{1}{2}$.} In this case Condition~\ref{cond:multimodal_cond} fails, and our goal is to show that the expected effort is decreasing in $s$ for every $\bar{a}\geq 0$. 
For $s>\Tilde{s}$, \eqref{eq:effort_v} implies that the expected effort is
\begin{align*}
e^\star_v(s, \bar{a}, q)=\big((2q-1)v(s)+1-q\big) (\frac{1}{\beta} -s-\bar{a})^+ + \big((1-2q)v(s)+q\big)(\frac{1}{\beta} -s)^+.
\end{align*}
Both $(\frac{1}{\beta} -s-\bar{a})^+$ and $(\frac{1}{\beta} -s)^+$ are decreasing in $s$, and as $q\geq\frac{1}{2}$, the weight $(2q-1)v(s)+1-q$ on the smaller of the two is (weakly) increasing in $s$. Hence, for any $\Tilde{s}<s<s'$, we have 
\begin{align*}
    e^\star_v(s, \bar{a}, q)  &= \big((2q-1)v(s)+1-q\big) (\tfrac{1}{\beta} -s-\bar{a})^+ + \big((1-2q)v(s)+q\big)(\tfrac{1}{\beta} -s)^+ \\
    &\geq \big((2q-1)v(s')+1-q\big) (\tfrac{1}{\beta} -s-\bar{a})^+ + \big((1-2q)v(s')+q\big)(\tfrac{1}{\beta} -s)^+\\
    &\geq \big((2q-1)v(s')+1-q\big) (\tfrac{1}{\beta} -s'-\bar{a})^+ + \big((1-2q)v(s')+q\big)(\tfrac{1}{\beta} -s')^+ \geq e^\star_v(s', \bar{a}, q).
\end{align*}
Thus, the expected effort is non-increasing for $s>\Tilde{s}$, and by \eqref{eq:reduction} it is monotonic in $s$.

\textbf{Sub-case $q<\frac{1}{2}$.} Condition~\ref{cond:multimodal_cond} is now equivalent to $\frac{1}{\beta}-\Tilde{s}>  \frac{(1-2q) v(\Tilde{s}) + q}{(1-2q) v'(\Tilde{s})}$.
Lemmas~\ref{lem:effort_case_pos_small} and~\ref{lem:effort_case_pos_large} characterize the monotonicity of the expected effort for $\bar{a}<\frac{1}{\beta}-\Tilde{s}$ and for $\bar{a}\geq\frac{1}{\beta}-\Tilde{s}$, respectively.
If Condition~\ref{cond:multimodal_cond} holds, Lemma~\ref{lem:effort_case_pos_small} shows there exists a threshold $\Tilde{a}\in(0,\frac{1}{\beta}-\Tilde{s})$ such that the expected effort is non-monotonic in $s$ for $\Tilde{a}<\bar{a}<\frac{1}{\beta}-\Tilde{s}$ and decreasing in $s$ for $\bar{a}\leq\Tilde{a}$, while Lemma~\ref{lem:effort_case_pos_large} shows that it is non-monotonic in $s$ for every $\bar{a}\geq\frac{1}{\beta}-\Tilde{s}$. Together, the expected effort is non-monotonic in $s$ if $\bar{a}>\Tilde{a}$ and decreasing in $s$ if $\bar{a}\leq\Tilde{a}$. If Condition~\ref{cond:multimodal_cond}  fails, the two lemmas show that the expected effort is decreasing in $s$ for every $\bar{a}\geq 0$.

Combining the cases, if Condition~\ref{cond:multimodal_cond} holds there exists a threshold $\Tilde{a}\geq 0$ (with $\Tilde{a}=0$ when $\omega<0$) such that the expected effort is non-monotonic in $s$ if $\bar{a}>\Tilde{a}$ and decreasing in $s$ if $\bar{a}\leq\Tilde{a}$; if Condition~\ref{cond:multimodal_cond} fails, the expected effort is decreasing in $s$ for any $\bar{a}\geq 0$.
\end{proof}

\begin{proof}[Proof of Lemma~\ref{lem:effort_case_pos_small}]
Fix $\bar{a}<\frac{1}{\beta}-\Tilde{s}$. By \eqref{eq:effort_v}, the expected effort is given by
\begin{equation*}
 e^\star_v(s, \bar{a}, q)  = 
    \begin{cases}
            \frac{1}{\beta} -s + \bar{a}(1-2q) v(s) -(1-q)\bar{a}, & \text{if } \Tilde{s}\leq s < \frac{1}{\beta} -\bar{a} ,\\
            \big((1-2q)v(s) + q\big) (\frac{1}{\beta} -s), & \text{if } \frac{1}{\beta} -\bar{a}\leq s < \frac{1}{\beta} ,\\
            0, & \text{if } s\geq \frac{1}{\beta}.
        \end{cases}
\end{equation*}
The two expressions coincide at $s=\frac{1}{\beta}-\bar{a}$, so the effort is continuous on $(\Tilde{s},\frac{1}{\beta})$.
Below, we use $\frac{\partial  e^\star_v(s, \bar{a}, q)}{\partial s}\big|_{s\uparrow s_0}$ and $\frac{\partial  e^\star_v(s, \bar{a}, q)}{\partial s}\big|_{s\downarrow s_0}$ to denote the left-hand and right-hand side  partial derivative at $s=s_0$. Differentiating with respect to $s$, we have
\begin{align*}
    \dfrac{\partial  e^\star_v(s, \bar{a}, q)}{\partial s}& =
        \begin{cases}
            -1 + \bar{a}(1-2q) v'(s) , & \text{if } \Tilde{s}\leq s < \frac{1}{\beta} -\bar{a} ,\\
            (1-2q)v'(s)  (\frac{1}{\beta} -s) - \big((1-2q)v(s) + q\big), & \text{if } \frac{1}{\beta} -\bar{a}\leq s < \frac{1}{\beta} ,\\
            0, & \text{if } s\geq \frac{1}{\beta}.
        \end{cases}
\end{align*}
From the monotonicity and concavity of $v(s)$, the derivative $\frac{\partial  e^\star_v(s, \bar{a}, q)}{\partial s}$ is decreasing in $s$ within each sub-interval. 

Evaluating the derivative at the endpoints of each sub-interval, i.e.,  at $\Tilde{s}$, $\frac{1}{\beta} -\bar{a}$, and $\frac{1}{\beta}$, we have
\begin{align*}
    & \dfrac{\partial  e^\star_v(s, \bar{a}, q)}{\partial s}\Big|_{s\downarrow\Tilde{s}} = -1 + \bar{a}(1-2q) v'(\Tilde{s}), \quad \dfrac{\partial  e^\star_v(s, \bar{a}, q)}{\partial s}\Big|_{s\uparrow \frac{1}{\beta} -\bar{a}} = -1 + \bar{a}(1-2q) v'( \frac{1}{\beta} -\bar{a}),\\
    & \dfrac{\partial  e^\star_v(s, \bar{a}, q)}{\partial s}\Big|_{s\downarrow \frac{1}{\beta} -\bar{a}} = \bar{a}(1-2q) v'( \frac{1}{\beta} -\bar{a}) - \big((1-2q) v(\frac{1}{\beta} -\bar{a}) + q\big), \quad \dfrac{\partial  e^\star_v(s, \bar{a}, q)}{\partial s}\Big|_{s\uparrow\frac{1}{\beta} } =- \big((1-2q) v(\frac{1}{\beta} ) + q\big).
\end{align*}
Note that we have the following ordering: $\frac{\partial  e^\star_v(s, \bar{a}, q)}{\partial s}\big|_{s\downarrow\Tilde{s}} \geq \frac{\partial  e^\star_v(s, \bar{a}, q)}{\partial s}\big|_{s\uparrow \frac{1}{\beta} -\bar{a}}$, $\frac{\partial  e^\star_v(s, \bar{a}, q)}{\partial s}\big|_{s\downarrow \frac{1}{\beta} -\bar{a}}\geq \frac{\partial  e^\star_v(s, \bar{a}, q)}{\partial s}\big|_{s\uparrow\frac{1}{\beta} }$, $\frac{\partial  e^\star_v(s, \bar{a}, q)}{\partial s}\big|_{s\uparrow \frac{1}{\beta} -\bar{a}} < \frac{\partial  e^\star_v(s, \bar{a}, q)}{\partial s}\big|_{s\downarrow \frac{1}{\beta} -\bar{a}}$, and $\frac{\partial  e^\star_v(s, \bar{a}, q)}{\partial s}\big|_{s\uparrow\frac{1}{\beta} }<0$.
The first two inequalities hold because the derivative is decreasing within each sub-interval. The third holds because the derivative jumps upward at $s=\frac{1}{\beta}-\bar{a}$ by $1-\big((1-2q)v(\frac{1}{\beta}-\bar{a})+q\big)\geq q>0$. Therefore, the derivative is positive for some $s$ if and only if it is positive at the left endpoint of one of the two sub-intervals. Since the effort is continuous for $s>\Tilde{s}$, the expected effort $ e^\star_v(s, \bar{a}, q)$ is non-monotonic for $s>\Tilde{s}$, and thus non-monotonic in $s$ by \eqref{eq:reduction}, if and only if
\begin{align}
    & \dfrac{\partial  e^\star_v(s, \bar{a}, q)}{\partial s}\Big|_{s\downarrow\Tilde{s}}=-1 + \bar{a}(1-2q) v'(\Tilde{s}) > 0 \quad \text{ or } \nonumber \\
    & \dfrac{\partial  e^\star_v(s, \bar{a}, q)}{\partial s}\Big|_{s\downarrow \frac{1}{\beta} -\bar{a}}=\bar{a}(1-2q) v'( \frac{1}{\beta} -\bar{a}) - \big((1-2q) v(\frac{1}{\beta} -\bar{a}) + q\big) >0. \label{eq:effortb_derivative}
\end{align}
Note that both expressions are increasing in $\bar{a}$. Hence, \eqref{eq:effortb_derivative} holds for some $\bar{a}<\frac{1}{\beta}-\Tilde{s}$ if and only if it holds in the limit $\bar{a}\to\frac{1}{\beta}-\Tilde{s}$, i.e., 
$\lim_{\bar{a}\to \frac{1}{\beta}-\Tilde{s}} \frac{\partial  e^\star_v(s, \bar{a}, q)}{\partial s}\big|_{s\downarrow\Tilde{s}} > 0$ or  $\lim_{\bar{a}\to \frac{1}{\beta}-\Tilde{s}} \frac{\partial  e^\star_v(s, \bar{a}, q)}{\partial s}\big|_{s\downarrow \frac{1}{\beta} -\bar{a}}>0$.
Substituting the expressions in \eqref{eq:effortb_derivative}, it is equivalent to
\begin{align*}
    & -1 + (\frac{1}{\beta}-\Tilde{s})(1-2q) v'(\Tilde{s}) > 0 \quad \text{or} \quad (\frac{1}{\beta}-\Tilde{s})(1-2q) v'(\Tilde{s}) - \big((1-2q) v(\Tilde{s}) + q\big) >0,
\end{align*}
Since $(1-2q)v(\Tilde{s})+q\leq 1-q<1$, the condition is exactly $\frac{1}{\beta}-\Tilde{s}>  \frac{(1-2q) v(\Tilde{s}) + q}{(1-2q) v'(\Tilde{s})}$.
When this condition holds, the monotonicity of $\frac{\partial  e^\star_v(s, \bar{a}, q)}{\partial s}\big|_{s\downarrow\Tilde{s}}$ and $\frac{\partial  e^\star_v(s, \bar{a}, q)}{\partial s}\big|_{s\downarrow \frac{1}{\beta} -\bar{a}}$ in $\bar{a}$ implies that there exists $\Tilde{a}<\frac{1}{\beta}-\Tilde{s}$ such that \eqref{eq:effortb_derivative} holds for all $\Tilde{a}<\bar{a}<\frac{1}{\beta}-\Tilde{s}$ and does not hold for all $\bar{a}\leq \Tilde{a}$. Thus, the expected effort $ e^\star_v(s, \bar{a}, q)$ is non‑monotonic in skill for all $\Tilde{a}<\bar{a}<\frac{1}{\beta}-\Tilde{s}$ and is decreasing for all $\bar{a}\leq \Tilde{a}$. When this condition fails, \eqref{eq:effortb_derivative} does not hold for any $\bar{a}<\frac{1}{\beta}-\Tilde{s}$, and the expected effort is decreasing.
\end{proof}
\begin{proof}[Proof of Lemma~\ref{lem:effort_case_pos_large}]
Fix $\bar{a}\geq\frac{1}{\beta}-\Tilde{s}$. Then $(\frac{1}{\beta}-s-\bar{a})^+=0$ for all $s\geq\Tilde{s}$, so by \eqref{eq:effort_v}, the expected effort is given by 
\begin{equation*}
 e^\star_v(s, \bar{a}, q)  = 
    \begin{cases}
            \big((1-2q)v(s) + q\big) (\frac{1}{\beta} -s), & \text{if } \Tilde{s}\leq s < \frac{1}{\beta} ,\\
            0, & \text{if } s\geq \frac{1}{\beta}.
        \end{cases}
\end{equation*}
Differentiating with respect to $s$, we have
\begin{align*}
    \dfrac{\partial  e^\star_v(s, \bar{a}, q)}{\partial s}& =
        \begin{cases}
            (1-2q)v'(s)  (\frac{1}{\beta} -s) - \big((1-2q)v(s) + q\big), & \text{if } \Tilde{s}\leq s < \frac{1}{\beta} ,\\
            0, & \text{if } s\geq \frac{1}{\beta}.
        \end{cases}
\end{align*}
Again, the derivative $\frac{\partial  e^\star_v(s, \bar{a}, q)}{\partial s}$ is decreasing in $s$ within each sub-interval and it does not depend on $\bar{a}$. Hence, by \eqref{eq:reduction}, for every $\bar{a}\geq \frac{1}{\beta}-\Tilde{s}$, the effort
 $e^\star_v(s, \bar{a}, q)$ is non‑monotonic in skill if and only if $\frac{\partial  e^\star_v(s, \bar{a}, q)}{\partial s}\big|_{s\downarrow\Tilde{s}} > 0$.
Substituting the effort expression, it is equivalent to
\begin{align*}
    (1-2q) v'(\Tilde{s}) (\frac{1}{\beta} -\Tilde{s}) - \big((1-2q)v(\Tilde{s}) + q\big)>0,
\end{align*}
which is exactly $\frac{1}{\beta}-\Tilde{s}>  \frac{(1-2q) v(\Tilde{s}) + q}{(1-2q) v'(\Tilde{s})}$, the same condition as above.
When this condition holds, for all $\bar{a}\geq\frac{1}{\beta}-\Tilde{s}$, the expected effort $ e^\star_v(s, \bar{a}, q)$ is non‑monotonic in skill. When this condition fails, for all $\bar{a}\geq\frac{1}{\beta}-\Tilde{s}$, the expected effort $ e^\star_v(s, \bar{a}, q)$ is decreasing.
\end{proof}

\subsection{The case of unimodal skill distribution (Lemma~\ref{lem:bayes_effort_decrease})}\label{app:pf_skill_unimodal_case}
\begin{proof}[Proof of Lemma~\ref{lem:bayes_effort_decrease}]
Our goal is to show that the skill distribution is unimodal when effort $e^\star_v(s,\bar{a},q)$ is decreasing in $s$. For any $N\geq 2$ states $S=(s_1,s_2,\cdots,s_N)$ and decay rate $\mu>0$, the transition ratio $r(s_k,\mu)=\frac{\lambda( e^\star_v(s_k,\bar{a},q))}{\mu}$ is decreasing in $k$, and thus it can cross $1$ at most once. 
Recall from Equation~\eqref{eq:steady_state_prob} that $\frac{\pi_{k+1}}{\pi_{k}}=\frac{\lambda( e^\star_v(s_k,\bar{a},q))}{\mu}$.
When the transition ratio remains greater than $1$, $\frac{\pi_{k+1}}{\pi_{k}}$ is greater than $1$, and the probability $\pi_k(\mu,\bar{a},q,S)$ increases; otherwise, when it drops below $1$, $\frac{\pi_{k+1}}{\pi_{k}}$ is smaller than $1$, and the probability $\pi_k(\mu,\bar{a},q,S)$ decreases. Thus, the resulting steady-state distribution is unimodal.
\end{proof}

\section{Extension to convex cost functions}\label{app:cvx_cost}
In this section, we examine the case of a convex cost function. For simplicity, we focus on the perfect substitution skill-AI interaction function $h(s,a)=s+a$. The structure of this section is as follows: we first investigate the basic model under convex cost functions. We then show how the productivity paradox carries over under the mechanisms of skill development and AI unreliability. Finally, we validate our findings with numerical examples.

Suppose that the production function $p(x)\geq 0$ is strictly increasing, strictly concave, continuous, and twice differentiable, and that the cost function $c(e)\geq 0$ is strictly increasing, strictly convex, continuous, and twice differentiable. 
We begin by extending our basic model. The utility-maximizing effort and corresponding productivity are given by
\begin{align*}
    e^\star(s,a)=\arg\max_{e\geq 0} \left(p(s+e+a) -c(e)\right) \quad \text{and} \quad p^\star(s,a)=p(s+e^\star(s,a)+a).
\end{align*}
In this section, we assume $p'(0)\geq c'(0)$; otherwise, $p'(s+e+a)\leq p'(0)<c'(0)\leq c'(e)$ for all $e\geq 0$, so the agent never exerts effort and productivity becomes exogenous. Define a critical level $x_{c}^\star=x_{c}^\star(p,c)$ as the largest input level for which the marginal gain from production is at least as large as the marginal cost at zero effort:
$$x_{c}^\star = \sup\{x:p'(x) \geq c'(0)\}.$$
Note that when $c$ is a linear function, the defined critical value $x_{c}^\star$ coincides with the definition in Section~\ref{sec:basic}, i.e., $c(e)=\gamma e$ implies $x_{c}^\star(p,c)=x^\star(p,\gamma)$. We further assume that $\lim_{x\to\infty} p'(x)<\lim_{x\to\infty} c'(x)$ to ensure boundedness of utility-maximizing effort as we will see in \eqref{eq:foc_cvx_cost}.

We have the following monotonicity result, which extends Proposition~\ref{prop:effort_productivity_expressions}.
\begin{proposition}\label{prop:effort_productivity_cvx_cost}
The effort level $e^\star(s,a)$ (weakly) decreases in $s+a$. The productivity $p^\star(s,a)$ strictly increases in $s+a$.
\end{proposition}
\begin{proof}[Proof of Proposition~\ref{prop:effort_productivity_cvx_cost}]
For a given AI assistance level $a$ and a given skill level $s$, the following case distinction characterizes the agent's effort and corresponding productivity:
\begin{enumerate}
    \item $s+a \leq x_{c}^\star$. The concavity of $p$ and convexity of $c$ imply that the utility is increasing-then-decreasing in $e$. The agent's utility-maximizing effort is a non-negative value characterized by the first-order condition: 
    \begin{equation}\label{eq:foc_cvx_cost}
        p'(s+e^\star(s,a)+a) = c'(e^\star(s,a)).
    \end{equation}
    Here, the utility-maximizing effort $e^\star(s,a)$ is finite as the condition $\lim_{x\to\infty} p'(x)<\lim_{x\to\infty} c'(x)$ holds.
    \item $s+a > x_{c}^\star$. The concavity of $p$ and convexity of $c$ imply that the agent's utility is decreasing in $e$ for all non-negative effort levels. Thus, the agent's utility-maximizing effort is $e^\star(s,a)=0$. The corresponding productivity is given by $p^\star(s,a)=p(s+a)$.
\end{enumerate}
Consider $x_1=s_1+a_1$ and $x_2=s_2+a_2$ where $x_1<x_2$. We have the following case distinction:
\begin{enumerate}
\item $x_1<x_2\leq x_{c}^\star$. The first-order condition \eqref{eq:foc_cvx_cost} implies 
\begin{equation}\label{eq:foc_cvx_cost_compare}
    1=\dfrac{p'(x_1+e^\star(s_1,a_1))}{c'(e^\star(s_1,a_1))}=\dfrac{p'(x_2+e^\star(s_2,a_2))}{c'(e^\star(s_2,a_2))}< \dfrac{p'(x_1+e^\star(s_2,a_2))}{c'(e^\star(s_2,a_2))},
\end{equation}
where the last inequality holds due to the concavity of $p$. Again, the concavity of $p$ and convexity of $c$ imply that $g(e):=\frac{p'(x_1+e)}{c'(e)}$ is a decreasing function in $e$. Since $g(e^\star(s_1,a_1))=1$ and $g(e^\star(s_2,a_2))>1$ as implied by \eqref{eq:foc_cvx_cost_compare}, we have $e^\star(s_1,a_1) > e^\star(s_2,a_2)$. Thus, the convexity of $c$ implies that $c'(e^\star(s_1,a_1))>c'(e^\star(s_2,a_2))$. Then the equality in \eqref{eq:foc_cvx_cost_compare} implies that $p'(x_1+e^\star(s_1,a_1))>p'(x_2+e^\star(s_2,a_2))$. Finally, we have $p(x_1+e^\star(s_1,a_1))<p(x_2+e^\star(s_2,a_2))$ since $p$ is concave, i.e., $p^\star(s_1,a_1) < p^\star(s_2,a_2)$.
\item $x_{c}^\star <x_1<x_2$ (when $x_{c}^\star\neq \infty$). Then the effort and productivity are given by $e^\star(s_1,a_1)=e^\star(s_2,a_2)=0$ and $p^\star(s_1,a_1) = p(x_1) <p(x_2)=p^\star(s_2,a_2)$.
\item $x_1\leq x_{c}^\star <x_2$ (when $x_{c}^\star\neq \infty$). The two cases above imply that the effort and productivity satisfy $e^\star(s_1,a_1)\geq 0=e^\star(s_2,a_2)$ and $p^\star(s_1,a_1) = p(x_1+e^\star(s_1,a_1))\leq p(x_{c}^\star) <p(x_2)=p^\star(s_2,a_2)$.
\end{enumerate}
Note that the last two cases exist if and only if $x_{c}^\star \neq \infty$. Combining all cases, we complete the proof.
\end{proof}

\subsection{Model of skill development (Section~\ref{sec:bdc})}
We then endogenize the skill development. We consider a birth-death chain with two skill states $s_1<s_2$. At each state $s_k$, the agent's effort level is given by 
$$e^\star(s_k,a)=\arg\max_{e\geq 0} \left(p(s_k+e+a) - c(e)\right).$$ 
The transition rate of each state is given by $\mathcal{K}_{s_1 s_2}(a)= \lambda(e^\star(s_1,a))$ and $\mathcal{K}_{s_2 s_1}(a)= \mu$, where $\mu > 0$ and $\lambda(e)>0$ is an increasing and differentiable transition function.
Recall that 
$$x_{c}^\star = \sup\{x:p'(x) \geq c'(0)\}.$$ 
We identify a condition for the curvature of production and cost functions, under which the productivity paradox persists and our notion of semi-elasticity gap is generalized.
\begin{condition}\label{cond:cvx_two_state}
Suppose that $z_0 + z_2 \mu^2 < z_1 \mu$, where
\begin{align*}
    & z_2 = \frac{p'(x_{c}^\star)}{\lambda(0)\big(p(x_{c}^\star +s_2-s_1)-p(x_{c}^\star)\big)} \frac{c''(0)}{c''(0)- p''(x_{c}^\star)}, \\
    & z_1 = \frac{\lambda'(0)}{\lambda(0)}\frac{-p''(x_{c}^\star)}{c''(0)- p''(x_{c}^\star)} -\frac{p'(x_{c}^\star +s_2-s_1)}{p(x_{c}^\star +s_2-s_1)-p(x_{c}^\star)} -\frac{p'(x_{c}^\star)}{p(x_{c}^\star +s_2-s_1)-p(x_{c}^\star)} \frac{c''(0)}{c''(0)- p''(x_{c}^\star)}, \\
    & z_0 = \lambda(0) \frac{p'(x_{c}^\star +s_2-s_1)}{p(x_{c}^\star +s_2-s_1)-p(x_{c}^\star)}.
\end{align*}  
\end{condition}
We note that $z_2\geq 0$, $z_0\geq 0$, and the condition can be digested as follows.
\begin{itemize}
\item For a large class of cost functions, Condition~\ref{cond:cvx_two_state} reduces to a positive semi-elasticity gap and a large decay rate: consider any convex cost function with $c''(0)=0$, thus $z_2=0$ and $z_1$ recovers the semi-elasticity gap defined in Definition~\ref{def:semi-elasticity_general}. Condition~\ref{cond:cvx_two_state} reduces to 
$$z_1=\frac{\lambda'(0)}{\lambda(0)} -\frac{p'(x_{c}^\star+s_2-s_1)}{p(x_{c}^\star+s_2-s_1)-p(x_{c}^\star)}>0 \quad \text{and} \quad \mu> \frac{z_0}{z_1},$$
as in Proposition~\ref{prop:bdc_prod_n=2}; see Proposition~\ref{prop:positive_semi-elasticity_gap} for details on when the positive semi-elasticity gap can hold.
\item When $z_2\neq 0$, Condition~\ref{cond:cvx_two_state} is equivalent to $z_1>2\sqrt{z_2z_0}$ and $\mu\in \big(\frac{z_1-\sqrt{z_1^2-4z_2z_0}}{2z_2}, \frac{z_1+\sqrt{z_1^2-4z_2z_0}}{2z_2}\big)$.
\end{itemize}

Under Condition~\ref{cond:cvx_two_state}, the productivity paradox persists, as stated next.
\begin{proposition}\label{prop:bdc_n=2_cvx_cost}
     Consider any concave production function $p(x)$, convex cost function $c(e)$, increasing transition function $\lambda(e)$, and two states $s_1<s_2$. Suppose $x_{c}^\star<\infty$ and $s_1<x_{c}^\star$. If Condition~\ref{cond:cvx_two_state} holds, then there exists $\tau<x_{c}^\star-s_1$ such that the steady-state productivity $\mathcal{P}(a)$ is decreasing on $(\tau, x_{c}^\star-s_1]$. 
\end{proposition}
\begin{proof}[Proof of Proposition~\ref{prop:bdc_n=2_cvx_cost}]
The steady-state productivity can be written as
$$\mathcal{P}(a)= \dfrac{\mu}{\mu+\lambda(e^\star(s_1,a))}\cdot p(s_1+e^\star(s_1,a)+a) + \dfrac{\lambda(e^\star(s_1,a))}{\mu+\lambda(e^\star(s_1,a))}\cdot p(s_2+e^\star(s_2,a)+a).$$
We focus on the regime where $s_1+a \leq x_{c}^\star<s_2+a$, i.e., $a\in (x_{c}^\star - s_2, x_{c}^\star-s_1]$. Then the analysis of Proposition~\ref{prop:effort_productivity_cvx_cost} implies that $e^\star(s_2,a)=0$.
Differentiating both sides of the first-order condition \eqref{eq:foc_cvx_cost} with respect to $a$ and substituting $s=s_1$, we obtain that
$$ p''(s_1+e^\star(s_1,a) +a) \cdot \left( \frac{\partial e^\star(s_1,a) }{\partial a} + 1 \right) = c''(e^\star(s_1,a))\frac{\partial e^\star(s_1,a) }{\partial a}.$$
Thus, the derivative of $e^\star(s_1,a)$ versus $a$ is given by
\begin{equation*}
    \dfrac{\partial e^\star(s_1,a) }{\partial a} = \dfrac{ p''(s_1+e^\star(s_1,a) +a)}{c''(e^\star(s_1,a) )- p''(s_1+e^\star(s_1,a)+a)}.
\end{equation*}
Recall that $e^\star(s_2,a)=0$, we now take the derivative of $\mathcal{P}$, which is given by
\begin{align*}  \mathcal{P}'(a)&=\dfrac{\mu}{\mu+\lambda(e^\star(s_1,a))}\cdot  \frac{\partial}{\partial a}p(s_1+e^\star(s_1,a)+a) + \dfrac{\lambda(e^\star(s_1,a))}{\mu+\lambda(e^\star(s_1,a))}\cdot p'(s_2+a)\\
    & + \dfrac{\mu \lambda'(e^\star(s_1,a))}{(\mu+\lambda(e^\star(s_1,a)))^2}\frac{\partial e^\star(s_1,a)}{\partial a}\Big(p(s_2+a) - p(s_1+e^\star(s_1,a)+a)\Big)\\
    & =\dfrac{\mu}{\mu+\lambda(e^\star(s_1,a))}\cdot  p'(s_1+e^\star(s_1,a)+a) (\frac{\partial e^\star(s_1,a)}{\partial a} + 1) + \dfrac{\lambda(e^\star(s_1,a))}{\mu+\lambda(e^\star(s_1,a))}\cdot p'(s_2+a)\\
    & + \dfrac{\mu \lambda'(e^\star(s_1,a))}{(\mu+\lambda(e^\star(s_1,a)))^2} \frac{\partial e^\star(s_1,a)}{\partial a}\Big(p(s_2+a) - p(s_1+e^\star(s_1,a)+a)\Big).
\end{align*}
Substituting the expression of $\frac{\partial e^\star(s_1,a) }{\partial a}$, we have
\begin{align*}
    \mathcal{P}'(a)    &=\dfrac{\mu}{\mu+\lambda(e^\star(s_1,a))}\cdot  p'(s_1+e^\star(s_1,a)+a) \dfrac{c''(e^\star(s_1,a))}{c''(e^\star(s_1,a))- p''(s_1+e^\star(s_1,a)+a)}\\
    &+ \dfrac{\lambda(e^\star(s_1,a))}{\mu+\lambda(e^\star(s_1,a))}\cdot p'(s_2+a)\\
    & + \dfrac{\mu \lambda'(e^\star(s_1,a))}{(\mu+\lambda(e^\star(s_1,a)))^2} \dfrac{p''(s_1+e^\star(s_1,a) +a)}{c''(e^\star(s_1,a))- p''(s_1+e^\star(s_1,a)+a)} \Big(p(s_2+a) - p(s_1+e^\star(s_1,a)+a)\Big).
\end{align*}
Evaluating the derivative at the right endpoint of the regime, i.e., $a=x_{c}^\star-s_1$, as $e^\star(s_1,x_{c}^\star-s_1)=0$, we have
\begin{align*}
     \mathcal{P}'_{-}(x_{c}^\star -s_1)&=\left(\dfrac{\mu}{\mu+\lambda(0)}\cdot  p'(x_{c}^\star) \dfrac{c''(0)}{c''(0)- p''(x_{c}^\star)} + \dfrac{\lambda(0)}{\mu+\lambda(0)}\cdot p'(x_{c}^\star +s_2-s_1)\right)\\
    & + \dfrac{\mu \lambda'(0)}{(\mu+\lambda(0))^2} \dfrac{p''(x_{c}^\star)}{c''(0)- p''(x_{c}^\star)} \big(p(x_{c}^\star +s_2-s_1) - p(x_{c}^\star)\big).
\end{align*}
By rearrangement, this term is negative if and only if
\begin{align*}
     &\mu^2 \left(p'(x_{c}^\star) \dfrac{c''(0)}{c''(0)- p''(x_{c}^\star)}\right)  + \lambda^2(0) p'(x_{c}^\star +s_2-s_1) \\
     < &\mu \left(-\lambda'(0)\dfrac{p''(x_{c}^\star)}{c''(0)- p''(x_{c}^\star)}\big(p(x_{c}^\star +s_2-s_1)-p(x_{c}^\star)\big) -\lambda(0)p'(x_{c}^\star +s_2-s_1) -\lambda(0) p'(x_{c}^\star) \dfrac{c''(0)}{c''(0)- p''(x_{c}^\star)}\right).
\end{align*}
Dividing both sides by $\lambda(0) \big(p(x_{c}^\star +s_2-s_1)-p(x_{c}^\star)\big)$, this corresponds to 
$$z_0 + z_2 \mu^2 < z_1 \mu.$$ 
Thus, under Condition~\ref{cond:cvx_two_state}, the sign of $\mathcal{P}'_{-}(x_{c}^\star -s_1)$ and the continuity of the derivative imply that there exists $\tau<x_{c}^\star-s_1$, such that $\mathcal{P}'(a)$ is negative on $(\tau, x_{c}^\star-s_1)$. It follows that the steady-state productivity $\mathcal{P}(a)$ is decreasing on $(\tau, x_{c}^\star-s_1]$. 
\end{proof}

\subsection{Model of AI unreliability (Section~\ref{sec:AI_unreliability})}
We then capture uncertainty in AI assistance. The utility-maximizing effort and corresponding productivity are given by
\begin{align*}
    & e^\star(s,\bar{a},q):=\arg\max_{e\geq 0} \left(q\cdot p(s+e+\bar{a}) + (1-q)\cdot p(s+e)-c(e)\right),\\
    & p^\star(s,\bar{a},q):= q\cdot p(s+e^\star(s,\bar{a},q)+\bar{a}) + (1-q)\cdot p(s+e^\star(s,\bar{a},q)).
\end{align*}
Recall that the absolute risk aversion is defined as $A(x)=-\frac{p''(x)}{p'(x)}$ (Definition~\ref{def:risk_averse}). We identify a condition for the absolute risk aversion and cost curvature, under which the productivity paradox persists and our IARA condition is generalized.
\begin{condition}\label{cond:cvx_unreliable}
Suppose that for any $s\geq 0$, $0<q<1$ and sufficiently large $\bar{a}$, the absolute risk aversion satisfies
$$A(s+e^\star(s,\bar{a},q) +\bar{a}) - A(s+e^\star(s,\bar{a},q))>\frac{c''(e^\star(s,\bar{a},q))}{(1-q) p'(s+e^\star(s,\bar{a},q)) }.$$
\end{condition}
Here, Condition~\ref{cond:cvx_unreliable} implies our IARA condition. It is satisfied by the following function class:
\begin{example}\label{ex:cvx_unreliable}
Define $C^3(\mathbb{R}_{\geq 0})$ as the set of functions that are three times continuously differentiable. Condition~\ref{cond:cvx_unreliable} holds for the bounded function class $\Omega$, where
$$\Omega=\{p, c\in C^3(\mathbb{R}_{\geq 0}):  \exists B_1>0, B_2>0 \text{ s.t. } A'(x)\geq B_1, \forall x\geq 0 \text{ and }\limsup_{e\to \infty}\frac{c''(e)}{c'(e)}\leq B_2\}.$$
In particular, the Gaussian $(\int_0^x \exp(-t^2)dt)$ and expo-power $(p(x)=1-\exp(- x^2))$ production function, and any cost function whose growth in the tail is at most exponential, are included in $\Omega$.
\end{example}
Under Condition~\ref{cond:cvx_unreliable}, the productivity paradox can persist, as stated next.
\begin{proposition}\label{prop:unreliable_cvx_cost}
    Consider any concave production function $p(x)$ and any convex cost function $c(e)$. Suppose $(1-q)p'(s)\geq c'(0)$ and $0<q<1$.
    \begin{enumerate}
        \item If $p$ satisfies Condition~\ref{cond:cvx_unreliable}, productivity $p^\star(s, \bar{a}, q)$ is decreasing in $\bar{a}$ when $\bar{a}$ is sufficiently large.
        \item If $p$ satisfies DARA, productivity $p^\star(s, \bar{a}, q)$ is  increasing in $\bar{a}$. 
    \end{enumerate}
\end{proposition}
\begin{proof}[Proof of Proposition~\ref{prop:unreliable_cvx_cost}]
Recall that the utility-maximizing effort is given by 
$$e^\star(s,\bar{a},q):=\arg\max_{e\geq 0} \left(q\cdot p(s+e+\bar{a}) + (1-q)\cdot p(s+e)-c(e)\right).$$
As $(1-q)p'(s)\geq c'(0)$, the marginal utility of effort at $e=0$ is strictly positive. Hence, the utility-maximizing effort is strictly positive and characterized by the first-order condition:
\begin{equation}\label{eq:unreliable_foc_cvx_cost}
    q \cdot p'(s+e^\star (s,\bar{a},q)+\bar{a}) + (1-q) \cdot p'(s+e^\star (s,\bar{a},q)) = c'(e^\star (s,\bar{a},q)).
\end{equation}
For simplicity, write $e^\star=e^\star (s,\bar{a},q)$. Differentiating both sides of \eqref{eq:unreliable_foc_cvx_cost} with respect to $\bar{a}$, we obtain that
$$ q \cdot p''(s+e^\star +\bar{a}) \cdot \left( \frac{\partial e^\star}{\partial \bar{a}} + 1 \right)+ (1-q) \cdot p''(s+e^\star ) \cdot \frac{\partial e^\star}{\partial \bar{a}} = c''(e^\star)\frac{\partial e^\star}{\partial \bar{a}}.$$
Thus, the derivative of $e^\star$ versus $\bar{a}$ is given by
\begin{equation*}
    \dfrac{\partial e^\star}{\partial \bar{a}} = \dfrac{q \cdot p''(s+e^\star +\bar{a})}{c''(e^\star)- q \cdot p''(s+e^\star +\bar{a}) - (1-q) \cdot p''(s+e^\star)}.
\end{equation*}
We now differentiate $p^\star$ with respect to $\bar{a}$:
\begin{align*}
    \frac{\partial p^\star}{\partial \bar{a}}(s,\bar{a},q) & = q \cdot p'(s+e^\star + \bar{a}) \cdot \left( \frac{\partial e^\star}{\partial \bar{a}} + 1 \right)+ (1-q) \cdot p'(s+e^\star ) \cdot \frac{\partial e^\star}{\partial \bar{a}}\\
    & =\big(q \cdot p'(s+e^\star + \bar{a}) + (1-q) \cdot p'(s+e^\star)\big) \cdot \frac{\partial e^\star}{\partial \bar{a}} + q \cdot p'(s+e^\star + \bar{a}).
\end{align*}
Substituting the expression of $\dfrac{\partial e^\star}{\partial \bar{a}}$, we obtain
\begin{align*}
    \frac{\partial p^\star}{\partial \bar{a}}(s,\bar{a},q)& = \dfrac{ \big(q \cdot p'(s+e^\star+\bar{a}) + (1-q) \cdot p'(s+e^\star) \big) \cdot q \cdot p''(s+e^\star +\bar{a})}{c''(e^\star)- q \cdot p''(s+e^\star +\bar{a}) - (1-q) \cdot p''(s+e^\star)} + q \cdot p'(s+e^\star + \bar{a})  \\
    & = \dfrac{q c''(e^\star) p'(s+e^\star +\bar{a}) - q(1-q)  p'(s+e^\star) p'(s+e^\star +\bar{a})\big(A(s+e^\star +\bar{a}) - A(s+e^\star)\big)}{c''(e^\star)- q \cdot p''(s+e^\star +\bar{a}) - (1-q) \cdot p''(s+e^\star)}  \\
    & = \dfrac{q  p'(s+e^\star +\bar{a})  \Big(c''(e^\star) - (1-q) p'(s+e^\star) \big(A(s+e^\star +\bar{a}) - A(s+e^\star)\big)\Big)}{c''(e^\star)- q \cdot p''(s+e^\star +\bar{a}) - (1-q) \cdot p''(s+e^\star)}. 
\end{align*}
Define $g(s,\bar{a},q)= (1-q) p'(s+e^\star) \big(A(s+e^\star +\bar{a}) - A(s+e^\star)\big)- c''(e^\star)$. Then 
$$ \frac{\partial p^\star}{\partial \bar{a}}(s,\bar{a},q) =  - \dfrac{q  p'(s+e^\star +\bar{a}) g(s,\bar{a},q)}{c''(e^\star)- q \cdot p''(s+e^\star +\bar{a}) - (1-q) \cdot p''(s+e^\star)}. $$
The denominator is positive due to the concavity of $p$ and convexity of $c$. We then focus on the numerator:
\begin{enumerate}
\item If $p$ satisfies Condition~\ref{cond:cvx_unreliable} (and thus IARA). By Condition~\ref{cond:cvx_unreliable}, we have $g(s,\bar{a},q)>0$. Thus the derivative $\frac{\partial p^\star}{\partial \bar{a}}(s, \bar{a}, q)$ is negative for sufficiently large $\bar{a}$ and the productivity $p^\star(s, \bar{a}, q)$ is decreasing in $\bar{a}$.
\item If $p$ satisfies DARA. Then together with the convexity of $c$, we have $g(s,\bar{a},q) < 0$ and the derivative $\frac{\partial p^\star}{\partial \bar{a}}(s, \bar{a}, q)$ is always positive and the productivity $p^\star(s, \bar{a}, q)$ is increasing in $\bar{a}$.
\end{enumerate}
\end{proof}

\begin{proof}[Proof of Example~\ref{ex:cvx_unreliable}]
Our goal is to prove that Condition~\ref{cond:cvx_unreliable} holds for any production and cost functions that satisfy $A'(x)\geq B_1, \forall x\geq 0 $ and $\limsup_{e\to \infty}\frac{c''(e)}{c'(e)}\leq B_2$, for some $B_1>0, B_2>0$ (recall that $A(x)=-\frac{p''(x)}{p'(x)}$).

By the continuity of $c',c''$, and the condition $\limsup_{e\to \infty}\frac{c''(e)}{c'(e)}\leq B_2$, there exists $\varepsilon>0$ and $M>0$, such that $\frac{c''(e)}{c'(e)}\leq M$ for any $e\geq \varepsilon$. Given $\varepsilon$, there exists $M'>0$ such that $\sup_{e\leq\varepsilon}\frac{c''(e)}{p'(s+e)}<M'$.
We then investigate the sign of $g(s,\bar{a},q):= (1-q) p'(s+e^\star) \big(A(s+e^\star +\bar{a}) - A(s+e^\star)\big)- c''(e^\star)$. 
\begin{enumerate}
\item $e^\star(s,\bar{a},q)\geq \varepsilon$. The first-order condition \eqref{eq:unreliable_foc_cvx_cost} and the concavity of $p$ imply that 
$$p'(s+e^\star)>q \cdot p'(s+e^\star+\bar{a}) + (1-q) \cdot p'(s+e^\star) =c'(e^\star).$$
Since $A(s+e^\star +\bar{a}) - A(s+e^\star)>0$ for any $\bar{a}>0$, we have the following bound:
\begin{align*}
    g(s, \bar{a}, q) &>  (1-q) c'(e^\star) \big(A(s+e^\star +\bar{a}) - A(s+e^\star)\big) -c''(e^\star) \\
    & = c'(e^\star) \Big( (1-q) \big(A(s+e^\star +\bar{a}) - A(s+e^\star)\big) - \frac{c''(e^\star)}{c'(e^\star)} \Big)\\
    & \geq c'(e^\star) \big( (1-q) \bar{a} B_1 - M\big),
\end{align*}
where the last inequality holds by the boundedness condition, given that the term in the bracket is nonnegative.
\item $e^\star(s,\bar{a},q)<\varepsilon$. Then we have the following bound:
\begin{align*}
    g(s, \bar{a}, q) &=   (1-q) p'(s+e^\star) \big(A(s+e^\star +\bar{a}) - A(s+e^\star)\big)- c''(e^\star)\\
    & = p'(s+e^\star) \Big( (1-q) \big(A(s+e^\star +\bar{a}) - A(s+e^\star)\big) - \frac{c''(e^\star)}{p'(s+e^\star)} \Big)\\
    & \geq c'(e^\star) \Big( (1-q) \big(A(s+e^\star +\bar{a}) - A(s+e^\star)\big) - \frac{c''(e^\star)}{p'(s+e^\star)} \Big)\\
    & \geq c'(e^\star) \big( (1-q) \bar{a} B_1- M' \big),
\end{align*}
where the last inequality holds by the boundedness condition, given that the term in the bracket is nonnegative.
\end{enumerate}
Combining the two cases, for sufficiently large $\bar{a}$ such that $\bar{a}>\frac{\max(M,M')}{(1-q)B_1}$, we have $ g(s, \bar{a}, q)>0$, which is exactly Condition~\ref{cond:cvx_unreliable} after rearrangement.
\end{proof}

\subsection{Numerics}
We conduct numerical experiments across multiple parametric specifications of production and cost functions to study the productivity paradoxes. For robustness checks, we include both linear and convex cost functions. 

\paragraph{Model of skill development (Section~\ref{sec:bdc}).} We numerically illustrate the model of skill development using a skill set with $N=4$ states, where $s_k=0.1 (k-1), k\in[N]$. We consider an upward transition function $\lambda(e)=0.01+e$ and a decay rate $\mu=0.2$. 
We study several concave production functions, including fractional and power law functions, and several quadratic convex cost functions varying the value of second-order coefficient: $$p(x)=\frac{x}{1+x},\, p(x)=\frac{1}{2}x^{1/2},\, p(x)=\frac{1}{2}x^{1/3},\, \text{and } \, p(x)=\frac{1}{2}x^{1/4},$$ 
$$c(e)=\frac{1}{2} e + c_2 e^2,\quad c_2 \in \left\{0, \frac{1}{16}, \frac{1}{8}\right\}.$$

Figure~\ref{fig:cvx_s3} shows the steady-state productivity $\mathcal{P}(a)$ as a function of AI assistance level $a$. Productivity exhibits a multimodal pattern for linear costs and the productivity paradox persists under convex costs, although the productivity decline is slightly attenuated.
\begin{figure}[ht]
    \centering
    \begin{minipage}{0.45\textwidth}
        \centering
        \includegraphics[width=\linewidth]{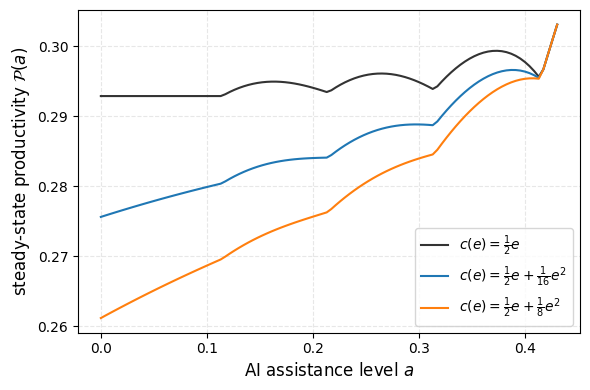}
        \subcaption[]{$p(x)=\frac{x}{1+x}$}
        
    \end{minipage}
    \hfill
    \begin{minipage}{0.45\textwidth}
        \centering
        \includegraphics[width=\linewidth]{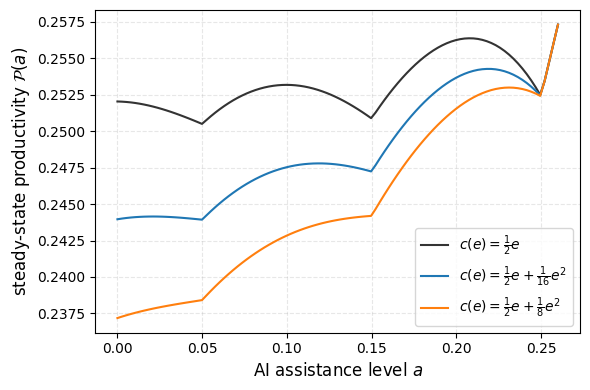}
        \subcaption[]{$p(x)=\frac{1}{2}x^{1/2}$}
        
    \end{minipage}
    \hfill
    
    \begin{minipage}{0.45\textwidth}
        \centering
        \includegraphics[width=\linewidth]{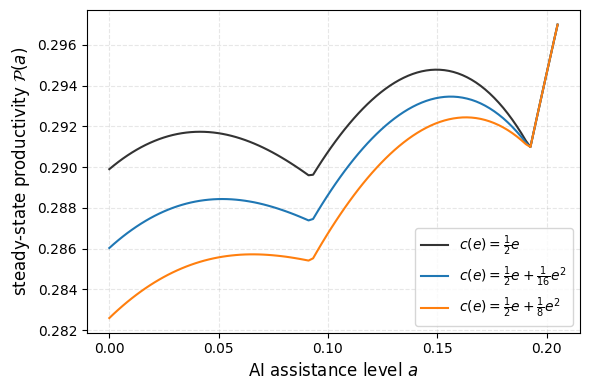}
        \subcaption[]{$p(x)=\frac{1}{2}x^{1/3}$}
        
    \end{minipage}
    \hfill
    \begin{minipage}{0.45\textwidth}
        \centering
        \includegraphics[width=\linewidth]{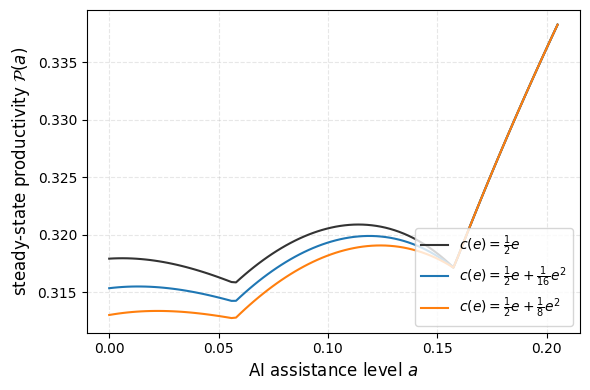}
        \subcaption[]{$p(x)=\frac{1}{2}x^{1/4}$}
        
    \end{minipage}
    \caption{Productivity paradox under endogenous skill development}
    \label{fig:cvx_s3}
\end{figure}

\paragraph{Model of AI unreliability (Section~\ref{sec:AI_unreliability}).} We numerically illustrate the model of AI unreliability by varying the reliability measure $q$ for a fixed skill $s=0$. We study several concave production functions that satisfy IARA, including Gaussian,  truncated quadratic, and expo-power functions (see Section~\ref{sec:unreliability_discussion}):
$$p(x)=\int_0^x \exp(-t^2)dt,\ p(x)=(-x^2+2 x)\mathbb{I}(x\leq 1) + \mathbb{I}(x> 1),\ \text{and} \ p(x)=1-\exp(- x^2).$$
We consider the same quadratic convex cost functions as above.

Figure~\ref{fig:cvx_s4} shows the productivity $p^\star(s,\bar{a},q)$ as a function of AI proficiency level $\bar{a}$. Productivity is either decreasing or V-shaped for linear costs, given that $q<1$. The productivity paradox persists under convex costs, although the productivity decline slightly attenuates.
\begin{figure}[ht]
    \centering
    \begin{minipage}{0.75\textwidth}
        \centering
        \includegraphics[width=\linewidth]{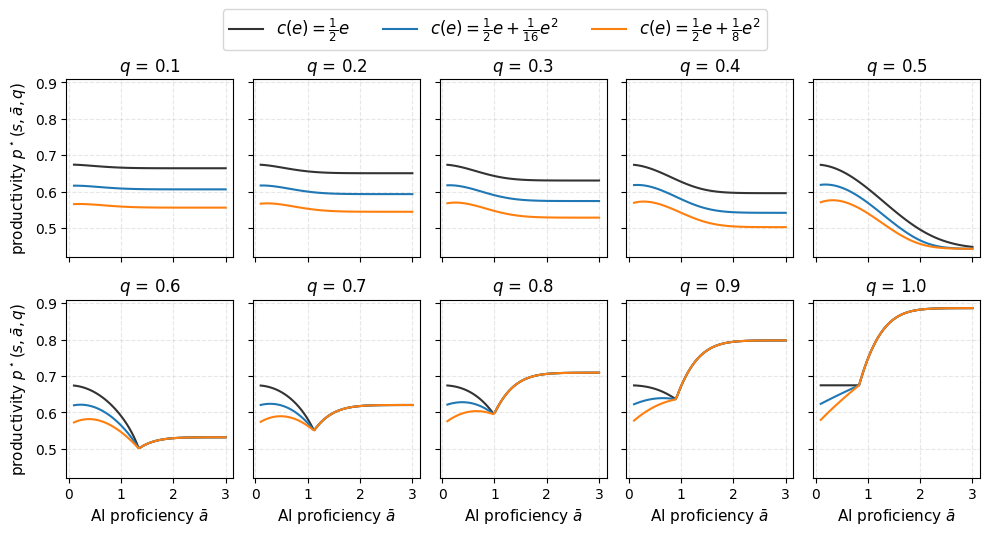}
        \subcaption[]{$p(x)=\int_0^x \exp(-t^2)dt$}
        
    \end{minipage}
    \hfill
    \begin{minipage}{0.75\textwidth}
        \centering
        \includegraphics[width=\linewidth]{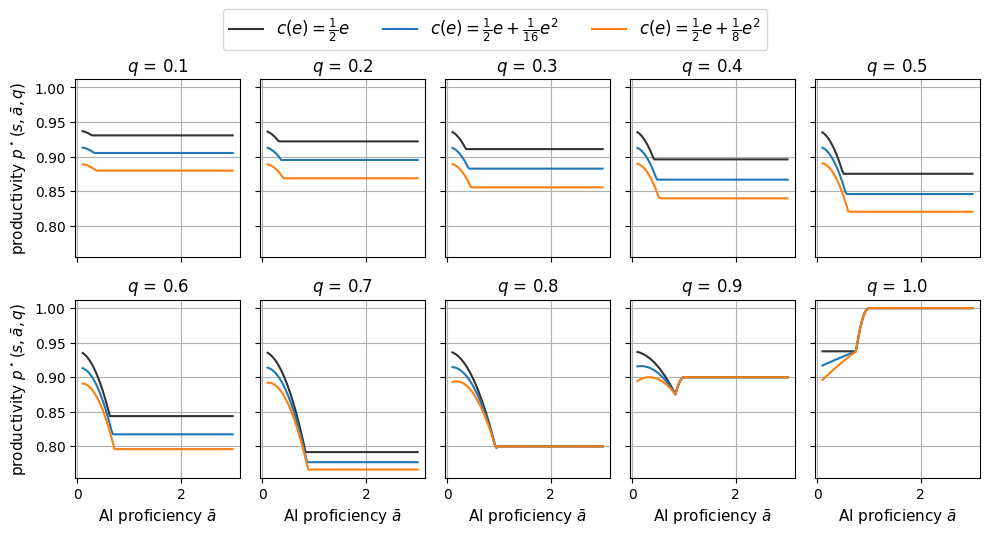}
        \subcaption[]{$p(x)=-x^2+2 x$ for $x\leq 1$ and $p(x)=1$ for $x>1$}
        
    \end{minipage}
    \hfill
    \begin{minipage}{0.75\textwidth}
        \centering
        \includegraphics[width=\linewidth]{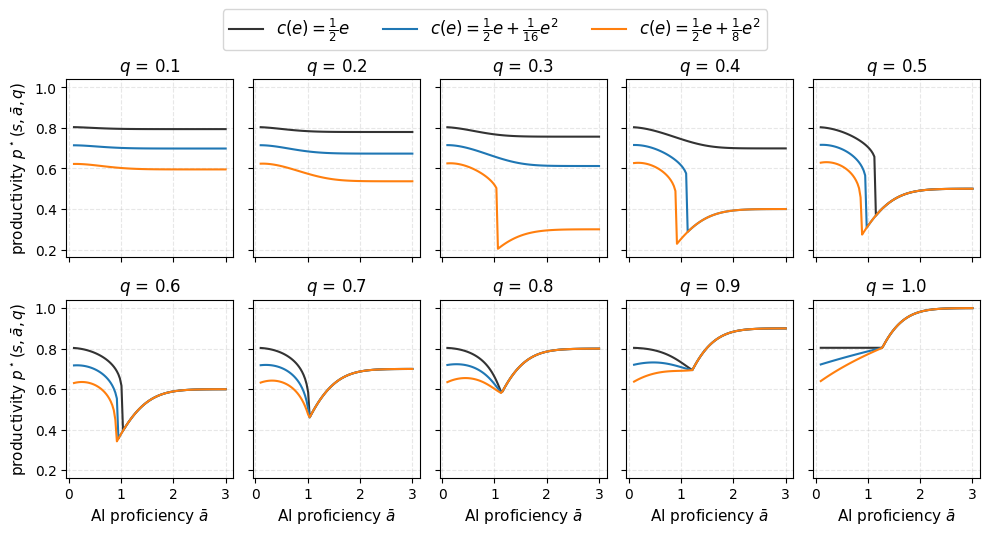}
        \subcaption[]{$p(x)=1-\exp(- x^2)$}
        
    \end{minipage}
    \caption{Productivity paradox under AI unreliability}
    \label{fig:cvx_s4}
\end{figure}

\section{Extensions to non-myopic agents}\label{app:non_myopic}
In this section, we consider agents with non-myopic behavior. For simplicity, we focus on the perfect substitution skill-AI interaction function $h(s,a)=s+a$. The structure of this section is as follows. We first show that the productivity paradox carries over under the mechanism of skill development; notice that modeling non-myopic behavior is not relevant to Mechanism 2 (AI unreliability) since the latter setup is studied in a single-stage model. We then validate our findings with numerical examples.

\subsection{Model of skill development (Section~\ref{sec:bdc})}
We consider a birth-death chain with two skill states $s_1<s_2$. At each initial state $s(0)$, the agent chooses an effort policy $e(s(t))\geq 0$ that maximizes the discounted utility:
\begin{equation*}
    V_\delta(s,a):= \max_{e(\cdot)} \mathbb E\Big[\int_0^{\infty} \exp(-\delta t)\Big(p\big(s(t)+e(s(t))+a\big)-\gamma e(s(t))\Big) dt \Big| s(0)=s \Big],
\end{equation*}
where $\delta>0$ is the discount rate. 
Over a small time horizon $[0,\varepsilon]$, the agent collects utility $\big(p(s+e+a)-\gamma e\big)\varepsilon$, discounts the continuation value by $(1-\delta \varepsilon)$ (Taylor expansion), and leaves the state $s$ for $s'$ with probability $\mathcal{K}_{ss'}(e)\varepsilon+o(\varepsilon)$. Thus, the value function is given by
$$V_\delta(s,a)=\max_{e}\Big(\big(p(s+e+a)-\gamma e\big)\varepsilon+(1-\delta \varepsilon)\big(V_\delta(s,a)+\varepsilon\sum_{s'\neq s}\mathcal{K}_{ss'}(e)(V_\delta(s',a)-V_\delta(s,a))\big)\Big)+o(\varepsilon).$$
Subtracting $(1-\delta \varepsilon)V_\delta(s,a)$ from both sides, dividing by $\varepsilon$, and letting $\varepsilon\to 0$ gives the following Hamilton-Jacobi-Bellman (HJB) equations:
\begin{align}
    \delta V_\delta(s_1,a) &= \max_{e\geq 0}\Big(p(s_1+e+a)-\gamma e + \lambda(e)\big(V_\delta(s_2,a)-V_\delta(s_1,a)\big)\Big), \label{eq:value_s1} \\
    \delta V_\delta(s_2,a) &= \max_{e\geq 0}\Big(p(s_2+e+a)-\gamma e\Big) + \mu \big(V_\delta(s_1,a)-V_\delta(s_2,a)\big),\label{eq:value_s2}
\end{align}
where $\lambda(e)$ and $\mu$ are upward and downward transition rates. In this section, we further assume that $\lambda$ is an increasing, concave, and twice-differentiable function.

Define the value gap function as $\rho_\delta(a):= V_\delta(s_2,a)-V_\delta(s_1,a)$. \eqref{eq:value_s1} and \eqref{eq:value_s2} imply that: at the high state, the agent behaves myopically,
$$e^\star_\delta(s_2,a)=e^\star(s_2,a),$$
while at the low state, the agent optimizes the myopic objective adjusted by the value gap:
$$e^\star_\delta(s_1,a)=\arg\max_{e\geq 0}\big(p(s_1+e+a)-\gamma e+\lambda(e)\rho_\delta(a)\big).$$
Under the forward-looking effort, the steady-state productivity is given by
$$\mathcal{P}_\delta(a)= \dfrac{\mu}{\mu+\lambda(e^\star_\delta(s_1,a))}\cdot p(s_1+e^\star_\delta(s_1,a)+a) + \dfrac{\lambda(e^\star_\delta(s_1,a))}{\mu+\lambda(e^\star_\delta(s_1,a))}\cdot p\big(s_2+e^\star(s_2,a)+a\big).$$

We identify a condition for the curvature of the production function, under which the productivity paradox persists. Before stating the main results, we provide several observations on the value gap $\rho_\delta(a)$.
Define\footnote{We further assume that $\limsup_{x\to\infty} \frac{p(x)+\rho\lambda(x)}{x}<\gamma$, which ensures that the expressions are well-defined.}
\begin{align*}
& u^\star(s,a):=\max_{e\geq 0}\big(p(s+e+a)-\gamma e\big), \quad
\varphi(\rho;a):= \max_{e\geq 0}\big(p(s_1+e+a)-\gamma e+\lambda(e)\rho\big).
\end{align*}
Subtracting Equation~\eqref{eq:value_s1} from Equation~\eqref{eq:value_s2} yields
$$\delta \rho_\delta(a)=\delta V_\delta(s_2,a)-\delta V_\delta(s_1,a)=u^\star(s_2,a)-\mu \rho_\delta(a) - \varphi(\rho_\delta(a);a).$$
\begin{remark}
To get intuition on the difference with the myopic case, consider
$\delta>0$ and $\mu>0$. 
\begin{enumerate}
\item $\rho_\delta(a)$ is the unique solution to $ (\delta+\mu)\rho+\varphi(\rho;a)=u^\star(s_2,a)$ and $0\leq \rho_\delta(a)\leq \frac{u^\star(s_2,a)-u^\star(s_1,a)}{\delta+\mu}$. 
\item The value gap $\rho_\delta(a)$ and the forward-looking effort $e^\star_\delta(s_1,a)\geq e^\star(s_1,a)$ are decreasing in $\delta$. 
\item When $\delta\to\infty$, the value gap $\rho_\delta(a)\to 0$ and thus the forward-looking effort goes to the myopic effort, i.e., $e^\star_\delta(s_1,a)\to e^\star(s_1,a)$.
\end{enumerate}
\end{remark}

While the forward-looking agent exerts greater effort at the low state than the myopic alternative, effort still decreases in AI assistance, as stated next.
\begin{lemma}\label{lem:nonmyopic_effort_decrease}
    Fix $\delta>0$ and $\mu>0$. For AI assistance level $a$ with $e^\star_\delta(s_1,a)<s_2-s_1$, the forward-looking effort  $e^\star_\delta(s_1,a)$ is decreasing in $a$.
\end{lemma}
\begin{proof}[Proof of Lemma~\ref{lem:nonmyopic_effort_decrease}]
Let $F(a,\rho):=(\delta+\mu)\rho+\varphi(\rho;a)-u^\star(s_2,a)$, so $F(a,\rho_\delta(a))=0$. The envelope theorem applied to $\varphi$ gives $\frac{\partial F}{\partial \rho}=(\delta+\mu)+\lambda(e^\star_\delta(s_1,a))>0$ and $\frac{\partial F}{\partial a}=p'(s_1+e^\star_\delta(s_1,a)+a)-p'(s_2+e^\star(s_2,a)+a)\geq p'(s_1+e^\star_\delta(s_1,a)+a)-p'(s_2+a)>0$, where the expression is positive because $e^\star_\delta(s_1,a)<s_2-s_1$. Thus, by the implicit function theorem:
\begin{equation}\label{eq:implicit_function}
    \rho_\delta'(a)=-\frac{\partial F/\partial a}{\partial F/\partial \rho}=-\frac{p'(s_1+e^\star_\delta(s_1,a)+a)-p'(s_2+e^\star(s_2,a)+a)}{(\delta+\mu)+\lambda(e^\star_\delta(s_1,a))}<0,
\end{equation}
which implies that $\rho_\delta(a)$ is decreasing in $a$. 

Recall that $e^\star_\delta(s_1,a)$ maximizes $p(s_1+e+a)-\gamma e+\lambda(e)\rho_\delta(a)$. The first order condition implies that $$p'(s_1+e^\star_\delta(s_1,a)+a)+\lambda'(e^\star_\delta(s_1,a))\rho_\delta(a)=\gamma.$$
As $p$ and $\lambda$ are concave, $p'$ and $\lambda'$ are decreasing. Combining with the fact that $\rho_\delta(a)$ is decreasing in $a$, we obtain that $e^\star_\delta(s_1,a)$ is decreasing in $a$.
\end{proof}

We then turn to the productivity paradox. The analysis of the myopic agent in Proposition~\ref{prop:bdc_prod_n=2} is based on the semi-elasticity gap $\Delta_1$ of Definition~\ref{def:semi-elasticity_general}. Forward-looking behavior adds a component: effort is not only substituted by AI assistance, but also affected by the agent's anticipation of the future (through the value gap $\rho_\delta$). To account for this additional factor, we introduce an additional notion:
$$\Omega_1:=\dfrac{\gamma}{-p''(x^\star)}\Big(\dfrac{\lambda'(0)}{\lambda(0)}\cdot\dfrac{p'(x^\star+s_2-s_1)-\gamma}{p(x^\star+s_2-s_1)-p(x^\star)}-\dfrac{\lambda''(0)}{\lambda(0)}\Big).$$
Before presenting the full statement, we define the adjacent intervals analogously to those in Section~\ref{sec:bdc_result}. Let $a^\delta_1 := \inf\{a\geq 0: e^\star_\delta(s_1,a)=0 \}$ be the point where the low skill agent exerts no effort,\footnote{When $\mu$ is sufficiently large, $\rho_\delta(a)$ is close to zero, and $a_1^\delta$ is finite.} or equivalently, $$p'\big(s_1+a^\delta_1\big)+\lambda'(0)\rho_\delta\big(a^\delta_1\big) = \gamma.$$ 
Define
${\cal I}^\delta_2:= \big(0,(x^\star-s_2)^+\big]$, ${\cal I}^\delta_1:= \big((x^\star-s_2)^+,a^\delta_1\big]$, and ${\cal I}^\delta_0:= (a^\delta_1,\infty)$.
The next result shows how the sign comparison between $\Delta_1$ and $\Omega_1$ leads to the productivity paradox, especially near the right endpoint of ${\cal I}^\delta_1$. 
\begin{proposition}\label{prop:nonmyopic_paradox}
    Consider any concave production function $p(x)$, linear cost function $c(e)=\gamma e$, concave transition function $\lambda(e)$, and two states $s_1<s_2$. Fix any discount rate $\delta>0$, there exists a threshold $\bar{\mu}_\delta>0$ such that, for every $\mu>\bar{\mu}_\delta$,
    when $\Delta_1>\Omega_1$,
    there exists $\tau<a^\delta_1$ such that the steady-state productivity $\mathcal{P}_\delta(a)$ is decreasing on $(\tau, a^\delta_1]$.
\end{proposition}
\begin{proof}[Proof of Proposition~\ref{prop:nonmyopic_paradox}]
    For notational simplicity, we write $e_1:=e^\star_\delta(s_1,a)$ and $e_1':=\frac{d}{da}e^\star_\delta(s_1,a)$. The steady-state productivity can be written as
    $$\mathcal{P}_\delta(a)= \dfrac{\mu}{\mu+\lambda(e_1)}\cdot p(s_1+e_1+a) + \dfrac{\lambda(e_1)}{\mu+\lambda(e_1)}\cdot p\big(s_2+e^\star(s_2,a)+a\big).$$

    We start from the region $a\in {\cal I}^\delta_0$. By the definition of $a^\delta_1$ and the effort expression in Proposition~\ref{prop:effort_productivity_expressions}, the effort is zero at both states. Thus,
    $$\mathcal{P}_\delta(a)= \dfrac{\mu p(s_1+a) + \lambda(0)p(s_2+a)}{\mu+\lambda(0)},$$
    which is increasing in $a$ since $p$ is increasing. 
    
    We then focus on the region $a\in {\cal I}^\delta_1$, given ${\cal I}^\delta_1$ is non-empty, on which $e_1>0$ while $e^\star(s_2,a)=0$. The steady-state productivity is
    \begin{align*}
        \mathcal{P}_\delta(a)= \dfrac{\mu}{\mu+\lambda(e_1)}\cdot p(s_1+e_1+a) + \dfrac{\lambda(e_1)}{\mu+\lambda(e_1)}\cdot p(s_2+a).
    \end{align*}
    Differentiating with respect to $a$ gives
    \begin{align}\label{eq:delta_derivative}
        \mathcal{P}_\delta'(a)&=\dfrac{\mu \lambda'(e_1)e_1'}{\big(\mu+\lambda(e_1)\big)^2}\big(p(s_2+a)-p(s_1+e_1+a)\big)+\dfrac{\mu}{\mu+\lambda(e_1)}p'(s_1+e_1+a)\big(1+e_1'\big)+\dfrac{\lambda(e_1)}{\mu+\lambda(e_1)}p'(s_2+a)\nonumber\\
        &=\dfrac{\big(\mu+\lambda(e_1)\big)\Big(\mu p'(s_1+e_1+a)\big(1+e_1'\big)+\lambda(e_1)p'(s_2+a)\Big)+  \mu\lambda'(e_1)e_1' \big(p(s_2+a)-p(s_1+e_1+a)\big) }{\big(\mu+\lambda(e_1)\big)^2}.
    \end{align}
    Notice that in the myopic case, we have $e_1'=-1$, and $s_1+e_1+a=x^\star$ is constant, which implies the term $\mu p'(s_1+e_1+a)(1+e_1')$ vanishes. Unlike the myopic case, the numerator of $\mathcal{P}_\delta'(a)$ may not be monotone in $a$. 
    
    We next express the effort derivative $e_1'$. Differentiating the first-order condition $p'(s_1+e_1+a)-\gamma+\lambda'(e_1)\rho_\delta(a)=0$ with respect to $a$, we obtain
    \begin{align}\label{eq:nonmyopic_effort_derivative}
        1+e_1'= \dfrac{\lambda'(e_1)\rho_\delta'(a)-\lambda''(e_1)\rho_\delta(a)}{-\big(p''(s_1+e_1+a)+\lambda''(e_1)\rho_\delta(a)\big)},
    \end{align}
    where in \eqref{eq:implicit_function} of Lemma~\ref{lem:nonmyopic_effort_decrease} we have shown
    \begin{equation}\label{eq:rho_derivative}
        \rho_\delta'(a)=\frac{p'(s_2+a)-p'(s_1+e_1+a)}{\delta+\mu+\lambda(e_1)}.
    \end{equation}
    Moreover, by Taylor expansion, we have $\varphi(\rho;a)=\varphi(0;a)+\lambda(e_1)\rho+O(\rho^2)$. Recall that $(\delta+\mu)\rho_\delta(a)+\varphi(\rho_\delta(a);a)-u^\star(s_2,a)=0$, which implies that
    \begin{equation}\label{eq:rho_taylor}
        \rho_\delta(a)=\frac{u^\star(s_2,a)-u^\star(s_1,a)}{\delta+\mu+\lambda(e_1)}+O(\mu^{-2}).
    \end{equation}
    Note that the above expressions hold for any $a\in{\cal I}^\delta_1$.

    We now fix $a=a^\delta_1$, where $e_1=0$. Evaluating \eqref{eq:delta_derivative} at the right endpoint of ${\cal I}^\delta_1$ gives
    \begin{align*}
        \mathcal{P}_{\delta,-}'(a^\delta_1)=\dfrac{\big(\mu+\lambda(0)\big)\Big(\mu p'(s_1+a^\delta_1)\big(1+e_1'\big)+\lambda(0)p'(s_2+a^\delta_1)\Big) +  \mu\lambda'(0)e_1' \big(p(s_2+a^\delta_1)-p(s_1+a^\delta_1)\big) }{\big(\mu+\lambda(0)\big)^2}.
    \end{align*}
    Substituting $e_1=0$ into \eqref{eq:nonmyopic_effort_derivative}\eqref{eq:rho_derivative}\eqref{eq:rho_taylor}, we obtain that
    \begin{align*}
        & 1+e_1'= \dfrac{\lambda'(0)\rho_\delta'(a^\delta_1)-\lambda''(0)\rho_\delta(a^\delta_1)}{-\big(p''(s_1+a^\delta_1)+\lambda''(0)\rho_\delta(a^\delta_1)\big)},\\
        &\rho_\delta'(a^\delta_1)=\frac{p'(s_2+a^\delta_1)-p'(s_1+a^\delta_1)}{\delta+\mu+\lambda(0)},\\
        & \rho_\delta(a^\delta_1)=\frac{u^\star(s_2,a^\delta_1)-u^\star(s_1,a^\delta_1)}{\delta+\mu+\lambda(0)}+O(\mu^{-2}).
    \end{align*}
    
    We focus on the case of sufficiently large $\mu$. Note that $a^\delta_1$ depends on $\mu$ through its dependence on $\rho_\delta$, and all limits below are taken on $a=a^\delta_1$. When $\mu\to\infty$, we have $\frac{\mu}{\delta+\mu+\lambda(e_1)}\to 1$, $\rho_\delta(a^\delta_1)\to 0$ and $\rho_\delta'(a^\delta_1)\to 0$. 
    Recall that $p'\big(s_1+a^\delta_1\big)+\lambda'(0)\rho_\delta\big(a^\delta_1\big) = \gamma$ by the definition of $a^\delta_1$, this implies that
    $$ p'(s_1+a^\delta_1)\to \gamma, \quad a^\delta_1\to x^\star-s_1.$$ 
    Since $u^\star(s_2,a^\delta_1)=p(s_2+a^\delta_1)$ and $u^\star(s_1,a^\delta_1)=p(x^\star)-\gamma(x^\star-s_1-a^\delta_1)$, we have
    $$u^\star(s_2,a^\delta_1)-u^\star(s_1,a^\delta_1)\to p(x^\star+s_2-s_1)-p(x^\star).$$
    Thus,
    \begin{align*}
        e_1'\to -1 \quad \text{and} \quad \mu\big(1+e_1'\big)\to
        K:=\dfrac{\lambda'(0) \big(p'(x^\star+s_2-s_1)-\gamma\big)-\lambda''(0)\big(p(x^\star+s_2-s_1)-p(x^\star)\big)}{-p''(x^\star)}.
    \end{align*}

    Substituting it into the derivative $\mathcal{P}_{\delta,-}'(a^\delta_1)$, the numerator becomes 
    \begin{align*}
        \mu \Big(\gamma K + \lambda(0) p'(x^\star+s_2-s_1) - \lambda'(0)\big(p(x^\star+s_2-s_1)-p(x^\star)\big)\Big) + \lambda(0) \Big(\gamma K + \lambda(0)p'(x^\star+s_2-s_1)\Big).
    \end{align*}
    Note that this term of order $\mu$ determines the sign of $\mathcal{P}_{\delta,-}'(a^\delta_1)$ for $\mu$ sufficiently large.
    The derivative at the right endpoint is therefore negative if and only if
    \begin{align*}
        \gamma K + \lambda(0) p'(x^\star+s_2-s_1) < \lambda'(0)\big(p(x^\star+s_2-s_1)-p(x^\star)\big),
    \end{align*}
    or equivalently,
    \begin{align*}
        \dfrac{\gamma K}{\lambda(0)\big(p(x^\star+s_2-s_1)-p(x^\star)\big)} < \dfrac{\lambda'(0)}{\lambda(0)}-\dfrac{p'(x^\star+s_2-s_1)}{p(x^\star+s_2-s_1)-p(x^\star)}.
    \end{align*}
    Substituting the expression of $K$, it becomes
    \begin{align*}
        \dfrac{\gamma}{-p''(x^\star)}\Big(\dfrac{\lambda'(0)}{\lambda(0)}\cdot\dfrac{p'(x^\star+s_2-s_1)-\gamma}{p(x^\star+s_2-s_1)-p(x^\star)}-\dfrac{\lambda''(0)}{\lambda(0)}\Big) < \dfrac{\lambda'(0)}{\lambda(0)}-\dfrac{p'(x^\star+s_2-s_1)}{p(x^\star+s_2-s_1)-p(x^\star)},
    \end{align*}
    which leads to $\Omega_1<\Delta_1$. 
    Thus, there exists $\bar{\mu}_\delta$ such that, for $\mu>\bar{\mu}_\delta$, we have: when $\Delta_1>\Omega_1$, we have $\mathcal{P}_{\delta,-}'(a^\delta_1)<0$. Thus, there exists $\tau<a^\delta_1$ such that $\mathcal{P}_\delta'(a)<0$ on $(\tau,a^\delta_1]$ and thus the productivity $\mathcal{P}_\delta(a)$ is decreasing on $(\tau,a^\delta_1]$. 
\end{proof}

\subsection{Numerics}
We conduct numerical experiments across multiple parametric specifications of production and cost functions to study the productivity paradoxes for non-myopic agents. For robustness checks, we include myopic, forward-looking (when $\delta>0$), and fully rational (when $\delta\to 0$) agents. For the case with $\delta\to 0$, we calculate the effort by optimizing the steady-state utility with two effort variables jointly:
$$ e_1^\star, e_2^\star = \arg\max_{e_1,e_2\geq 0}\Big(\frac{\mu}{\mu+\lambda(e_1)}(p(s_1+e_1+a)-\gamma e_1) + \frac{\lambda(e_1)}{\mu+\lambda(e_1)}(p(s_2+e_2+a)-\gamma e_2)\Big). $$

We numerically illustrate the model of skill development using a skill set with $N=2$ states, where $s_1=0, s_2=1$. We consider an upward transition function $\lambda(e)=0.01+e$ and a decay rate $\mu=1$. 
We study several concave production functions, including fractional, exponential, and power law functions, and a linear cost function: 
$$p(x)=\frac{x}{1+x},\, p(x)=1-\exp(-x), p(x)=\frac{1}{2}x^{1/2},\, \text{and } \, c(e)=0.3 e.$$

Figure~\ref{fig:forward_s3} shows the steady-state productivity $\mathcal{P}(a)$ as a function of AI assistance level $a$. The productivity paradox persists for non-myopic agents, although the productivity decline is slightly attenuated.

\begin{figure}[ht]
    \centering
    \begin{minipage}{0.32\textwidth}
        \centering
        \includegraphics[width=\linewidth]{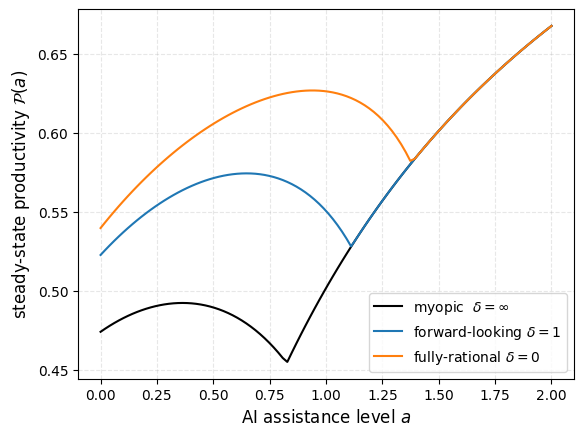}
        \subcaption[]{$p(x)=\frac{x}{1+x}$}
    \end{minipage}
    \hfill
    \begin{minipage}{0.32\textwidth}
        \centering
        \includegraphics[width=\linewidth]{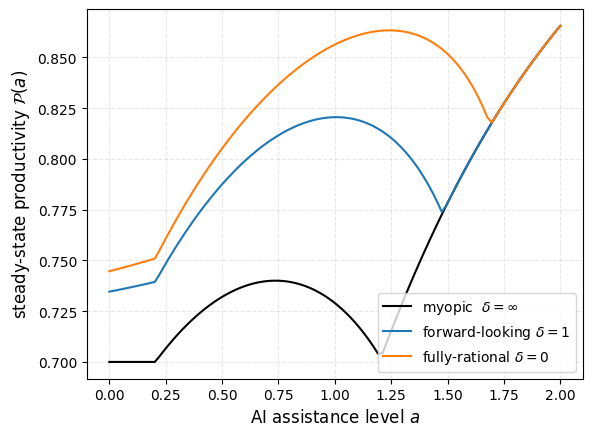}
        \subcaption[]{$p(x)=1-\exp{(-x)}$}
    \end{minipage}
    \hfill
    \begin{minipage}{0.32\textwidth}
        \centering
        \includegraphics[width=\linewidth]{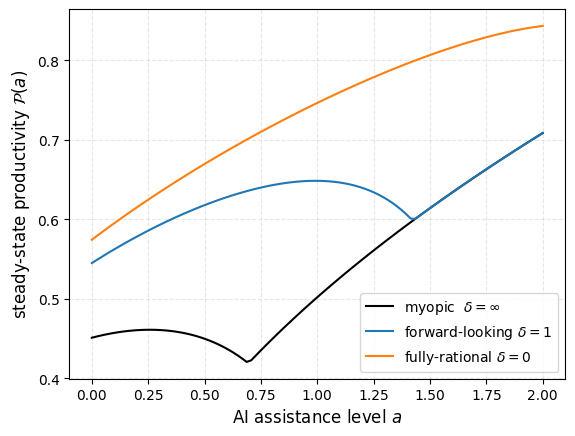}
        \subcaption[]{$p(x)=\frac{1}{2}x^{1/2}$}
    \end{minipage}
    \caption{Productivity paradox for forward-looking agents}
    \label{fig:forward_s3}
\end{figure}

\section{Extension to impact on welfare}\label{app:welfare}
In this section, we consider AI's impact on total welfare. For simplicity, we focus on the perfect substitution skill-AI interaction function $h(s,a)=s+a$. We examine whether AI assistance can negatively affect welfare under the skill development mechanism. The underlying intuition is again a distributional effect toward the low‑skill state, analogous to Simpson's paradox. It is worth noting that under our model of AI unreliability, we do not observe the same pattern, as welfare increases in AI proficiency level for a single skill state even when AI is unreliable.

For the two-skill state space, we show an equivalent result to Proposition~\ref{prop:bdc_prod_n=2} for welfare. That is, under a semi-elasticity gap condition, a reduction of welfare occurs as AI assistance level improves.

Recall that $u^\star(s,a):=\max_{e\geq0}\big(p(s+e+a)-\gamma e\big)=p^\star(s,a)-\gamma e^\star(s,a)$. We define the steady-state welfare as:
\begin{equation*}
    \mathcal{U}(a):= \pi_1(a)u^\star(s_1,a) + \pi_2(a)u^\star(s_2,a).
\end{equation*}
Notice that utility is productivity minus cost of effort. To reflect this, we introduce a modified notion of semi-elasticity gap.
\begin{definition}\label{def:utility_gap_notion}
The adjusted semi-elasticity gap is defined as
$$\Delta_1^{\mathcal{U}} := \frac{\lambda'(0)}{\lambda(0)} - \frac{p'(x^\star+s_2-s_1)+\gamma}{p(x^\star+s_2-s_1)-p(x^\star)}.$$
\end{definition}
As stated next, welfare can decline with AI assistance level when the adjusted semi-elasticity is positive and the following condition holds. 
\begin{condition}\label{cond:utility_paradox}
$\Delta_1^{\mathcal{U}}> \frac{\sqrt{4\gamma p'(x^\star+s_2-s_1)}}{p(x^\star+s_2-s_1)-p(x^\star)}$, and $\mu\in (\mu_2^-,\mu_2^+)$, where
$$\mu_2^{\pm}:=\dfrac{\lambda(0)\big(p(x^\star+s_2-s_1)-p(x^\star)\big)\Delta_1^{\mathcal{U}} \pm  \lambda(0) \sqrt{\Big(\big(p(x^\star+s_2-s_1)-p(x^\star)\big)\Delta_1^{\mathcal{U}}\Big)^2-4\gamma p'(x^\star+s_2-s_1)}}{2\gamma}.$$
\end{condition}

\begin{proposition}\label{prop:welfare_paradox}
     Consider any concave production function $p(x)$, linear cost function $c(e)=\gamma e$, concave transition function $\lambda(e)$, and two states $s_1<s_2$.
     \begin{enumerate}
        \item If Condition~\ref{cond:utility_paradox} holds, utility $\mathcal{U}(a)$ is  decreasing at the right endpoint of ${\cal I}_1$; and (weakly) increasing on ${\cal I}_0$ and ${\cal I}_2$.
        \item Otherwise, utility $\mathcal{U}(a)$ is (weakly) increasing on $[0,\infty)$.
    \end{enumerate}
\end{proposition}
\begin{proof}[Proof of Proposition~\ref{prop:welfare_paradox}]
    The steady-state utility can be written as
    $$\mathcal{U}(a)= \dfrac{\mu}{\mu+\lambda(e^\star(s_1,a))}\cdot u^\star(s_1,a) + \dfrac{\lambda(e^\star(s_1,a))}{\mu+\lambda(e^\star(s_1,a))}\cdot u^\star(s_2,a).$$
    We start from considering the intervals ${\cal I}_0$ and ${\cal I}_2$, where we want to show that the utility (weakly) increases.
    \begin{itemize}
    \item If $a \in {\cal I}_0$, the utility-maximizing effort is zero at each state and thus the utility $\mathcal{U}(a)$ equals the productivity $\mathcal{P}(a)$, which is increasing in $a$.
    \item If $a \in {\cal I}_2$, the utility at each state is  $u^\star(s_k,a)=p(x^\star)-\gamma(x^\star-s_k-a)$. Then,
    $$\mathcal{U}(a)= p(x^\star)-\gamma(x^\star-s_1-a)+\gamma (s_2-s_1)\dfrac{\lambda(x^\star-s_1-a)}{\mu+\lambda(x^\star-s_1-a)}, \quad \mathcal{U}'(a)=\gamma\Big[1-(s_2-s_1)\dfrac{\mu\lambda'(x^\star-s_1-a)}{\big(\mu+\lambda(x^\star-s_1-a)\big)^2}\Big].$$
    The concavity of $\lambda$ implies that $\lambda(x^\star-s_1-a)\geq (x^\star-s_1-a)\lambda'(x^\star-s_1-a)$. Since $x^\star-s_1-a\geq s_2-s_1$ on ${\cal I}_2$, we have $\lambda(x^\star-s_1-a)\geq (s_2-s_1)\lambda'(x^\star-s_1-a)$. Thus, $\mathcal{U}'(a)\geq \gamma\big[1-\frac{\mu\lambda(x^\star-s_1-a)}{(\mu+\lambda(x^\star-s_1-a))^2}\big]\geq \frac{3}{4}\gamma>0$. It follows that $U(a)$ increases for $a \in {\cal I}_2$.
    \end{itemize}
    We then focus on intermediate cases where $a\in{\cal I}_1$, given ${\cal I}_1$ is non-empty. The steady-state utility is
    \begin{align*}
        \mathcal{U}(a)= \dfrac{\mu}{\mu+\lambda(x^\star-s_1-a)}\cdot \big(p(x^\star)-\gamma (x^\star-s_1-a)\big) + \dfrac{\lambda(x^\star-s_1-a)}{\mu+\lambda(x^\star-s_1-a)}\cdot p(s_2+a).
    \end{align*}
    Differentiating with respect to $a$ gives
    \begin{align*}
        \mathcal{U}'(a)&=\dfrac{\mu \lambda'(x^\star-s_1-a)}{\big(\mu+\lambda(x^\star-s_1-a)\big)^2}\Big(\big(p(x^\star)-\gamma (x^\star-s_1-a)\big)-p(s_2+a)\Big)\\
        &+\dfrac{\mu}{\mu+\lambda(x^\star-s_1-a)}\gamma+\dfrac{\lambda(x^\star-s_1-a)}{\mu+\lambda(x^\star-s_1-a)}p'(s_2+a)\\
        &=\dfrac{\big(\mu+\lambda(x^\star-s_1-a)\big)\big(\mu\gamma+\lambda(x^\star-s_1-a)p'(s_2+a)\big)-\mu\lambda'(x^\star-s_1-a)\big(p(s_2+a)-p(x^\star)+\gamma (x^\star-s_1-a)\big)}{\big(\mu+\lambda(x^\star-s_1-a)\big)^2}.
    \end{align*}
    Denote the numerator of $\mathcal{U}'(a)$ as $F(a)$.
    \begin{align*}
    F'(a) = &-\lambda'(x^\star-s_1-a)\lambda(x^\star-s_1-a)p'(s_2+a) - \big(\mu + \lambda(x^\star-s_1-a)\big)\lambda'(x^\star-s_1-a)p'(s_2+a) \\
    &+ \big(\mu + \lambda(x^\star-s_1-a)\big)\lambda(x^\star-s_1-a)p''(s_2+a) + \mu\lambda''(x^\star-s_1-a)\Big(p(s_2+a) - p(x^\star) + \gamma(x^\star-s_1-a)\Big) \\
    &- \mu\lambda'(x^\star-s_1-a)p'(s_2+a).
    \end{align*}
    The concavity of $p$ and $\lambda$ implies that  $F'(a)\leq 0$. Thus, the numerator of the derivative $\mathcal{U}(a)$ is a decreasing function in $a$.
    Evaluating the derivative at the endpoints, we have
    \begin{align*}
        \mathcal{U}'_{+}(x^\star-s_2) &=\dfrac{\gamma\big(\mu+\lambda(s_2-s_1)\big)^2-\mu\gamma\lambda'(s_2-s_1)(s_2-s_1)}{\big(\mu+\lambda(s_2-s_1)\big)^2}>0,\\
        \mathcal{U}'_{-}(x^\star -s_1) &=\dfrac{\big(\mu+\lambda(0)\big) \big(\mu\gamma+\lambda(0)p'(x^\star +s_2-s_1) \big) - \mu\lambda'(0) \big(p(x^\star +s_2-s_1)-p(x^\star)\big) }{\big(\mu+\lambda(0)\big)^2} \\
        &=\dfrac{\gamma\mu^2-\mu\Big(\lambda'(0)\big(p(x^\star +s_2-s_1)-p(x^\star)\big)-\lambda(0)p'(x^\star +s_2-s_1)-\gamma\lambda(0)\Big)+\lambda^2(0)p'(x^\star +s_2-s_1)}{\big(\mu+\lambda(0)\big)^2},
    \end{align*}
    where the first inequality follows from the concavity of the function $\lambda$ again. The derivative at the left endpoint is strictly positive. The derivative at the right endpoint is negative if and only if
    \begin{align*}
        \gamma\mu^2-\mu\Big(\lambda'(0)\big(p(x^\star +s_2-s_1)-p(x^\star)\big)-\lambda(0)p'(x^\star +s_2-s_1)-\gamma\lambda(0)\Big)+\lambda^2(0)p'(x^\star +s_2-s_1)<0,
    \end{align*}
    There exists $\mu$ for the above inequality to hold if and only if 
    \begin{align}
        &\lambda'(0)\big(p(x^\star+s_2-s_1)-p(x^\star)\big)-\lambda(0)p'(x^\star+s_2-s_1)-\gamma\lambda(0)>0, \label{eq:utility_1}\\
        &\Big(\lambda'(0)\big(p(x^\star+s_2-s_1)-p(x^\star)\big)-\lambda(0)p'(x^\star+s_2-s_1)-\gamma\lambda(0)\Big)^2>4\gamma\lambda^2(0)p'(x^\star+s_2-s_1),\label{eq:utility_2}
    \end{align}
    which leads to the region $\mu\in (\mu_2^-,\mu_2^+)$.
    Dividing the \eqref{eq:utility_1} by $\lambda(0)(p(x^\star+s_2-s_1)-p(x^\star))>0$ leads to $\Delta_1^{\mathcal{U}}>0$. Dividing \eqref{eq:utility_2} by $\lambda^2(0)$ and substituting $\Delta_1^{\mathcal{U}}$ lead to $\Big(\big(p(x^\star+s_2-s_1)-p(x^\star)\big)\Delta_1^{\mathcal{U}}\Big)^2 > 4\gamma p'(x^\star+s_2-s_1)$. When these conditions hold, which is exactly Condition~\ref{cond:utility_paradox}, there exists a unique $\tau$ such that $\mathcal{U}'(\tau)=0$, and the steady-state welfare $\mathcal{U}(a)$ is increasing on $(x^\star-s_2,\tau]$ and decreasing on $(\tau, x^\star-s_1]$. Otherwise, the steady-state welfare $\mathcal{U}(a)$ is always increasing.
\end{proof}

\end{document}